\documentclass[12pt]{amsart}
\usepackage{color}
\usepackage{graphicx}
\usepackage{graphicx}
\usepackage{epsfig}
\usepackage{eucal}
\usepackage{tikz}
\usetikzlibrary{arrows}
\newtheorem{thm}{Theorem}[section]

\newtheorem{defn}[thm]{Definition}
\newtheorem{lemma}[thm]{Lemma}
\newtheorem{conj}[thm]{Conjecture}
\newtheorem{cor}[thm]{Corollary}
\newtheorem{remark}[thm]{Remark}
\newtheorem{example}[thm]{Example}

\usepackage{amsmath}
\usepackage{amsxtra}
\usepackage{amscd}
\usepackage{amsthm}
\usepackage{amsfonts}
\usepackage{amssymb}
\usepackage{eucal}
\newcommand{\bmb}{\left( \begin{array}{rr}}
\newcommand{\enm}{\end{array}\right)}

\newcommand{\cF}{\mathcal F}
\newcommand{\F}{\mathcal F}

\newcommand{\gl}{{\mathfrak{gl}}}

\newcommand{\half}{\frac12}

\renewcommand{\S}{\mathcal S}
\renewcommand{\sl}{{\mathfrak{sl}}}

\newcommand{\C}{{\mathbb C}}

\newcommand{\Z}{{\mathbb Z}}

\newcommand{\N}{{\mathbb N}}

\newcommand{\bx}{{\mathbf x}}

\newcommand{\bu}{{\mathbf u}}
\newcommand{\bv}{{\mathbf v}}

\newcommand{\al}{{\alpha}}

\newcommand{\e}{\mathfrak{e}}
\newcommand{\f}{\mathfrak{f}}
\numberwithin{equation}{section}

\begin{document}

\title[$(t,q)$ Q-systems, DAHA and quantum toroidal algebras]{(${\mathbf  t},{\mathbf q}$)-deformed Q-systems, DAHA and quantum toroidal algebras via generalized Macdonald operators}
\author{Philippe Di Francesco} 
\address{PDF: Department of Mathematics, University of Illinois MC-382, Urbana, IL 61821, U.S.A. e-mail: philippe@illinois.edu}
\author{Rinat Kedem}
\address{RK: Department of Mathematics, University of Illinois MC-382, Urbana, IL 61821, U.S.A. e-mail: rinat@illinois.edu}
\date{\today}
\begin{abstract}
We introduce difference operators on the space of symmetric functions which are a natural generalization of the $(q,t)$-Macdonald operators. In the $t\to\infty$ limit, they satisfy the $A_{N-1}$ quantum $Q$-system \cite{qKR,DFK15}. We identify the elements in the spherical $A_{N-1}$ DAHA \cite{Cheredbook} which are represented by these operators, as well as within the quantum toroidal algebra of $\gl_1$ \cite{FJMM} and the elliptic Hall algebra \cite{SHIVAS,BS,SCHIFFEHA}.
%On the other hand, our generalized Macdonald operators have a simple expression in terms of the functional representation of the 
%$(q,t)$ $A_r$ Double-Affine Hecke Algebra, by making explicit use of its $SL(2,\Z)$ symmetry. 
%Both deformations coincide, and further relations to the Elliptic Hall algebra are investigated. 
We present a plethystic, or bosonic, formulation of the generating functions for the generalized Macdonald operators, which we relate to recent work of Bergeron et al \cite{BGLX}. Finally we derive constant term identities for the current that allow to interpret them in terms of shuffle products \cite{NegutShuffle}.
In particular we obtain in the $t\to\infty$ limit a shuffle presentation of the quantum $Q$-system relations.
\end{abstract}

\maketitle
\date{\today}
\tableofcontents
\section{Introduction: generalized Macdonald operators}
%intro
%summary of results and main references
\subsection{Introduction}

The aim in this paper is to make explicit the relation between the algebra satisfied by the generators of the $A_{N-1}$ quantum $Q$-system, 
or more precisely, their $t$-deformation \cite{DFK16}, and the following algebras, whose relation to 
each other is better known: the $A_{N-1}$ spherical Double Affine Hecke Algebra (sDAHA) 
\cite{Cheredbook}, the level-$(0,0)$ and level-$(1,0)$ quantum toroidal 
algebras of $\widehat{\mathfrak gl}_1$ \cite{Miki07,FJMM,AFS}, the shuffle algebra \cite{FO,NegutShuffle} and the 
Elliptic Hall Algebra (EHA) \cite{SHIVAS,BS}. 

Quantum $Q$-systems arise naturally as the quantization \cite{BZ} of the cluster algebras \cite{FZ} associated with 
the classical Q-system \cite{Ke07,DFKnoncom}. The latter is a recursion relation for characters of Kirillov-Reshetikhin modules \cite{KR} of quantum affine algebras, and it is directly connected \cite{HKOTY,DFK08} with fermionic formulas for the characters of the tensor products of KR-modules. The quantum $Q$-system is similarly connected with the graded characters of the Feigin-Loktev fusion product \cite{FL} of the same spaces \cite{qKR}, where the quantum parameter $q$
is associated with the homogeneous grading. 

Graded characters of tensor products of KR-modules 
can be written as iterated, ordered products of difference operators, which satisfy the quantum $Q$-system, acting on the polynomial 1. That is, there is a representation of the algebra satisfied by the generators of the quantum $Q$-system in terms of difference operators acting on the space of symmetric functions \cite{qKR,DFK15}.

In Ref. \cite{DFK16}, we gave the relation between this algebra and a non-standard level-$0$ version of the quantum affine algebra of $\widehat{\mathfrak sl}_2$ in the Drinfeld presentation. For finite $N$, it is a rank-dependent quotient of a quantum affine algebra.
The difference operators introduced in \cite{DFK15} were identified as a generalization of the the dual Whittaker limit $t\to \infty$ of the celebrated Macdonald 
difference operators \cite{macdo}. 

This observation naturally lead  \cite{DFK16} to the introduction of a $t$-deformation of these operators. The present paper is a continuation of this work, and its purpose is to give the algebraic framework of this deformation. To do this, we define a new set of difference operators.
%In this context, we start by considering a natural further generalization of the Macdonald difference operators of \cite{DFK16}.

\subsection{Generalized Macdonald operators}

Let $\C_{q,t}:=\C(q,t)$ be the field of rational functions in the formal variables $q$ and $t$.
We denote by $\mathcal F_N$ the space of rational symmetric functions of $N$ variables over $\C_{q,t}$:
\begin{equation}{\mathcal F}_N:=\C_{q,t}(x_1,x_2,...,x_N)^{S_N}.\end{equation} 
It is a subpace of ${\mathcal S}_N$, symmetric functions of $N$ variables over the same field. 
Here, the symmetric group $S_N$ acts naturally on $\{x_1,...,x_N\}$ by permutations. 

%\begin{defn}
Given any rational function $f(x_1,x_2,...,x_N)$, let ${\rm Sym}(f)$ denote its symmetrization, and
% of the function, defined as
%\begin{equation}{\rm Sym}( f(x_1,x_2,...,x_N))
%=\sum_{\sigma\in S_N} f(x_{\sigma(1)},x_{\sigma(2)},...,x_{\sigma(N)})\ .
%\end{equation}
%\end{defn}
%Furthermore, 
let $\Gamma_i$ be the multiplicative shift operator, acting on functions in $\S_N$ as follows:
\begin{equation}\label{shift}
\Gamma_i\, f(x_1,...,x_N)=f(x_1,...,x_{i-1},q x_i,x_{i+1},...,x_N) \qquad (1\leq i\leq N).
\end{equation}
Alternatively, $\Gamma_i=q^{\delta_i}$ where $\delta_i$ is the additive shift operator, 
$\delta_i: \ x_j\mapsto x_j+\delta_{i,j}$.

Consider
%The main definition of this section is that of a 
the following $q$-difference operator, which is a generalization of the Macdonald operator, acting on 
%on symmetric functions in 
$\F_N$:
%and is a simple generalization of the Macdonald operator acting on the same space.
\begin{defn}\label{gmacdef} 
Let $\al\in [1,N]$ and $P\in \F_\al$. Define the associated generalized 
Macdonald operator to be the difference operator ${\mathcal D}_\al(P)$ acting on $\F_N$ as
\begin{equation}\label{symdiffop}
{\mathcal D}_\al(P):=\frac{1}{\al!\,(N-\al)!}{\rm Sym}\left( P(x_1,x_2,...,x_\al) 
\prod_{1\leq i\leq  \al<j \leq N} \frac{tx_i-x_j }{x_i-x_j}\, 
\Gamma_1\Gamma_2\cdots \Gamma_\al\right).\end{equation}
\end{defn}

The following special examples of the operator in Definition \ref{gmacdef}:
% are of interest to us in this paper:
\begin{enumerate}
\item The original Macdonald difference operators \cite{macdo} correspond to $P=1\in \F_\al$:
\begin{equation}\label{macdop}
\widetilde{\mathcal D}_\al:={\mathcal D}_\al(1)= \sum_{I\subset [1,N]\atop |I|=\al } 
\prod_{i\in I\atop j\not\in I} \frac{t x_i -x_j}{x_i-x_j} \, \prod_{i\in I} \Gamma_i .
\end{equation}

\item The generalized Macdonald operators, introduced in \cite{DFK16}, correspond to the special case of $P=(x_1x_2\cdots x_\al)^n\in \F_\al$, $n\in \Z$:
\begin{equation}\label{genmacdop}
{\mathcal M}_{\al;n}:={\mathcal D}_\al\Big( (x_1x_2\cdots x_\al)^n\Big)=
\sum_{I\subset [1,N]\atop |I|=\al } \prod_{i\in I}(x_i)^n\,\prod_{i\in I\atop j\not\in I} 
\frac{t x_i -x_j}{x_i-x_j} \, \prod_{i\in I} \Gamma_i.
\end{equation}
These operators were inspired by the functional representation of the quantum Q-system. Note that ${\mathcal M}_{\al;0}=\widetilde{\mathcal D}_\al$ are the original Macdonald operators.

\item More generally,
%In this paper, we will focus on a more general case of Definition \ref{gmacdef}, obtained 
%
Let $P$ be
a {\em generalized Schur function} $s_{a_1,...,a_\al}\in \F_\al$ with $a_i\in \Z$. This is the natural 
generalization of Schur polynomials to Schur Laurent polynomials:
\begin{defn}
For any $(a_1,a_2,...,a_\al) \in \Z^\al$, define the generalized Schur function:
\begin{equation}\label{defschur}
s_{a_1,...,a_\al}(x_1,x_2,...,x_\al)
:=\frac{\det_{1\leq i,j\leq \al}\left( x_i^{a_j+\al-j}\right)}{\prod_{1\leq i<j\leq \al} x_i-x_j}={\rm Sym}\left(
\frac{\prod_{i=1}^\al x_i^{a_i+\al-i}}{\prod_{i<j} x_i-x_j} \right) .
\end{equation}
\end{defn}
\noindent By construction, $s_{a_1,...,a_\al}(x_1,x_2,...,x_\al)$ is a Laurent polynomial in $\F_\al$.
We denote the difference operators corresponding to generalized Schur functions by
\begin{eqnarray}
{\mathcal M}_{a_1,a_2,...,a_\al}&:=&{\mathcal D}_\al\Big(s_{a_1,...,a_\al}(x_1,x_2,...,x_\al)\Big) \nonumber \\
&=&\sum_{I\subset [1,N]\atop |I|=\al } s_{a_1,...,a_\al}(\bx_I) \,\prod_{i\in I\atop j\not\in I} \frac{t x_i -x_j}{x_i-x_j} 
\, \prod_{i\in I} \Gamma_i \label{schurmacdo}
\end{eqnarray}
where $\bx_I$ stands for the ordered collection of variables $(x_i)_{i\in I}$.

As a particular case, we recover ${\mathcal M}_{\al;n}={\mathcal M}_{n,n,...,n}$, which when $n>0$
corresponds to the usual Schur function $s_{n,n,...,n}(x_1,...,x_\al)=(x_1x_2\cdots x_\al)^n$, 
with rectangular Young diagram $\al\times n$. 
\end{enumerate}

%The latter is known to be the classical character of the irreducible module of $A_{N-1}$ with rectangular 
%Young diagram $\al\times n$, also identified with the Kirillov-Reshetikhin module $KR_{\al;n}$. As such, it satisfies 
%the classical $Q$-system for $A_{N-1}$. It is interesting to note 
%that ${\mathcal D}_\al$ maps this classical character to the ``quantum character" ${\mathcal M}_{\al;n}$.

\subsection{Main results and outline}
The paper is organized as follows.

Our first task is to show in Section \ref{dahasec} how the operators ${\mathcal M}_{\al;n}$
of \eqref{genmacdop} appear naturally in the 
context of the functional representation of the spherical DAHA. To this end, we introduce 
families of commuting operators $Y_{i,n}= (X_1 X_2 \cdots X_{i-1})^{-n}\, Y_i \, (X_1 X_2 \cdots X_i)^n$ 
in terms of the standard generators $X_i,Y_i$ of the DAHA (see Sect. \ref{defdahasec} for definitions).
Next we show in Theorem \ref{genmacthm} that the operators \eqref{genmacdop} correspond to their elementary
symmetric functions, in the same way as usual Macdonald operators are obtained from elementary symmetric
functions of the $Y_i$-generators of the DAHA. We find that $Y_{i,n}$ is proportional to the $n$-th iterate of the action
on $Y_i$ of one particular generator of the $SL_2(\Z)$ action on DAHA. This clarifies the algebraic origin of the
operators ${\mathcal M}_{\al;n}$.

%In the rest of the paper, the main tool used is that of generating functions.

In Section \ref{qtorsec}, we show (Theorem \ref{gentoro}) that the generating functions for ${\mathcal M}_n={\mathcal M}_{1;n}$, 
for all operators with $\al=1$ is an element of the functional (level-$(0,0)$) representation of the quantum 
toroidal algebra of $\widehat{\mathfrak gl}_1$ \cite{Miki99,FJMM}, or more precisely a particular quotient 
thereof that imposes the finite number $N$ of variables (see \eqref{elipquo} and \eqref{defpsipm}). We call these the ``fundamental currents" 
for reasons which will become clear in Section 4.
This section explains a posteriori one of the findings of Ref. \cite{DFK16}, which corresponds to
the limit $t\to\infty$, where these generating functions are part of a non-standard representation 
of the quantum enveloping algebra of the affine algebra $\widehat{\mathfrak sl}_2$.

In Section \ref{bososec}, we give the plethystic formulation of the above currents, equivalent to their 
so-called bosonization in the limit of an infinite number of variables $N\to\infty$ (Theorem \ref{bobothm}). 
It describes their action on formal power series of the power sum functions 
$p_k=\sum_i x_i^k$ for $k \in \Z^*$. 
We show how this is related to a construction of Bergeron et al. \cite{BGLX}. One feature of this $N\to\infty$ limit
is the fact that the limiting currents pertain to a non-trivial representation of the quantum toroidal algebra,
with levels $(1,0)$ (as opposed to level $(0,0)$ for finite $N$).

In Section \ref{shufflesec}, we give an alternative definition ${\mathcal M}_\al(P)$ (see Definition \ref{malphadef})
for the generalized Macdonald operators  ${\mathcal D}_\al(P)$ of  Definition \ref{gmacdef}, which expresses
them as  a multiple constant term involving products of the above fundamental currents. The coincidence of 
these two definitions is the subject of Theorem \ref{mainthm}. The latter allows to write the multiple generating function
for the operators ${\mathcal M}_{a_1,...,a_\al}$ of \eqref{schurmacdo} in terms solely of those of the ${\mathcal M}_{n}$'s
(Theorem \ref{Mofm}). We conjecture (Conjecture \ref{polyconj}) that these may be reduced to polynomial expressions 
modulo the relations of the quantum toroidal algebra, and give the proof in the case $\al=2$ (Theorem \ref{polynomialitythm}). Such polynomials play the role of $(q,t)-$determinants (see the expression \eqref{mtwopol} for $\al=2$). 
In Section \ref{shuprosec}, we show that
the definition of ${\mathcal M}_\al(P)$ is naturally compatible with a suitably defined non-commutative product 
$*: \F_\al\times \F_\beta\to \F_{\al+\beta}$, 
$(P,P')\mapsto P*P'$ (the shuffle product \cite{FO,NegutShuffle}), 
which satisfies the morphism  property: ${\mathcal M}_{\al+\beta}(P*P')={\mathcal M}_\al(P){\mathcal M}_\beta(P')$
(see Thorem \ref{shufmac}).
We may therefore translate relations between the generalized Macdonald operators into shuffle product identities, 
which sometimes are easier to prove (see examples in Sections \ref{appsecone}-\ref{appsecthree}).

Section \ref{EHAsec} presents a functional representation of the EHA in terms of our generalized 
Macdonald operators. The established connection between spherical DAHA and EHA in the case of $A_\infty$ 
(infinite number of variables) \cite{SHIVAS} extends to the quotient corresponding to $A_{N-1}$ (finite
number $N$ of variables). This connection allows to derive new formulas for the operators 
${\mathcal M}_{\al;n}$ of \eqref{genmacdop} as {\it polynomials} of the fundamental operators with $\al=1$, thus proving 
Conjecture \ref{polyconj} for $a_1=a_2=\cdots =a_\al=n$ (Theorem \ref{polpol}).

In Section \ref{whitaklimsec}, we explore the dual Whittaker limit $t\to\infty$
of the constructions of this paper. In particular, we relate the finite $t$ Macdonald operators 
${\mathcal M}_{n}$ to their $t\to\infty$ limit $M_n$. We also find an explicit formula for the
$t\to\infty$ limit  of the operators ${\mathcal M}_{a_1,a_2,...,a_\al}$ as a quantum determinant, which involves
a summation over $\al\times \al$ Alternating Sign Matrices (Theorem \ref{qdethm}).
%More generally, the operators
%${\mathcal M}_{\al;n}$ are variables
%in a $t$-deformation of the quantum cluster algebra associated to the $Q$-system for $A_{N-1}$. 
By considering the $t\to\infty$ limit of the 
shuffle product, we find an alternative shuffle expression for the quantum cluster algebra relations (Theorem \ref{msyshuf}).

Section \ref{concsec}
gathers a few concluding remarks on the $(q,t)$-determinant that expresses the operator ${\mathcal M}_{a_1,...,a_\al}$
as a polynomial of the ${\mathcal M}_n$'s, and suggest that the $A_{N-1}$ 
EHA quotient corresponding to a finite number $N$ of variables is the natural $t$-deformation of the quantum
$Q$-system algebra.
%\color{red}
%Deformed commutation relations. $q,t$- Laurent phenomenon. Quantum determinant. Flat connection. Yang-Baxter.
%Compare to \cite{FHHSY}.
%\color{black}

\vskip.2in

\noindent{\bf Acknowledgments.}  We thank O. Babelon, J.-E. Bourgine, I. Cherednik, M. Jimbo, Y. Matsuo,
A. Negut, V. Pasquier, O. Schiffmann for
valuable discussions and especially F. Bergeron for bringing Ref. \cite{BGLX} to our attention and providing us 
with detailed explanations. R.K.Õs research is supported by NSF grant DMS-1404988 and the conference travel grant NSF DMS-1643027.
P.D.F. acknowldeges support from the NSF grant DMS-1301636 and the Morris and 
Gertrude Fine endowment. We thank the Institut de Physique Th\'eorique (IPhT) of Saclay, France, 
and the Institut Henri Poincar\'e, Paris program on ``Combinatorics and Interactions", France for hospitality during various stages of this work.

%\subsection{Main results}
%\input{results}
%\section{Quantum $Q$-system, current algebras and difference operators}
%\input{Msys}

\section{Formulation of the generalized Macdonald operators in DAHA}\label{dahasec}
%DAHA
%Realization of difference operators in terms of DAHA

In this section, we find the explicit elements in the spherical double affine Hecke algebra (sDAHA) whose functional representation are the generalized Macdonald operators $\mathcal M_{\al;n}$ of Equation \eqref{genmacdop}. The standard definitions and properties in this section can be found in Cherednik's book
\cite{Cheredbook}.

\subsection{The $A_{N-1}$ DAHA: Definition and relations}

\subsubsection{Generators and relations}\label{defdahasec}
Let $q$ and $\theta$ be indeterminates, where $\theta=t^{\half}$. 
The $A_{N-1}$ double affine Hecke algebra is the algebra generated 
over $\C(q,t)$ by the generators $\{X_i,Y_i,T_j: i\in [1,N], j\in [1,N-1]\}$ ,
%where $\{T_j\}_{j\in[1,n-1]}$ generate the finite Hecke algebra, 
subject to the following relations:
\begin{eqnarray}
 T_i\, T_{i+1}\,T_i&=&T_{i+1}\, T_i\,T_{i+1}; \label{braid}\\
  (T_i-\theta)(T_i+\theta^{-1})&=&0; \nonumber \\
 T_i\, X_i\,T_i&=&X_{i+1}; \qquad T_i^{-1}\,Y_i\,T_i^{-1}=Y_{i+1},\qquad (1\leq i\leq N-1);\label{TY}\\
 T_i\, X_j&=&X_j\, T_i; \qquad
  T_i\, Y_j =Y_j\, T_i ,\quad (j\neq i,i+1);\nonumber\\
 X_1\,Y_{2}&=&Y_2\,T_1^{2}\,X_1;\nonumber \\
X_i X_j &=& X_j X_i;\qquad Y_i Y_j = Y_j Y_i; \nonumber \\ 
{Y_1\cdots Y_N}\, X_j&=&q \, X_j\, {Y_1\cdots Y_N};\nonumber \\
 {X_1\cdots X_N}\, Y_j&=&q^{-1}\, Y_j\, {X_1\cdots X_N} .\label{prodXy}
\end{eqnarray}

\subsubsection{Other useful relations}
We list here several useful relations among the generators of the DAHA which follow from the definition above.
Some of them can be found in \cite{Cheredbook} (see Sect. {\bf Some relations} on pp 103-104, eqs. (1.4.64-69)).

The following relations give a way to reorder $X_i$ and $Y_{j}$:
\begin{eqnarray*}
X_i\,Y_{i+1}&=&Y_{i+1}\,T_i^{2}\,X_i=Y_{i+1}\,T_i \,X_{i+1}\,T_i^{-1}=
T_i^{-1}\,Y_i\, T_i \, X_i,\qquad (i=1,2,...,N-1);\label{firstxy}\\
Y_i \, X_{i+1}&=&X_{i+1} \, T_i^{-2} \, Y_i=X_{i+1} \, T_i^{-1}\, Y_{i+1}\,T_i=T_i\,X_i\,T_i^{-1}\,Y_i,\qquad (i=1,2,...,N-1).\label{secxy}
\end{eqnarray*}
%The latter is readily obtained from the former by multiplying from the left and right by $T_i$.
More generally, let $1\leq i\leq j \leq N$. Then
\begin{eqnarray*} X_i \,Y_{j+1}&=&Y_{j+1}(T_jT_{j-1}\cdots T_{i+1} T_i^2 T_{i+1}^{-1}T_{i+2}^{-1}\cdots T_j^{-1})X_i\\
&=&Y_{j+1}(T_jT_{j-1}\cdots T_i)X_{i+1}(T_i^{-1}T_{i+1}^{-1}\cdots T_j^{-1});\\
Y_i \,X_{j+1}&=&X_{j+1}(T_j^{-1}T_{j-1}^{-1}\cdots T_{i+1}^{-1} T_i^{-2} T_{i+1}T_{i+2}\cdots T_j) Y_i\\
&=&X_{j+1}(T_j^{-1}T_{j-1}^{-1}\cdots T_i^{-1})Y_{i+1}(T_iT_{i+1}\cdots T_j).
\end{eqnarray*}
Moreover,
\begin{eqnarray*}
(X_j X_{j-1}\cdots X_i )\, Y_{j+1}
&=&Y_{j+1} (T_jT_{j-1}\cdots T_{i+1}T_i^2T_{i+1}T_{i+2}\cdots T_{j}) (X_j X_{j-1}\cdots X_i)\nonumber \\
&=& Y_{j+1} (T_jT_{j-1}\cdots T_i) (X_{j+1} X_{j}\cdots X_{i+1} ) (T_i^{-1}T_{i+1}^{-1}\cdots T_j^{-1})\label{comXY} .
\end{eqnarray*}
This equation can be iterated to obtain
\begin{equation}\label{gencoXY}
(X_j X_{j-1}\cdots X_i )^n\, Y_{j+1}
=Y_{j+1} (T_jT_{j-1}\cdots T_i) (X_{j+1} X_{j}\cdots X_{i+1})^n (T_i^{-1}T_{i+1}^{-1}\cdots T_j^{-1}).
\end{equation}

%The next Lemma will also be used below, and is a consequence of the defining relations.
\begin{lemma}\label{comXT}
For all $1\leq i\leq j\leq N$, 
$$(X_iX_{i+1}\cdots X_j)(T_i^{-1}T_{i+1}^{-1}\cdots T_{j-1}^{-1})
=(T_i^{-1}T_{i+1}^{-1}\cdots T_{j-1}^{-1})(X_iX_{i+1}\cdots X_j).$$
\end{lemma}
\begin{proof}
For all $1\leq i\leq j\leq N$, we have:
\begin{eqnarray*}(X_iX_{i+1}\cdots X_j)(T_i^{-1}T_{i+1}^{-1}\cdots T_{j-1}^{-1})
&=&(X_iX_{i+1}\cdots X_{j-1})(T_i^{-1}T_{i+1}^{-1}\cdots T_{j-2}^{-1})T_{j-1}X_{j-1}\\
&=&X_i (T_iX_i \cdots T_{j-2}X_{j-2}T_{j-1}X_{j-1})\\
&=&T_i^{-1}X_{i+1}X_i(T_{i+1}X_{i+1}\cdots T_{j-1}X_{j-1})\\
&=&T_i^{-1}X_{i+1}(T_{i+1}X_{i+1}\cdots T_{j-1}X_{j-1})X_i\\
&=&T_i^{-1}T_{i+1}^{-1}X_{i+2}(T_{i+2}X_{i+1}\cdots T_{j-1}X_{j-1})X_iX_{i+1}\\
&=&(T_i^{-1}T_{i+1}^{-1}\cdots T_{j-1}^{-1})(X_iX_{i+1}\cdots X_j).
\end{eqnarray*}
The lemma follows.
\end{proof}

\subsubsection{The generator $\pi$}

Define
$$\pi=Y_1^{-1}T_1T_2\cdots T_{N-1} . $$
Using $Y_{i+1}=T_i^{-1}\,Y_i\,T_i^{-1}$, we may express each of the $Y_i$s as:
$$Y_i=T_{i}T_{i+1}\cdots T_{N-1} \pi^{-1} T_1^{-1}T_2^{-1}\cdots T_{i-1}^{-1}\qquad (i=1,2...,N).$$
In other words, we may express $\pi$ in $N$ different manners:
\begin{equation}\label{piti}
\pi=T_1^{-1}T_2^{-1}\cdots T_{i-1}^{-1} Y_i^{-1} T_iT_{i+1}\cdots T_{N-1} \qquad (i=1,2,...,N).
\end{equation}

The following two Lemmas show that $\pi$ acts as a translation operator on $T_i$s and $X_i$s:
\begin{lemma}
$$\pi T_i=T_{i+1}\pi \qquad (i=1,2,...,N-2)$$
\end{lemma}
\begin{proof}
Using the $i$th expression for $\pi$ \eqref{piti}, we compute:
\begin{eqnarray*}\pi\, T_i&=&(T_1^{-1}T_2^{-1}\cdots T_{i-1}^{-1}) Y_i^{-1} (T_iT_{i+1}T_i)( T_{i+2}\cdots T_{N-1})\\
&=&(T_1^{-1}T_2^{-1}\cdots T_{i-1}^{-1}) Y_i^{-1}T_{i+1}(T_iT_{i+1}\cdots T_{N-1})\\
&=&T_{i+1}(T_1^{-1}T_2^{-1}\cdots T_{i-1}^{-1}) Y_i^{-1}(T_iT_{i+1}\cdots T_{N-1})=T_{i+1}\pi
\end{eqnarray*}
where we have first used the braid relations \eqref{braid} and the commutation relations\eqref{TY}.
\end{proof}

\begin{lemma}
$$\pi \, X_i=X_{i+1}\,\pi \qquad (i=1,2,...,N-1)\quad {\rm and} \quad \pi X_{N}=q^{-1}\, X_1\,\pi$$
\end{lemma}
\begin{proof}
Using
$Y_i X_{i+1}=X_{i+1} T_i^{-2} Y_i$,
wich implies $Y_i^{-1}X_{i+1} T_i^{-2}=X_{i+1}Y_i^{-1}$,
and the expression \eqref{piti} for $\pi$, we compute:
\begin{eqnarray*}
\pi X_i&=&(T_1^{-1}T_2^{-1}\cdots T_{i-1}^{-1}) Y_i^{-1} (T_i T_{i+1}\cdots T_{N-1})  X_i=
(T_1^{-1}T_2^{-1}\cdots T_{i-1}^{-1}) Y_i^{-1}T_i X_i (T_{i+1}\cdots T_{N-1})\\
&=&(T_1^{-1}T_2^{-1}\cdots T_{i-1}^{-1}) Y_i^{-1}X_{i+1}T_i^{-1} (T_{i+1}\cdots T_{N-1})\\
&=&X_{i+1}(T_1^{-1}T_2^{-1}\cdots T_{i-1}^{-1}) Y_i^{-1}T_i^{-1} (T_{i+1}\cdots T_{N-1})=X_{i+1}\pi
\end{eqnarray*}
The last relation is obtained by using 
$$\pi {X_1\cdots X_N}=q^{-1}{X_1\cdots X_N}\pi,$$ 
obtained from the relation \eqref{prodXy}.
\end{proof}

\subsection{The functional representation of the DAHA}\label{secpol}

Since the variables $X_1,...,X_N$ commute among themselves, we can define the functional representation $\rho$ of the DAHA acting on $V=C_{q,t}(x_1,...,x_N)$ as follows:
\begin{eqnarray*}
\rho(X_i)\, f(x_1,x_2,...,x_N)&=& x_i\, f(x_1,x_2,...,x_N), \quad f\in V;\\
\rho(s_i)\, f(x_1,...,x_i,x_{i+1},...,x_N) &=&f(x_1,...,x_{i+1},x_{i},...,x_N),\quad f\in V; \\
\rho(T_i)&=&\theta \rho(s_i) +\frac{\theta-\theta^{-1}}{x_ix_{i+1}^{-1}-1} (\rho(s_i)-1)\\
&=&
\frac{\theta x_i-\theta^{-1}x_{i+1}}{x_i-x_{i+1}}\rho(s_i)-x_{i+1}\frac{\theta-\theta^{-1}}{x_i-x_{i+1}};\\
\rho(T_i^{-1})&=&\rho(T_i)-\theta+\theta^{-1};\\&=&
\frac{\theta x_i-\theta^{-1}x_{i+1}}{x_i-x_{i+1}}\rho(s_i)-x_{i}\frac{\theta-\theta^{-1}}{x_i-x_{i+1}};\\
\rho(\pi)\, f(x_1,x_2,...,x_{N})&=& f(x_2,x_3,...,x_{N},q^{-1}x_1), \quad f\in V; \\
\rho(Y_i)&=&\rho(T_{i}T_{i+1}\cdots T_{N-1} \pi^{-1} T_1^{-1}T_2^{-1}\cdots T_{i-1}^{-1})\qquad (i=1,2,...,N).\\
\end{eqnarray*}

%In other words, $\pi$ acts on functions of $x_1,...,x_{N}$ via the substitutions:
%\begin{eqnarray*} \pi:(x_1,x_2,...,x_{N})&\mapsto& (x_2,x_3,...,x_{N},q^{-1}x_1)\\
%\pi^{-1}:(x_1,x_2,...,X_{N})&\mapsto& (qx_{N},x_1,x_2...,x_{N-1})
%\end{eqnarray*}
We see that the $q$-shift operators $\Gamma_i$ of \eqref{shift} are a representation of the following element in the DAHA:
\begin{equation} \label{defgamma}
\Gamma_i=\rho(s_i s_{i+1} \cdots s_{N-2}s_{N-1} \pi^{-1}s_1 s_2\cdots s_{i-1}),\quad i=1,2,...,N.
\end{equation}

%\section{Generalized Macdonald difference operators}

\subsection{Macdonald difference operators}

The operators $Y_1,...,Y_{N}$ commute among themselves. Therefore one can define
the elementary symmetric functions $e_m(Y_1,...,Y_{N})$ unambiguously. 
\begin{defn} The Macdonald operators are
\begin{equation}\label{orimac}
{D}_\al:=e_\al(Y_1,...,Y_{N}),\qquad (\al=0,1,2,...,N). \end{equation}
\end{defn}
Equivalently one can write
$\sum_{\al=0}^{N} z^\al {D}_\al=\prod_{i=1}^{N} (1+z Y_i)$.

It is well-known that the operators  $\rho({D}_\al)$ act on the space ${\mathcal S}_N$  of symmetric {\it functions} in the variables $x_1,...,x_{N}$.
The following is a standard result of Macdonald theory:
\begin{thm}\label{macdopthm}
The restriction of the operators $\rho({D}_\al)$ to ${\mathcal S}_N$ is
$$\rho({D}_\al)\vert_{{\mathcal S}_N}=:{\mathcal D}_\al=\theta^{-\al(N-\al)}\, \widetilde{\mathcal D}_{\al}$$
where $\widetilde{\mathcal D}_\al$ are the Macdonald operators \eqref{macdop}.
\end{thm}

%The difference operators ${\mathcal D}_\al$ are the Macdonald difference operators 
%(in Cherednik's normalization).
%Equivalently we also have:
%$$\sum_{\al=0}^{r+1} z^\al {\mathcal D}_\al=\prod_{i=1}^{r+1} (1+z Y_i)$$
%expressed in the spherical DAHA polynomial representation.

\subsection{More commuting operators}

In this section, we introduce families of commuting operators $\{Y_{i,n}\}_{i\in [1,N]}$ for each $n\in \Z$. These are related
to Cherednik's $SL_2(\Z)$ action on $Y_i$ by $n$ iterations of the generator $\tau_+$.

\subsubsection{Definition and commutation}

\begin{defn}
We introduce the family of operators:
%\begin{equation} 
$$
{Y}_{i,n} = (X_1 X_2 \cdots X_{i-1})^{-n}\, Y_i \, (X_1 X_2 \cdots X_i)^n ,\qquad (i=1,2,...,N;n\in \Z).
$$
%\end{equation}
\end{defn}

In particular, $Y_{i,0}=Y_i$, and $Y_{1,n}=Y_1 X_1^n$. We also see that
$$Y_{N,n}=(X_1\cdots X_{N-1})^{-n}Y_{N}\,(X_1\cdots X_N)^n=q^nX_N^n Y_N.$$

\begin{lemma}\label{commuYin}
For fixed $n\in \Z$, the elements $\{Y_{i,n}: i\in[1,N]\}$ commute among themselves:
$$Y_{i,n}\, Y_{j,n}=Y_{j,n}\, Y_{i,n}\qquad \forall\, i,j\in[1,N].$$
\end{lemma}
\begin{proof}
Writing $j=i+k$, $k>0$, we have:
\begin{eqnarray*}
(X_1\cdots X_{i-1})^nY_{i,n}&& \!\!\!\!\!\!\!\!\!\!\!\!\!\! Y_{i+k,n}(X_1\cdots X_{i})^{-n}\\
&=&Y_i (X_{i+1}\cdots X_{i+k-1})^{-n}Y_{i+k} (X_{i+1}\cdots X_{i+k})^n\\
&=& Y_{i+k}Y_i  (T_{i+k-1}\cdots T_{i+1})(X_{i+2}\cdots X_{i+k})^{-n} 
(T_{i+1}^{-1}\cdots T_{i+k-1}^{-1})(X_{i+1}\cdots X_{i+k})^n\\
&=&  Y_{i+k}(T_{i+k-1}\cdots T_{i+1}) (Y_i  X_{i+1}^n) (T_{i+1}^{-1}\cdots T_{i+k-1}^{-1})\\
&=&  Y_{i+k}(T_{i+k-1}\cdots T_i)X_i^n(T_i^{-1}\cdots T_{i+k-1}^{-1})Y_i 
\end{eqnarray*}
where we have first used the relation \eqref{gencoXY}:
$$(X_{i+1}\cdots X_{i+k-1})^{-n}Y_{i+k}=Y_{i+k}(T_{i+k-1}\cdots T_{i+1})(X_{i+2}\cdots X_{i+k})^{-n} 
(T_{i+1}^{-1}\cdots T_{i+k-1}^{-1})$$
then Lemma \ref{comXT}:
$$(T_{i+1}^{-1}\cdots T_{i+k-1}^{-1})(X_{i+1}\cdots X_{i+k})=(X_{i+1}\cdots X_{i+k})(T_{i+1}^{-1}\cdots T_{i+k-1}^{-1})$$
and finally $Y_i  X_{i+1}^n=T_iX_i^nT_i^{-1} Y_i$ by iteration of \eqref{secxy}.
Likewise, we have:
\begin{eqnarray*}
(X_1\cdots X_{i-1})^nY_{i+k,n}&& \!\!\!\!\!\!\!\!\!\!\!\!\!\! Y_{i,n}(X_1\cdots X_{i})^{-n}\\
&=&(X_{i}\cdots X_{i+k-1})^{-n}Y_{i+k} (X_{i}\cdots X_{i+k})^{n}Y_i\\
&=&Y_{i+k} (T_{i+k-1}\cdots T_i)(X_{i+1}\cdots X_{i+k})^{-n}(T_i^{-1}\cdots T_{i+k-1}^{-1})(X_{i}\cdots X_{i+k})^{n}Y_i\\
&=&Y_{i+k} (T_{i+k-1}\cdots T_i)X_i^n (T_i^{-1}\cdots T_{i+k-1}^{-1})Y_i\\
\end{eqnarray*}
by use of the relations
\begin{eqnarray*}
(X_{i}\cdots X_{i+k-1})^{-n}Y_{i+k}&=& Y_{i+k}(T_{i+k-1}\cdots T_{i})(X_{i+1}\cdots X_{i+k})^{-n} 
(T_{i}^{-1}\cdots T_{i+k-1}^{-1})\\
(T_i^{-1}\cdots T_{i+k-1}^{-1})(X_{i}\cdots X_{i+k})&=& (X_{i}\cdots X_{i+k})(T_i^{-1}\cdots T_{i+k-1}^{-1})
\end{eqnarray*}
The lemma follows.
\end{proof}

\subsubsection{Expression in the functional representation}

From this point on, we will work in the functional representation $\rho$ of Section~\ref{secpol}.
We introduce the following symmetric function of $X_1,...,X_{N}$:
\begin{equation}\label{gammadefn}
\gamma:=\exp\left\{ \sum_{i=1}^{N} \frac{{\rm Log}(X_i)^2}{2 {\rm Log}(q)}\right\} 
\end{equation}
This element does not belong to the DAHA, but to a suitable completion (see \cite{Cheredbook}). Nevertheless, it has some useful commutation relations with elements of the DAHA.
%Alternatively, writing $X_i=q^{\xi_i}$, we may write $\gamma=q^{\sum_{i=1}^{N}\frac{\xi_i^2}{2}}$.
%We have the following:

\begin{lemma}\label{piga}
$$\pi^{-1} \, \gamma=q^{\frac{1}{2}}X_{N} \,  \gamma \, \pi^{-1} $$
\end{lemma}
\begin{proof}
Using $\pi^{-1}X_i=X_{i-1}\pi^{-1},$ for $i\in[2,N]$ and $\pi^{-1}X_1=q X_{N}\pi^{-1}$,
we compute
$$\pi^{-1} \left(\sum_{j=1}^{r+1} \frac{{\rm Log}(X_j)^2}{2 {\rm Log}(q)}\right)=
\left(\sum_{j=1}^{r+1} \frac{{\rm Log}(X_j)^2}{2 {\rm Log}(q)}+{\rm Log}(X_{r+1})+\frac{{\rm Log}(q)}{2}\right)\pi^{-1}$$
and the Lemma follows.
\end{proof}

%For later use, we also have the following:

\begin{lemma}\label{gaga}
$$\Gamma_i \, \rho(\gamma)=q^{\frac{1}{2}}x_i \,  \rho(\gamma) \, \Gamma_i $$
\end{lemma}
\begin{proof}
We simply note that 
$$\Gamma_i \left(\sum_{j=1}^{r+1} \frac{{\rm Log}(x_j)^2}{2 {\rm Log}(q)}\right)=
\left(\sum_{j=1}^{r+1} \frac{{\rm Log}(x_j)^2}{2 {\rm Log}(q)}+{\rm Log}(x_i)+\frac{{\rm Log}(q)}{2}\right)\Gamma_i$$
\end{proof}

\begin{thm}\label{conjthm}
For all $n\in \Z$, we have:
$$Y_{i,n}=q^{\frac{n}{2}}\,\gamma^{-n}\,Y_i \, \gamma^n.$$
\end{thm}
\begin{proof}
As $\gamma$ is a symmetric function of the $X_i$'s, it commutes with $s_i$,
and with all the $T_i$ in the functional representation. 
We compute:
\begin{eqnarray*}
\gamma^{-n}\,Y_i \, \gamma^n&=&\gamma^{-n}\,T_i... T_{N-1} \pi^{-1}T_1^{-1}...T_{i-1}^{-1}\, \gamma^n\\
&=&T_i... T_{N-1}\, \gamma^{-n}\,\pi^{-1}\, \gamma^n\, T_1^{-1}...T_{i-1}^{-1}\\
&=& q^{\frac{n}{2}}T_i... T_{N-1}\, X_{N}^n\,\pi^{-1}\,  T_1^{-1}...T_{i-1}^{-1}\\
&=&q^{-\frac{n}{2}}T_i... T_{N-1}\, (X_1...X_{i-1})^{-n} (X_1...X_{i-1})^n\,\pi^{-1}\, X_{1}^n T_1^{-1}...T_{i-1}^{-1}\\
&=&q^{-\frac{n}{2}}(X_1...X_{i-1})^{-n}T_i... T_{N-1}\, \pi^{-1}(X_1...X_{i})^nT_1^{-1}...T_{i-1}^{-1}\\
&=&q^{-\frac{n}{2}}(X_1...X_{i-1})^{-n}Y_i(X_1...X_i)^n=q^{-\frac{n}{2}}\, Y_{i,n}.
\end{eqnarray*}
where we have first used the fact that $\gamma$ is a symmetric function of $X_1,...,X_{N}$ 
and therefore commutes with $T_j$ for all $j$, then we have used Lemma \ref{piga}, and 
finally the commutations between the $X$'s and the $T$'s,
in particular that the symmetric function $X_1...X_i$ of the variables $X_1,...,X_i$ commutes with $T_j$
for $j=1,2,...,i-1$, and also that $X_1...X_{i-1}$ commutes with $T_j$ for $j=i,i+1,...,N-1$.
\end{proof}

\begin{remark}
Theorem \ref{conjthm} above implies immediately the commutation  of the operators $Y_{i,n}$ for any fixed $n$.
However, the element $\gamma$ \eqref{gammadefn} only belongs to a completion of the DAHA, as it involves infinite power series of the generators $X_i$. The direct proof of Lemma \ref{commuYin} bypasses this complication.
\end{remark}

\subsubsection{Comparison with the standard $SL(2,\Z)$ action on DAHA}

Theorem \ref{conjthm} allows to identify the conjugation w.r.t. $\gamma^{-1}$ as the action 
of the generator $\tau_+$ of the standard $SL(2,\Z)$ action on DAHA \cite{Cheredbook}. 
Indeed, using the definition\footnote{This definition is in fact dual to that of \cite{Cheredbook},
and corresponds to the definitions of Chapter 1.}:
\begin{eqnarray*}&&\tau_+(X_i)=X_i, \quad \tau_+(T_i)=T_i, \quad \tau_+(q)=q, \quad \tau_+(t)=t, \\
&& \tau_+(Y_1Y_2\cdots Y_i)=q^{-i/2}\, (Y_1 Y_2 \cdots Y_i)( X_1 X_2 \cdots X_i) 
\end{eqnarray*}
This leads to the expression $Y_{i,n}=q^{n/2}\, \tau_+^n (Y_i)$, which allows to finally identify:

\begin{lemma}\label{taupluslemma}
The generator $\tau_+$ of the standard $SL(2,\Z)$ action on DAHA reads:
$$ \tau_+={\rm ad}_{\gamma^{-1}} $$
namely it acts by conjugation w.r.t. $\gamma^{-1}$ of \eqref{gammadefn}.
\end{lemma}

The second generator $\tau_-$ of the standard $SL(2,\Z)$ action on DAHA is obtained by use of the
anti-involution $\epsilon$ of the DAHA acting on generators and parameters as:
$$\epsilon:\, (X_i,Y_i,T_i,q,t)\, \mapsto \, (Y_i,X_i,T_i^{-1},q^{-1},t^{-1})$$
and such that 
$$\tau_-=\epsilon \, \tau_+\, \epsilon $$
This leads to the following:

\begin{lemma}\label{taumoinslemma}
The generator $\tau_-$ corresponds to the conjugation w.r.t. the element $\eta^{-1}$, where:
\begin{equation}\label{etadefn}
\eta:=\exp\left\{ -\sum_{i=1}^{N} \frac{{\rm Log}(Y_i)^2}{2 {\rm Log}(q)}\right\} 
\end{equation}
namely
$$ \tau_-={\rm ad}_{\eta^{-1}} $$
\end{lemma}
\begin{proof}
Apply the anti-involution $\epsilon$ to $\gamma$, and note that $\epsilon(\gamma)=\eta$.
\end{proof}

\begin{remark}\label{nablarem}
The quantity $\eta^{-1}$ is very similar to the nabla operator $\nabla$ 
of \cite{BG} in a version suitable for the case of $N$ variables. To avoid confusion, we write $\eta^{-1}=\nabla^{(N)}$.
It is known \cite{Cheredbook} that the Macdonald polynomial $P_\lambda(x_1,...,x_N)$ for any 
partition $\lambda=(\lambda_1\geq \lambda_2 \geq \cdots \geq \lambda_N)$ (or equivalently Young diagram
with $\lambda_i$ boxes in row $i$), is an eigenvector
in the functional representation of any symmetric function $f(\{Y_i\}_{i=1}^N)$, with
eigenvalue $f(\{t^{\frac{N+1}{2}-i} q^{\lambda_i}\}_{i=1}^N)$. In particular, this holds for $\nabla^{(N)}$,
with the result:
\begin{eqnarray*}
{\nabla^{(N)}}\, P_\lambda&=& \exp\left\{ \frac{1}{2\,{\rm Log}(q)}\sum_{i=1}^{N} \left( (\frac{N+1}{2}-i)\,{\rm Log}(t)+\lambda_i\, {\rm Log}(q)\right)^2\right\}\, P_\lambda\\
&=&\left(C_N\, \prod_{i=1}^{N} q^{\frac{\lambda_i^2}{2}}\, t^{(\frac{N+1}{2}-i)\lambda_i}\right)\, P_\lambda
=C_N\, u_\lambda\, P_\lambda
\end{eqnarray*}
where ${\rm Log}(C_N)=\frac{N(N^2-1)}{24} \frac{{\rm Log}(t)^2}{{\rm Log}(q)}$ and
$u_\lambda=t^{\frac{N-1}{2}|\lambda|-n(\lambda)}q^{\frac{1}{2}|\lambda|+n(\lambda')} $, where
$n(\lambda)=\sum_i (i-1)\lambda_i$, and $\lambda'$ is the usual reflected diagram, with  $1$ box in the bottommost
$\lambda_1-\lambda_2$ rows, $2$ boxes in the next $\lambda_2-\lambda_3$ rows, etc, such that 
$n(\lambda')=\sum_i \lambda_i(\lambda_i-1)/2$.
In \cite{BG}, the $\nabla$ operator is defined to have eigenvalue $t^{n(\lambda)}q^{n(\lambda')}$ on the {\em modified}
Macdonald polynomials ${\widetilde H}_\lambda$, obtained from $P_\lambda$ by a certain transformation. 
We see that $\nabla^{(N)}$ is an analogue of the operator $\nabla$, acting instead on the $P_\lambda$.
\end{remark}

\subsection{Generalized Macdonald difference operators}

\begin{defn}
We define operators:
\begin{equation}\label{defnewmac} 
{D}_{\al;n}\equiv {D}_{\al;n}^{q,t} := q^{-\al n}\sum_{1\leq i_1<i_2<\cdots <i_\al \leq N}  
{Y}_{i_1,n}{Y}_{i_2,n} \cdots {Y}_{i_\al,n} \qquad (\al=0,1,...,N)
\end{equation}
Equivalently, we have:
$\sum_{\al=0}^{r+1}z^\al q^{n\al}\,{D}_{\al;n}=\prod_{i=1}^{N}(1+z Y_{i,n})$.
\end{defn}

Theorem \ref{conjthm} allows to rewrite immediately:
\begin{lemma}\label{dalga}
We have the following identity in the $A_{N-1}$ DAHA functional representation:
$$ \rho({D}_{\al;n}) =q^{-\frac{\al n}{2}}\rho(\gamma)^{-n} \rho({D}_{\al;0})\rho(\gamma)^{n}\ ,$$
where $ {D}_{\al;0}\equiv  {D}_{\al}$ are given by \eqref{orimac}.
\end{lemma}

\begin{thm}\label{genmacthm}
The operators $\rho({D}_{\al;n})$ leave the space ${\mathcal S}_N$ of symmetric functions of the $x$'s invariant. 
They take the following form:
\begin{equation}\label{genmac}
\rho({D}_{\al;n})\vert_{{\mathcal S}_N}=:{\mathcal D}_{\al;n}=\theta^{-\al(N-\al)}\, {\mathcal M}_{\al;n}
\end{equation}
where ${\mathcal M}_{\al;n}$ are the generalized Macdonald operators \eqref{genmacdop}, and ${ \mathcal D}_{\al;n}$
their slightly renormalized version:
\begin{equation}\label{defcald}
{ \mathcal D}_{\al;n} 
=\sum_{|I|=\al,\, I\subset [1,N]} (x_I)^n \prod_{i\in I \atop j\not \in I} 
\frac{\theta x_i -\theta^{-1} x_j}{x_i-x_j}\, \Gamma_I 
\end{equation}
where $x_I=\prod_{i\in I}x_i$ and $\Gamma_I=\prod_{i\in I}\Gamma_i$, with ${\mathcal D}_{\al;0}=\widetilde{\mathcal D}_\al$.
\end{thm}
\begin{proof}
We use Lemma \ref{gaga} to write for any index subset $I$ of cardinality $\al$:
$$ \Gamma_I\, \rho(\gamma)=(\prod_{i\in I}\Gamma_i) \, \rho(\gamma) =q^{\frac{\al}{2}}x_I \, \rho(\gamma) \, \Gamma_I $$
%with $x_I=\prod_{i\in I}x_i $ and $\Gamma_I=\prod_{i\in I}\Gamma_i$.
Starting from the expression of Lemma \ref{dalga}, we may now conjugate the formula of Theorem \ref{macdop} for 
$\rho(\widetilde{\mathcal D}_\al)=\rho({\mathcal D}_{\al;0})$ 
with $\rho(\gamma)^n$ as follows:
\begin{eqnarray*}\rho({\mathcal D}_{\al;n})=q^{-\frac{\al n}{2}}\rho(\gamma)^{-n}\rho(\widetilde{\mathcal D}_\al)\rho(\gamma)^n
&=&q^{-\frac{\al n}{2}}\sum_{|I|=\al,\, I\subset [1,N]}  \prod_{i\in I \atop j\not \in I} 
\frac{\theta x_i -\theta^{-1} x_j}{x_i-x_j}\, \rho(\gamma)^{-n} \Gamma_I \rho(\gamma)^n\\
&=&\sum_{|I|=\al,\, I\subset [1,N]} (x_I)^n 
\prod_{i\in I \atop j\not \in I} \frac{\theta x_i -\theta^{-1} x_j}{x_i-x_j}\,  \Gamma_I
\end{eqnarray*}
where we have used Lemma \ref{gaga}.
The Theorem follows.
\end{proof}

%Note that ${\mathcal D}_{\al,0}={\mathcal D}_{\al}$ of Theorem \ref{macdopthm}.

\section{Interpretation via quantum toroidal algebra}\label{qtorsec}
%toroidal
%quantum toroidal, t->infty limit, connection to elliptic Hall, Schiffman Vasserot

\subsection{Definitions and results}

In this section, we show that the generalized Macdonald operators $\{{D}_{1;n}: n\in \Z\}$ satisfy the relations of the quantum toroidal algebra \cite{Miki07,FJMM} at level $(0,0)$. We call the generating functions of these generators the fundamental currents, and in Section \ref{EHAsec} (Theorem \ref{polpol}) we show, by use of the elliptic Hall algebra, that the generators $D_{\al;n}$ are polynomials in the fundamental generators $\{D_{1;m}\}$.

\subsubsection{Quantum toroidal algebra $\widehat{\gl}_1$}

For generic parameters $q,t \in \C^\star$ let us introduce the functions:
\begin{equation}\label{defofg}
g(z,w):= (z-q w)(z-t^{-1}w)(z-q^{-1}t w),\quad 
G(x):=-\frac{g(1,x)}{g(x,1)}=\frac{(1-q x)(1-t^{-1}x)(1-q^{-1}t x)}{(1-q^{-1} x)(1-t x)(1-qt^{-1} x)}
\end{equation}
We will use the formal delta function
\begin{equation}\label{formaldet} \delta(u)=\sum_{n\in \Z} u^n \end{equation}
with the property that $\delta(u) f(u)=\delta(u) f(1)$ for any function $f$ which is non-singular at $u=1$.
The following definitions are borrowed from Ref.\cite{AFS}.

\begin{defn}\label{qtorgendef}
The quantum toroidal algebra of $\widehat{\gl}_1$ is defined by generators and relations. The generators are
the modes of the currents $x^\pm(z)=\sum_{n\in \Z} x_n^{\pm} z^{-n}$ and the series
$\varphi^\pm(z)=\sum_{n\geq 0} \varphi_n^{\pm} z^{\mp n}\in \C[[z^{\mp 1}]]$, together with two central elements
${\hat \gamma},\ {\hat \delta}$, and the relations read:
\begin{eqnarray*}
[\varphi^\pm(z),\psi^\pm(w)]&=&0,\qquad \varphi^+(z)\,\psi^-(w)=\frac{G({\hat \gamma}w/z)}{G({\hat \gamma}^{-1}w/z)}\,
\varphi^-(w)\,\varphi^+(z)\\
\varphi^+(z)\, x^\pm(w)&=&G({\hat \gamma}^{\mp 1}w/z)^{\mp 1}\, x^\pm(w)\,\varphi^+(z),\quad 
\varphi^-(z)\, x^\pm(w)=G({\hat \gamma}^{\mp 1}w/z)^{\pm 1}\, x^\pm(w)\,\varphi^-(z)\\
x^\pm(z)\, x^\pm(w)&=&G(z/w)^{\pm 1} \, x^\pm(w)\, x^\pm(z), \qquad \psi_0^\pm={\hat \delta}^{\mp 1},\\
{[} x^+(z), x^-(w) {] }&=& \frac{(1-q)(1-t^{-1})}{(1-q t^{-1})} \left\{ \delta({\hat \gamma}^{-1}z/w)\varphi^+({\hat \gamma}^{-1/2}z)
-\delta({\hat \gamma} z/w)\varphi^-({\hat \gamma}^{1/2}z)\right\}  \\
0&=&{\rm Sym}_{z_1,z_2,z_3}\left( \frac{z_2}{z_3} {\Big[}{ x^\pm}(z_1),{[}{x^\pm}(z_2),{ x^\pm}(z_3){]}{\Big]}\right)
\end{eqnarray*}
\end{defn}

A particular class of representations \cite{FHHSY} indexed by integers $(\ell_1,\ell_2)\in \Z_+^2$ (referred to as levels)
corresponds to diagonal actions of the central elements ${\hat \gamma},\ {\hat \delta}$
with respective eigenvalues $\gamma^{\ell_1},\gamma^{\ell_2}$, where $\gamma=(t q^{-1})^{1/2}$.

In the following, we adopt sligthtly different conventions for the naming of the generators (which, up to a change of variables
amounts to applying the anti-automorphism $\omega$ of the algebra, that sends $(x^\pm,\varphi^\pm,{\hat \gamma},{\hat \delta})$
to $(x^\mp,\varphi^\mp,{\hat \gamma}^{-1},{\hat \delta}^{-1})$, namely we set:
\begin{eqnarray}
&&x^+(z)=\frac{(1-q)(1-t^{-1})}{q^{1/2}}{\mathfrak e}(q^{-1/2}z),\quad x^-(z)
=\frac{(1-q^{-1})(1-t)}{q^{-1/2}}{\mathfrak f}(q^{-1/2} {\hat \gamma}^{-1}z),\nonumber \\ 
&&\varphi^+(z)=\psi^-(q^{-1/2}{\hat \gamma}^{1/2}z),\qquad \varphi^-(z)=\psi^+(q^{-1/2}{\hat \gamma}^{-1/2}z)
\label{dictio}
\end{eqnarray}

\subsubsection{Level $(0,0)$ quantum toroidal $\widehat{\gl}_1$}

Unless otherwise mentioned, we will mainly concentrate on the level-$(0,0)$ representations corresponding to 
$\hat \gamma=\hat \delta=1$. In our set of generators with currents
${\mathfrak e}(z)=\sum_{n\in \Z} z^n\, e_n$, ${\mathfrak f}(z)=\sum_{n\in \Z} z^n\, f_n$,
$\psi^\pm(z)=\sum_{n\geq 0} z^{\pm n}\, \psi^{\pm}_n\in \C[[z^{\pm 1}]]$, the corresponding relations read:
%The quantum toroidal algebra of $\widehat{\gl}_1$ is defined by generators and relations. The generators are $\{ e_n,f_n, \psi^{\pm}_m: n\in \Z, m\in \Z_+\}$, and the relations
%are most easily expressed in terms of generating functions:
%\begin{equation*}{\mathfrak e}(z)=\sum_{n\in \Z} z^n\, e_n\ , \quad {\mathfrak f}(z)=\sum_{n\in \Z} z^n\, f_n\ , \quad
%\psi^{\pm}(z) = \sum_{n\geq 0} z^{\pm n}\, \psi^{\pm}_n ,
%\end{equation*}
%where $z$ is a formal variable.
%In general, the algebra has two central charges $c_1,c_2$ explicitly appearing in the defining relations. In the functional representation which appears in our context, we have $c_1=c_2=0$, also referred to as level 0. 

\begin{defn}\label{qtordef}
The level $(0,0)$ quantum toroidal $\widehat{\gl}_1$ is the algebra generated by $\{ e_n,f_n, \psi^{\pm}_m: \ n\in \Z, \ m\in \Z_+\}$ with relations:
\begin{eqnarray*}
&&g(z,w){\mathfrak e}(z){\mathfrak e}(w)+g(w,z){\mathfrak e}(w){\mathfrak e}(z)=0, \ \ g(w,z){\mathfrak f}(z){\mathfrak f}(w)+g(z,w){\mathfrak f}(w){\mathfrak f}(z)=0, \\
&&g(z,w)\psi^\pm(z)\,{\mathfrak e}(w)+g(w,z){\mathfrak e}(w)\,\psi^\pm(z)=0, \ \ g(w,z)\psi^\pm(z)\,{\mathfrak f}(w)+g(z,w){\mathfrak f}(w)\,\psi^\pm(z)=0, \\
&&{[} {\mathfrak e}(z),{\mathfrak f}(w) {]}=\frac{\delta(z/w)}{g(1,1)}\, (\psi^+(z)-\psi^-(z)),\\
&&{\rm Sym}_{z_1,z_2,z_3}\left( \frac{z_2}{z_3} {\Big[}{\mathfrak e}(z_1),{[}{\mathfrak e}(z_2),{\mathfrak e}(z_3){]}{\Big]}\right)
=0 ,\ \ 
{\rm Sym}_{z_1,z_2,z_3}\left( \frac{z_2}{z_3} {\Big[}{\mathfrak f}(z_1),{[}{\mathfrak f}(z_2),{\mathfrak f}(z_3){]}{\Big]}\right)
=0 ,
\end{eqnarray*}
with $\psi_0^{\pm}=1$ and $\psi_{n}^{\pm}$ mutually commuting for $n\in \Z_+$.
\end{defn}

Note that when $t\to \infty$, we may set $g_0(z,w)=\lim_{t\to\infty} t^{-1}g(z,w)=z-qw$, and the first relations become relations in the quantum affine algebra of $\sl_2$ in the Drinfeld presentation (with a non-standard deformation parameter $\sqrt{q}$, as in the Hall algebra of \cite{spherical_hall}).
The last two identities are Serre-type relations and distinguish the quantum toroidal algebra from the original Ding-Iohara algebra \cite{DI}.

\subsubsection{Macdonald currents}
We claim that the generalized Macdonald operators introduced in Equation \eqref{defcald} are elements of a functional representation of a quotient of the quantum toroidal algebra. In order to make the comparison explicit, we first define generating currents for the generalized Macdonald operators \eqref{defcald}
for $\al=1,2,...,N$:
\begin{equation}\label{highercur}
{\mathfrak e}_\al(z):=\frac{q^\frac{\al}{2}}{(1-q)^\al}\sum_{n\in \Z} q^{n\al/2}\, z^n\, {\mathcal D}_{\al;n}^{q,t} , \qquad 
{\mathfrak f}_\al(z):=\frac{q^{-\frac{\al}{2}}}{(1-q^{-1})^\al}\sum_{n\in \Z} q^{-n\al/2}\, z^n\, {\mathcal D}_{\al;n}^{q^{-1},t^{-1}} .
\end{equation}
where we have indicated the $q,t$ dependence as superscripts, so that:
\begin{equation}\label{defdtilde}
{ \mathcal D}_{\al;n}^{q^{-1},t^{-1}} 
=\sum_{|I|=\al,\, I\subset [1,N]} (x_I)^n \prod_{i\in I \atop j\not \in I} 
\frac{\theta^{-1} x_i -\theta x_j}{x_i-x_j}\, \Gamma_I^{-1}
\end{equation}
and $\Gamma_i^{-1}$ acts on functions of $x_1,...,x_{N}$ by multiplying the $i$-th variable $x_i$  by $q^{-1}$.

\begin{remark}\label{remef}
Note also that if $S$ denotes the involution acting on functions of $(x_1,...,x_{N})$ by sending $x_i\mapsto x_i^{-1}$ 
for all $i$, then we have $S\Gamma_IS=\Gamma_I^{-1}$ and 
${ \mathcal D}_{\al;-n}^{q^{-1},t^{-1}} =S{ \mathcal D}_{\al;n}^{q,t}S$, so that 
$$S\, {\mathfrak e}_\al(z)\, S=(-1)^\al \, {\mathfrak f}_\al(z^{-1})\ .$$
\end{remark}

The currents \eqref{highercur} can be explicitly expressed as
\begin{eqnarray*}
{\mathfrak e}_\al(z)&=& \frac{q^\frac{\al}{2}}{(1-q)^\al}\sum_{|I|=\al,\, I\subset [1,N]} \delta(q^{\al/2}z x_I) 
\prod_{i\in I \atop j\not \in I} \frac{\theta x_i -\theta^{-1} x_j}{x_i-x_j}\, \Gamma_I, \\
{\mathfrak f}_\al(z)&=&\frac{q^{-\frac{\al}{2}}}{(1-q^{-1})^\al} \sum_{|I|=\al,\, I\subset [1,N]} \delta(q^{-\al/2}z x_I) 
\prod_{i\in I \atop j\not \in I} \frac{\theta^{-1} x_i -\theta x_j}{x_i-x_j}\, \Gamma_I^{-1},
\end{eqnarray*}
by use of the formal $\delta$ function \eqref{formaldet}.

Note that the finite number $N$ of variables implies the vanishing relations:
\begin{equation}\label{elipquo} {\mathfrak e}_{N+1}(z)=0 \ \ {\rm and} \ \ {\mathfrak f}_{N+1}(z)=0 .
\end{equation}

\subsubsection{Main result}
We call ${\mathfrak e}(z):={\mathfrak e}_1(z)$ and ${\mathfrak f}(z):={\mathfrak f}_1(z)$ the fundamental currents. Our main result is that these satisfy relations in the quantum toroidal algebra, as suggested by the notation. We will see later that the vanishing condition \eqref{elipquo} implies a certain quotient of the algebra by polynomials of degree $N+1$ in the fundamental generators.

\begin{thm}\label{gentoro}
The Macdonald currents ${\mathfrak e}(z):={\mathfrak e}_1(z)$ and ${\mathfrak f}(z):={\mathfrak f}_1(z)$ \eqref{highercur},
together with the series:
\begin{equation}\label{defpsipm}
\psi^{\pm}(z):=\prod_{i=1}^{N} 
\frac{(1-q^{-\frac{1}{2}}t (z x_i)^{\pm 1})(1-q^{\frac{1}{2}}t^{-1} (z x_i)^{\pm 1})}{(1-q^{-\frac{1}{2}} (z x_i)^{\pm 1})(1-q^{\frac{1}{2}} (z x_i)^{\pm 1})} \in \C[[z^{\pm 1}]]
\end{equation}
satisfy the level $(0,0)$ quantum toroidal $\widehat{\gl}_1$ algebra relations of Def.~\ref{qtordef}.
\end{thm}

We will prove this theorem in several steps in the following sections.

Note that $\psi^\pm_n$ are explicit symmetric polynomials of $x_1,...,x_N$, with $\psi_0^+=\psi_0^-=1$, and in particular 
\begin{eqnarray}
\psi_1^+&=&(t^{\frac{1}{2}}-t^{-\frac{1}{2}})(q^{\frac{1}{2}}t^{-\frac{1}{2}}-q^{-\frac{1}{2}}t^{\frac{1}{2}})
(x_1+x_2+\cdots+x_{N})\nonumber \\
\psi_1^-&=&(t^{\frac{1}{2}}-t^{-\frac{1}{2}})(q^{\frac{1}{2}}t^{-\frac{1}{2}}-q^{-\frac{1}{2}}t^{\frac{1}{2}})(x_1^{-1}+x_2^{-1}+\cdots+x_{N}^{-1})\label{psiones}
\end{eqnarray}

The finite number $N$ of variables clearly imposes algebraic relations between the coefficients of the series $\psi^\pm$.
The Macdonald currents generate a quotient of the quantum toroidal algebra by the relations \eqref{elipquo}
and those relations.

\subsubsection{The dual Whittaker limit $t\to\infty$}
Before proceeding to the proof of Theorem \ref{gentoro}, we make the connection with our previous work in \cite{DFK16} regarding the quantum Q-system generators.
In that paper, we showed that the $A_{N-1}$ quantum $Q$-system algebra can be realized as a
quotient, depending on $N$, of the algebra of nilpotent currents in the quantum enveloping algebra of $\widehat{\mathfrak sl}_2$, 
$U_{\sqrt{q}}({\mathfrak n}[u,u^{-1}])$. Furthermore, we were able to define currents ${\mathfrak e}$, ${\mathfrak f}$
in terms of fundamental quantum $Q$-system solutions which satisfied non-standard $U_{\sqrt{q}}(\widehat{\mathfrak sl}_2)$
with truncated Cartan currents. The functional representation of these was constructed
in terms of the degenerate generalized Macdonald operators, corresponding to the so-called dual Whittaker
($t\to\infty$) limit of the $t$-deformation considered in this paper.

These observations become more transparent when we take the $t\to \infty$ limit of Theorem \ref{gentoro}.
Note that it is crucial that $N$ be finite in order for this limit to make sense. 
Defining the limits of the renormalized currents:
\begin{equation}
\label{limefpsi}
{\mathfrak e}^{(\infty)}(u):= \lim_{t\to\infty} t^{-\frac{N-1}{2}} \, {\mathfrak e}(u),\quad {\mathfrak f}^{(\infty)}(u):= \lim_{t\to\infty} t^{-\frac{N-1}{2}} \, {\mathfrak f}(u), \quad \psi^{\pm(\infty)}(z):= \lim_{t\to\infty} t^{-N} \psi^{\pm}(z),
\end{equation}
and $g^{(\infty)}(z,w):=\lim_{t\to\infty} \frac{-qt^{-1}}{z w}\, g(z,w)= z-q w$, we find that the limiting Cartan currents become:
\begin{equation}\label{psilim}
\psi^{\pm(\infty)}(z)=(-q^{-\frac{1}{2}}z^{\pm 1})^N\, \prod_{i=1}^{N} 
\frac{(x_i)^{\pm 1}}{(1-q^{-\frac{1}{2}} (z x_i)^{\pm 1})(1-q^{\frac{1}{2}} (z x_i)^{\pm 1})}  ,
\end{equation}
whereas the other relations reduce to those of $U_{\sqrt{q}}(\widehat{\mathfrak sl}_2)$ (the Serre relation is 
automatically satisfied). Note that the Cartan currents have a non-standard valuation of $z^{\pm N}$ in this limit, and vanish when $N\to\infty$, whereas in the Drinfeld presentation of the quantum affine algebra, the zero modes of the Cartan currents are invertible elements.

\subsection{Proof of Theorem \ref{gentoro}}

For readability, we decompose the proof of the Theorem into several steps, Theorems \ref{efrela}, \ref{commuef}, 
\ref{commuefpsi} and \ref{thserrre} below.

\begin{thm}\label{efrela}
We have the following relations:
\begin{eqnarray*}
g(z,w){\mathfrak e}(z){\mathfrak e}(w)+g(w,z){\mathfrak e}(w){\mathfrak e}(z)&=&0\\
g(w,z){\mathfrak f}(z){\mathfrak f}(w)+g(z,w){\mathfrak f}(w){\mathfrak f}(z)&=&0
\end{eqnarray*}
\end{thm}
\begin{proof}
We start with the computation of
\begin{eqnarray*}
\frac{(1-q)^2}{q}\, {\mathfrak e}(z){\mathfrak e}(w)&=&\sum_{i=1}^{r+1} \delta(q^{\frac{1}{2}}z x_i)\prod_{k\neq i}\frac{\theta x_i-\theta^{-1}x_k}{x_i-x_k}\Gamma_i \sum_{j=1}^{r+1} \delta(q^{\frac{1}{2}}w x_j)\prod_{k\neq j}\frac{\theta x_j-\theta^{-1}x_k}{x_j-x_k}\Gamma_j\\
&=&\sum_{1\leq i\neq j\leq r+1}\delta(q^{\frac{1}{2}}z x_i)\delta(q^{\frac{1}{2}}w x_j)\frac{\theta x_i-\theta^{-1}x_j}{x_i-x_j}
\frac{\theta x_j-q\theta^{-1}x_i}{x_j-qx_i}\\
&&\qquad\qquad \times
\prod_{k\neq i,j} \frac{\theta x_i-\theta^{-1}x_k}{x_i-x_k}\frac{\theta x_j-\theta^{-1}x_k}{x_j-x_k}\Gamma_i\Gamma_j\\
&&\quad +\sum_{i=1}^{r+1}\delta(q^{\frac{1}{2}}z x_i)\delta(q^{\frac{3}{2}} w x_i)\prod_{k\neq i} \frac{\theta x_i-\theta^{-1}x_k}{x_i-x_k}\frac{q\theta x_i-\theta^{-1}x_k}{q x_i-x_k}\Gamma_i^2.\\
\end{eqnarray*}
The second term is proportional to $\delta(z/(q w))$, and since $(z-q w)\delta(z/(q w))=0$, we have:
\begin{eqnarray}
&&\frac{(1-q)^2}{q}\,(z-q w){\mathfrak e}(z){\mathfrak e}(w)=\sum_{1\leq i\neq j\leq r+1} \delta(q^{\frac{1}{2}}z x_i)\delta(q^{\frac{1}{2}}w x_j)\times \nonumber \\
&& \qquad \times \frac{(z-qw)(\theta x_i-\theta^{-1}x_j)(\theta x_j-q\theta^{-1}x_i)}{(x_i-x_j)(x_j-qx_i)}
\prod_{k\neq i,j} \frac{(\theta x_i-\theta^{-1}x_k)(\theta x_j-\theta^{-1}x_k)}{(x_i-x_k)(x_j-x_k)}\Gamma_i\Gamma_j
\nonumber \\
&&\qquad \qquad = \frac{(z-tw)(z-qt^{-1}w)}{z-w}\sum_{1\leq i<j\leq r+1} (\delta(q^{\frac{1}{2}}z x_i)\delta(q^{\frac{1}{2}}w x_j)+\delta(q^{\frac{1}{2}}z x_j)\delta(q^{\frac{1}{2}}w x_i))\times \nonumber \\
&&\qquad \qquad \times
\prod_{k\neq i,j}\frac{(\theta x_i-\theta^{-1}x_k)(\theta x_j-\theta^{-1}x_k)}{(x_i-x_k)(x_j-x_k)}\Gamma_i\Gamma_j
\label{pree}\end{eqnarray}
This makes the quantity $(z-q w)(z-t^{-1}w)(z-t q^{-1}w){\mathfrak e}(z){\mathfrak e}(w)$ 
manifestly skew-symmetric under the interchange $z\leftrightarrow w$.
The relation for ${\mathfrak f}$ follows by taking $(q,t)\to (q^{-1},t^{-1})$, under which $g(z,w)\to g(w,z)$.
\end{proof}
This shows that the fundamental generalized Macdonald currents satisfy the first two relations of the quantum toroidal algebra. We now consider the third relation.

\begin{thm}\label{commuef}
We have the commutation relation between Macdonald currents:
$$[{\mathfrak e}(z),{\mathfrak f}(w)]=\frac{\delta(z/w)}{g(1,1)}\, (\psi^+(z)-\psi^-(z))$$
where $\psi^\pm(z)$ are the power series \eqref{defpsipm} of $z^{\pm 1}$.
\end{thm}
\begin{proof}
Let us first compute
\begin{eqnarray*}
(1-q)(1-q^{-1}){\mathfrak e}(z){\mathfrak f}(w)&=& 
\sum_{i=1}^{r+1} \delta(q^{\frac{1}{2}}z x_i)\prod_{k\neq i}\frac{\theta x_i-\theta^{-1}x_k}{x_i-x_k}\Gamma_i \,
\sum_{j=1}^{r+1} \delta(q^{-\frac{1}{2}}w x_j)\prod_{k\neq j}\frac{\theta^{-1} x_j-\theta x_k}{x_j-x_k}\Gamma_j^{-1}\\
&=&\sum_{1\leq i\neq j\leq r+1} \delta(q^{\frac{1}{2}}z x_i)\delta(q^{-\frac{1}{2}}w x_j)\frac{\theta x_i-\theta^{-1}x_j}{x_i-x_j}
\frac{\theta^{-1} x_j-q\theta x_i}{x_j-qx_i}\\
&&\qquad\qquad\qquad \times \prod_{k\neq i,j} \frac{(\theta x_i-\theta^{-1}x_k)(\theta^{-1} x_j-\theta x_k)}{(x_i-x_k)(x_j-x_k)}\Gamma_i\Gamma_j^{-1}\\
&&\quad +\sum_{i=1}^{r+1} \delta(q^{\frac{1}{2}}z x_i)\delta(q^{\frac{1}{2}} w x_i)\prod_{k\neq i}\frac{\theta x_i-\theta^{-1}x_k}{x_i-x_k}\frac{\theta^{-1} q x_i-\theta x_k}{q x_i-x_k},
\end{eqnarray*}
and
\begin{eqnarray*}
(1-q)(1-q^{-1}){\mathfrak f}(w){\mathfrak e}(z)&=& \sum_{j=1}^{r+1} \delta(q^{-\frac{1}{2}}w x_j)\prod_{k\neq j}\frac{\theta^{-1} x_j-\theta x_k}{x_j-x_k}\Gamma_j^{-1}\sum_{i=1}^{r+1} \delta(q^{\frac{1}{2}}z x_i)\prod_{k\neq i}\frac{\theta x_i-\theta^{-1}x_k}{x_i-x_k}\Gamma_i\\
&=&\sum_{1\leq i\neq j\leq r+1} \delta(q^{\frac{1}{2}}z x_i)\delta(q^{-\frac{1}{2}}w x_j)\frac{\theta^{-1} x_j-\theta x_i}{x_j-x_i}\frac{\theta x_i-q^{-1}\theta^{-1}x_j}{x_i-q^{-1}x_j}\\
&&\qquad\qquad\qquad \times 
\prod_{k\neq i,j} \frac{(\theta x_i-\theta^{-1}x_k)(\theta^{-1} x_j-\theta x_k)}{(x_i-x_k)(x_j-x_k)}\Gamma_i\Gamma_j^{-1}\\
&&\quad +\sum_{i=1}^{r+1} \delta(q^{-\frac{1}{2}}z x_i)\delta(q^{-\frac{1}{2}}w x_i)\prod_{k\neq i}\frac{q^{-1}\theta x_i-\theta^{-1}x_k}{q^{-1}x_i-x_k}\frac{\theta^{-1} x_i-\theta x_k}{x_i-x_k}.
\end{eqnarray*}
In the commutator $[{\mathfrak e}(z),{\mathfrak f}(w)]$, the terms corresponding to $i\neq j$ cancel out, 
while the remaining terms are proportional to $\delta(z/w)$, so that we are left
with:
\begin{eqnarray*}
[{\mathfrak e}(z),{\mathfrak f}(w)]&=&\frac{\delta(z/w)}{(1-q)(1-q^{-1})}\sum_{i=1}^{r+1}\left\{
 \delta(q^{\frac{1}{2}}z x_i)\prod_{k\neq i}\frac{\theta x_i-\theta^{-1}x_k}{x_i-x_k}\frac{\theta^{-1} q x_i-\theta x_k}{q x_i-x_k}\right.\\
&&\qquad\qquad  - \left. \delta(q^{-\frac{1}{2}}z x_i)\prod_{k\neq i}\frac{q^{-1}\theta x_i-\theta^{-1}x_k}{q^{-1}x_i-x_k}\frac{\theta^{-1} x_i-\theta x_k}{x_i-x_k}\right\}.
\end{eqnarray*}
Recalling that $\delta(z)=\sum_{n\in\Z} z^n$, this is in agreement with the partial fraction decomposition of \eqref{defpsipm},
which reads:
\begin{eqnarray}
&&\qquad \quad \psi^+(z)=1+\frac{g(1,1)}{(1-q)(1-q^{-1})}\sum_{i=1}^{r+1}\left\{\frac{1}{1-q^{\frac{1}{2}}z x_i}\prod_{k\neq i}\frac{\theta x_i-\theta^{-1}x_k}{x_i-x_k}\frac{\theta^{-1} q x_i-\theta x_k}{q x_i-x_k}\right. \label{psiplusX}\\
&&\qquad\qquad\qquad\qquad\qquad \left. - \frac{1}{1-q^{-\frac{1}{2}}z x_i}\prod_{k\neq i}\frac{q^{-1}\theta x_i-\theta^{-1}x_k}{q^{-1}x_i-x_k}\frac{\theta^{-1} x_i-\theta x_k}{x_i-x_k}\right\}\nonumber  ,\\
&&\qquad \quad  \psi^-(z)=1-\frac{g(1,1)}{(1-q)(1-q^{-1})}\sum_{i=1}^{r+1}\left\{\frac{1}{1-q^{-\frac{1}{2}}z^{-1} x_i^{-1}}\prod_{k\neq i}\frac{\theta x_i-\theta^{-1}x_k}{x_i-x_k}\frac{\theta^{-1} q x_i-\theta x_k}{q x_i-x_k}\right.\label{psiminusX}\\
&&\qquad\qquad\qquad\qquad\qquad \left. - \frac{1}{1-q^{\frac{1}{2}}z^{-1} x_i^{-1}}\prod_{k\neq i}\frac{q^{-1}\theta x_i-\theta^{-1}x_k}{q^{-1}x_i-x_k}\frac{\theta^{-1} x_i-\theta x_k}{x_i-x_k}\right\}\nonumber ,
\end{eqnarray}
as power series of $z$ and $z^{-1}$, respectively. Here, we have identified the constant
$g(1,1)/((1-q)(1-q^{-1}))=(1-t^{-1})(1-t q^{-1})/(1-q^{-1})$.
%To show this, we note that e.g. $\psi^+(z)$
%is a rational fraction of $z$ with single poles at $z=q^{\pm \frac{1}{2}}x_i^{-1}$, $i=1,2,...,r$, and is bounded at $z\to\infty$,
%so it has the form: $\psi^+(z)=P(z)/Q(z)$, with $P$ a polynomial of degree at most $2r+2$,
%and $Q(z)=\prod_{i=1}^{r+1}(1-q^{-\frac{1}{2}}z x_i)(1-q^{\frac{1}{2}}z x_i)$. We then check that $\psi^+(z)$ vanishes
%at $z=(q\theta^{-2})^{\pm \frac{1}{2}}x_i^{-1}$, $i=1,2,...,r+1$. Finally, $P$ is entirely fixed by its zeros and the value
%$P(0)=\psi^+_0$.  
%The values \eqref{psiones}
%are easily obtained by expanding (\ref{psiplus}-\ref{psiminus}) at first order in $z$, $z^{-1}$
%respectively.
\end{proof}

We proceed with the next two relations of the level $(0,0)$ quantum toroidal algebra.

\begin{thm}\label{commuefpsi}
We have the following relations between the Macdonald currents and the series \eqref{defpsipm}:
\begin{eqnarray*}
g(z,w)\psi^\pm(z)\,{\mathfrak e}(w)+g(w,z){\mathfrak e}(w)\,\psi^\pm(z)&=&0\\
g(w,z)\psi^\pm(z)\,{\mathfrak f}(w)+g(z,w){\mathfrak f}(w)\,\psi^\pm(z)&=&0
\end{eqnarray*}
\end{thm}
\begin{proof} By explicit calculation,
\begin{eqnarray*}
\frac{1-q}{q^{\frac{1}{2}}}\psi^+(z)\,{\mathfrak e}(w)&=&
\prod_{j=1}^{r+1} 
\frac{(1-q^{-\frac{1}{2}}\theta^2 z x_j)(1-q^{\frac{1}{2}}\theta^{-2} z x_j)}{(1-q^{-\frac{1}{2}} z x_j)(1-q^{\frac{1}{2}} z x_j)} 
\sum_{i=1}^{r+1} \delta(q^{\frac{1}{2}}w x_i)\prod_{k\neq i}\frac{\theta x_i-\theta^{-1}x_k}{x_i-x_k}\Gamma_i \\
&=&\sum_{i=1}^{r+1} \delta(q^{\frac{1}{2}}w x_i)\frac{(w-q^{-1}t z)(w-t^{-1} z)}{(w-q^{-1} z)(w-z)}\\
&&\qquad \times
\prod_{k\neq i}\frac{(1-q^{-\frac{1}{2}}\theta^2 z x_k)(1-q^{\frac{1}{2}}\theta^{-2} z x_k)}{(1-q^{-\frac{1}{2}} z x_k)(1-q^{\frac{1}{2}} z x_k)}\frac{\theta x_i-\theta^{-1}x_k}{x_i-x_k}\Gamma_i ,
\end{eqnarray*}
\begin{eqnarray*}
\frac{1-q}{q^{\frac{1}{2}}}{\mathfrak e}(w)\,\psi^+(z)&=&
\sum_{i=1}^{r+1} \delta(q^{\frac{1}{2}}w x_i)\prod_{k\neq i}\frac{\theta x_i-\theta^{-1}x_k}{x_i-x_k}\Gamma_i \prod_{j=1}^{r+1} 
\frac{(1-q^{-\frac{1}{2}}\theta^2 z x_j)(1-q^{\frac{1}{2}}\theta^{-2} z x_j)}{(1-q^{-\frac{1}{2}} z x_j)(1-q^{\frac{1}{2}} z x_j)}  \\
&=&\sum_{i=1}^{r+1} \delta(q^{\frac{1}{2}}w x_i)
\frac{(w-t z)(w-q t^{-1} z)}{(w- z)(w-q z)}\\
&&\qquad \times
\prod_{k\neq i}\frac{\theta x_i-\theta^{-1}x_k}{x_i-x_k}
\frac{(1-q^{-\frac{1}{2}}\theta^2 z x_k)(1-q^{\frac{1}{2}}\theta^{-2} z x_k)}{(1-q^{-\frac{1}{2}} z x_k)(1-q^{\frac{1}{2}} z x_k)}\Gamma_i .
\end{eqnarray*}
Using
$$ g(z,w)\frac{(w-q^{-1}t z)(w-t^{-1} z)}{(w-q^{-1} z)(w-z)}=-g(w,z)\frac{(w-t z)(w-q t^{-1} z)}{(w- z)(w-q z)}$$
we see that $g(z,w)\psi^+(z)\,{\mathfrak e}(w)+g(w,z){\mathfrak e}(w)\,\psi^+(z)=0$. The derivation of the $\psi^-$
equation follows a similar calculation.
\end{proof}

The final two relations are the cubic, Serre-type relations for the fundamental currents.
\begin{thm}\label{thserrre}
The fundamental currents satisfy the cubic relations
\begin{eqnarray}
{\rm Sym}_{z_1,z_2,z_3}\left( \frac{z_2}{z_3} {\Big[}{\mathfrak e}(z_1),{[}{\mathfrak e}(z_2),{\mathfrak e}(z_3){]}{\Big]}\right)&=&0  \label{serre1}, \\
{\rm Sym}_{z_1,z_2,z_3}\left( \frac{z_2}{z_3} {\Big[}{\mathfrak f}(z_1),{[}{\mathfrak f}(z_2),{\mathfrak f}(z_3){]}{\Big]}\right)&=&0 \label{serre2}.
\end{eqnarray}
\end{thm}
\begin{proof}
It is sufficient to prove the statement of the theorem for
${\mathfrak e}$, eq. \eqref{serre1}, as
that for ${\mathfrak f}$ follows by the substitution $(q,t)\to (q^{-1},t^{-1})$.
The proof of \eqref{serre1} is straightforward but extremely tedious. 
We prefer to postpone it to Section \ref{shuprosec}, where it reduces
to a suffle product identity (Theorem \ref{shufserre}) which is considerably easier to prove.
\end{proof}
We conclude that the currents ${\mathfrak e}(z)$, ${\mathfrak f}(z)$, $\psi^\pm(z)$ satisfy all the relations of the
level $(0,0)$ quantum toroidal algebra ${\mathfrak gl}_1$, and Theorem \ref{gentoro} follows.

\section{Plethysms and bosonization}\label{bososec}
\label{secboso}
%\color{red}
%Problem of the non-independence of the $p_k$ as $N$ is finite.
%\color{black}
In this section, we show that the action of the currents $\e(z),\f(z)$ on the space of symmetric functions in $N$ variables can be expressed in terms of plethysms, which act naturally on functions expressed in terms of the power sum symmetric functions. In the infinite rank limit $N\to\infty$, the power sum symmetric functions become algebraically independent, and together with derivatives with respect to these functions, they generate an infinite-dimensional Heisenberg algebra. Thus the plethystic expression becomes, in this limit, a formula for the bosonization of the action of the currents on the space of symmetric functions. This formulation will be compared to recent work of Bergeron et al \cite{BGLX}. 

%In this section, we present a bosonization of the currents
%${\mathfrak e(z)}$ and ${\mathfrak f(z)}$ of \eqref{highercur} with $\al=1$, introduced in the previous section. 
%The idea is to express these currents as functions of two types of operators: those of multiplication by the power sum symmetric functions,
%$p_k=\sum_{i=1}^N x_i^k$ for $k\in \Z^*$, and those of differentiation with respect to $p_k$.  In the limit where $N\to\infty$, these generate of a Heisenberg algebra, hence the term bosonization. 
%When $N$ is finite, $\mathfrak e(z)$ and ${\mathfrak f(z)}$ are expressed in terms of the power sum symmetric functions with a finite number of variables.

\subsection{Power sum action and plethysms}

Let $\{p_k\}_{k\in \Z\setminus\{0\}}$ be an infinite set of algebraically independent variables, and consider 
the space $\mathcal P$ of functions that can be expressed as formal
power series of the $p_k$'s. For any collection $X$ of variables, e.g. $X=(x_1,x_2,...,x_N)$,
we may evaluate functions in $\mathcal P$ by using the substitution 
$p_k\mapsto p_k[X]:=\sum_{i=1}^N x_i^k$, namely by interpreting the $p_k$'s as power sums of the $x$'s, for $k\in \Z^*$, in which case any $F\in {\mathcal P}$ may be evaluated 
as $F[X]$ by taking all $p_k= p_k[X]$ in its power series. If the number of variables $N$ is finite, the resulting space is a quotient of the original space as the functions $p_k[X]$s are not algebraically independent. Note that $F[X]$ is a symmetric function. 

Given the collection $X$, we
%Consider the subspace ${\mathcal P}_N\subset \mathcal{S}_N$ consisting of symmetric functions in the variables $\{x_1,...,x_N\}$ which can be expressed as power series of the {\it power sum} variables 
%$p_k=\sum_{i=1}^N x_i^k$, $k\in \Z^*$. (When there is a finite number $N$ of variables, the power sum symmetric functions are not all algebraically independent.)
 consider the action of the generalized Macdonald operators ${\mathcal D}^{q,t}_{1;n}$ and
${\mathcal D}^{q^{-1},t^{-1}}_{1;n}$ on the space of functions of the form $F[X]$, for $F\in {\mathcal P}$, which we denote by ${\mathcal P}[X]$. First, we present some commutation relations which hold for arbitrary $N$.

\begin{lemma}\label{compk}
For any $N$ and $i\leq N$, and given the collection $X=(x_1,x_2,...,x_N)$, the following commutation relations hold:
%\begin{equation}
$$
[p_k[X], {\mathcal D}^{q,t}_{1;n}]=(1-q^k) \, {\mathcal D}^{q,t}_{1;n+k},\quad {\rm and}\quad 
[p_k[X], {\mathcal D}^{q^{-1},t^{-1}}_{1;n}]=(1-q^{-k}) \, {\mathcal D}^{q^{-1},t^{-1}}_{1;n+k}.
%\end{equation}
$$
\end{lemma}
\begin{proof} By direct computation, using $[p_k[X],\Gamma_i]=(1-q^k)\,x_i^k\,\Gamma_i$, with $\Gamma_i$ as in \eqref{shift}.\end{proof}

\begin{cor}\label{efpleth} For any $N$, $k\in \Z^*$, and $X$ as above,
\begin{eqnarray*}
{\mathfrak e}(z)\, p_k[X]&=&\left( p_k[X]+\frac{q^{k/2}-q^{-k/2}}{z^k} \right) {\mathfrak e}(z),\\
{\mathfrak f}(z)\, p_k[X]&=&\left( p_k[X]-\frac{q^{k/2}-q^{-k/2}}{z^k} \right) {\mathfrak f}(z).
\end{eqnarray*}
\end{cor}

%Note here that despite the fact that all $p_k$, $k\in \Z^*$ are not independent (as the number of variables $x_i$ is finite),
%the commutation relations above are valid for {\it all} $k\in \Z^*$.
Therefore, up to a scalar multiple coming from the action of the currents on the constant function 1, the action of the currents ${\mathfrak e}$, ${\mathfrak f}$ on ${\mathcal P}[X]$
is by substitutions on all the power sums of the form $p_k[X]\mapsto p_k[X]+\mu_k$ for some specific sequences $\mu_k$.

Such substitutions are {\em plethysms} (see the notes \cite{Haiman} for a detailed exposition). In $\lambda$-ring notations, one writes as above
$X=(x_1,x_2,...)$ for the collection of variables (an alphabet), for which power sums are $p_k[X]=\sum_i x_i^k$.
For two alphabets $X$, $Y$, the sum $X+Y$ refers to the concatenation of the two alphabets, with power
sums $p_k[X+Y]=p_k[X]+p_k[Y]$, while for any scalar $\lambda$ we have $p_k[\lambda X]=\lambda^k p_k[X]$. 
Note finally that for a single variable alphabet $\mu$, we have $p_k[X+\mu]=p_k[X]+\mu^k$.

In the plethystic notation, Corollary \ref{efpleth} can be written as
\begin{equation}\label{plet} {\mathfrak e}(z) \,F[X]= F\left[X+\frac{q^{1/2}-q^{-1/2}}{z}\right]\,\e(z),\qquad 
 {\mathfrak f}(z)\,F[X]=F\left[X-\frac{q^{1/2}-q^{-1/2}}{z}\right]\,\f(z)  \end{equation}
for any $F\in {\mathcal P}$, 
where the alphabet $X$ is extended by two variables.
 
%The action of the currents on the right is fixed up to a multiplicative factor, the action on the constant function $1$. 
Define the two functions in ${\mathcal P}[X]$
\begin{eqnarray}
C(z)&=&\prod_{i=1}^{N}(1-z \,x_i)= e^{-\sum_{k=1}^\infty p_k[X] \frac{z^k}{k} }\label{Cdef} ,\\
{\widetilde C}(z)&=&\prod_{i=1}^{N}(1-z \,x_i^{-1})=e^{-\sum_{k=1}^\infty p_{-k}[X] \frac{z^k}{k} }\label{Ctdef}.
\end{eqnarray}
In terms of these functions we can write
\begin{thm}\label{pletbo}
The currents ${\mathfrak e}$ and ${\mathfrak f}$ act as the following plethysms on symmetric functions $F[X]\in {\mathcal P}[X]$:
\begin{eqnarray*}
{\mathfrak e}(z)\, F[X] &=&\frac{q^{1/2}}{1-q} \frac{t^{-1/2}}{1-t^{-1}}\left\{ t^{\frac{N}{2}}\frac{C(q^{1/2}t^{-1}z)}{C(q^{1/2}z)}-t^{-\frac{N}{2}}\frac{{\widetilde C}(q^{-1/2}tz^{-1})}{{\widetilde C}(q^{-1/2}z^{-1})}\right\}
F\left[ X+\frac{q^{1/2}-q^{-1/2}}{z}\right],\\
{\mathfrak f}(z)\, F[X] &=&\frac{q^{-1/2}}{1-q^{-1}} \frac{t^{1/2}}{1-t}\left\{t^{-\frac{N}{2}}\frac{C(q^{-1/2}tz)}{C(q^{-1/2}z)}- t^{\frac{N}{2}}\frac{{\widetilde C}(q^{1/2}t^{-1}z^{-1})}{{\widetilde C}(q^{1/2}z^{-1})}\right\}
F\left[ X-\frac{q^{1/2}-q^{-1/2}}{z}\right].
\end{eqnarray*}
\end{thm}
\begin{proof}
The plethystic part of the action was derived in \eqref{plet}. To fix the overall factors, apply the
currents to the constant $F[X]=1$. For example,
\begin{eqnarray*}
{\mathfrak e}(z)\cdot 1&=& \frac{q^{1/2}}{1-q} \sum_{i=1}^N \delta(q^{1/2}z x_i) \prod_{j\neq i} \frac{\theta x_i -\theta^{-1} x_j}{x_i-x_j}\\
&=&\frac{q^{1/2}}{1-q}\sum_{i=1}^N \left\{ t^{\frac{N-1}{2}} \frac{1}{1-q^{1/2}z x_i} \prod_{j\neq i} \frac{ x_i -t^{-1} x_j}{x_i-x_j}
+t^{-\frac{N-1}{2}} \frac{1}{1-q^{-1/2}z^{-1} x_i^{-1}} \prod_{j\neq i} \frac{ t x_i -x_j}{x_i-x_j}\right. \\
&&\qquad \qquad \qquad \qquad \qquad \qquad \qquad -\left.
\prod_{j\neq i} \frac{\theta x_i -\theta^{-1} x_j}{x_i-x_j}\right\}.
\end{eqnarray*}

First, we have
\begin{equation}\label{firstcn}
\sum_{i=1}^N\prod_{j\neq i} \frac{\theta x_i -\theta^{-1} x_j}{x_i-x_j} =\frac{t^{\frac{N-1}{2}}-t^{-\frac{N+1}{2}}}{1-t^{-1}}
=-\left\{  \frac{t^{-\frac{N+1}{2}}}{1-t^{-1}}+ \frac{t^{\frac{N+1}{2}}}{1-t}\right\},
\end{equation}
by noting that the left hand side of this equation is a symmetric rational function, whose common denominator is the Vandermonde determinant
$\prod_{i<j}x_i-x_j$, and which has
total degree $0$, hence must be a constant $c_N$. The constant can be found, for example, by taking the limit $x_1\to\infty$:
$$c_N=\theta^{N-1}+\theta^{-1} \sum_{i=2}^{N} \prod_{j\neq i\atop j>1} \frac{\theta x_i -\theta^{-1} x_j}{x_i-x_j}=
\theta^{N-1}+\theta^{-1} c_{N-1} ,$$
which, together with $c_0=0$, yields $c_N=\frac{\theta^N-\theta^{-N}}{\theta-\theta^{-1}}$, and \eqref{firstcn} follows.

Next, by simple fraction decomposition,
$$\frac{C(q^{1/2}t^{-1}z)}{C(q^{1/2}z)} =\prod_{i=1}^N \frac{1-q^{1/2}t^{-1}zx_i}{1-q^{1/2}zx_i}=t^{-N}+(1-t^{-1}) 
\sum_{i=1}^N \frac{1}{1-q^{1/2}z x_i} \prod_{j\neq i} \frac{ x_i -t^{-1} x_j}{x_i-x_j}.$$
Therefore
$$\frac{t^{\frac{N-1}{2}}}{1-t^{-1}} \frac{C(q^{1/2}t^{-1}z)}{C(q^{1/2}z)}=
\frac{t^{-\frac{N+1}{2}}}{1-t^{-1}}+t^{\frac{N-1}{2}}
\sum_{i=1}^N \frac{1}{1-q^{1/2}z x_i} \prod_{j\neq i} \frac{ x_i -t^{-1} x_j}{x_i-x_j}.$$

By using the change of variables $(q,t,z,x_i)\to (q^{-1},t^{-1},z^{-1},x_i^{-1})$ we get:
$$\frac{t^{-\frac{N-1}{2}}}{1-t} \frac{{\widetilde C}(q^{-1/2}tz^{-1})}{{\widetilde C}(q^{-1/2}z^{-1})}=
\frac{t^{\frac{N+1}{2}}}{1-t}+t^{-\frac{N-1}{2}}\sum_{i=1}^N\frac{1}{1-q^{-1/2}z^{-1} x_i^{-1}} 
\prod_{j\neq i} \frac{ t x_i -x_j}{x_i-x_j}.$$
Summing the two above contributions yields the full action of ${\mathfrak e}(z)$ on $1$. 
That of ${\mathfrak f}(z)$
follows immediately by sending $(q,t)\to (q^{-1},t^{-1})$.
\end{proof}

Moreover, the Cartan currents act as scalars on symmetric functions of $x_1,x_2,...,x_N$:
$$
\psi^{+}(z)=\frac{C(q^{-1/2}t z)C(q^{1/2}t^{-1} z)}{C(q^{-1/2} z)C(q^{1/2} z)} , \quad 
\psi^{-}(z)=\frac{{\widetilde C}(q^{-1/2}t z^{-1}){\widetilde C}(q^{1/2}t^{-1} z^{-1})}{{\widetilde C}(q^{-1/2} z^{-1}){\widetilde C}(q^{1/2} z^{-1})}\ .
$$

\begin{remark}
The plethystic formulas above for ${\mathfrak e}$ and $-{\mathfrak f}$, as well as $\psi^+$ and $\psi^-$, 
are exchanged under the involution $x_i\mapsto 1/x_i$
for all $i$, and $z\mapsto z^{-1}$, under which $C(z)\mapsto {\widetilde C}(z^{-1})$, 
and the one-variable plethysm $[X+\mu]\mapsto [X+\mu^{-1}]$. This is in agreement with Remark \ref{remef} for $\al=1$.
\end{remark}

\subsection{Bosonization formulas in the $N\to\infty$ limit}
As mentioned at the beginning of the section, in the infinite rank $N\to\infty$ limit the  power sum symmetric functions $p_k$ become algebraically independent for all $k\in \Z\setminus\{0\},$ and together with the derivatives with respect to these functions, form a Heisenberg algebra.

\subsubsection{Plethystic formulas for the $N\to\infty$ limit}

The limit of an infinite number of variables is obtained by taking $N\to\infty$, whereas the collection $X=(x_1,x_2,...)$ becomes infinite. 
In view of the plethystic formulas of Theorem \ref{pletbo}, it is natural to define:
${\mathfrak e}_\infty(z) :=\lim_{N\to \infty} t^{\frac{1-N}{2}}\, {\mathfrak e}(z)$, where we implicitly assumed that $|t|>1$.
Accordingly, we define the limiting functions:
$$C_\infty(z):=\prod_{i=1}^\infty (1-z x_i),\qquad 
{\widetilde C}_\infty(z):=\prod_{i=1}^\infty (1-z x_i^{-1})$$
and therefore the limiting Cartan current:
\begin{equation}
\label{psip}
\psi_\infty^{+}(z)=\frac{C_\infty(q^{-1/2}t z)C_\infty(q^{1/2}t^{-1} z)}{C_\infty(q^{-1/2} z)C_\infty(q^{1/2} z)}
\end{equation}
which is a power series of $z$.

To define the limiting ${\mathfrak f}$ current,
recall there are two equivalent ways of defining the current ${\mathfrak f}(z)$ in terms of ${\mathfrak e}(z)$ for
finite number of variables $N$:
$$(1)\ {\mathfrak f}(z)=-S{\mathfrak e}(z^{-1})S\qquad (2) \  {\mathfrak f}(z)={\mathfrak e}(z)\vert_{q\to q^{-1},t\to t^{-1}}$$
where we used (2) as the original definition, and (1) comes from Remark \ref{remef}.
It turns out that applying these to ${\mathfrak e}_\infty(z)$ leads to two different definitions for ${\mathfrak f}_\infty(z)$.

Indeed, it is easy to see that (1) leads to 
${\mathfrak f}^{(1)}_\infty(z) :=\lim_{N\to \infty} t^{\frac{1-N}{2}}\, {\mathfrak f}(z)$,
where it is still assumed that $|t|>1$. However, this naive choice leads to a trivial commutation relation, as:
$$[{\mathfrak e}_\infty(z),{\mathfrak f}^{(1)}_\infty(w)]=
\lim_{N\to\infty} t^{1-N}\, [{\mathfrak e}(z),{\mathfrak f}(w)]=\lim_{N\to\infty} t^{1-N}\, \frac{\delta(z/w)}{g(1,1)}(\psi^+(z)-\psi^-(z))=0$$
as $\psi^\pm$ tend to well-defined infinite products.

We shall therefore rather use the second definition (2), and write
${\mathfrak f}_\infty(z)={\mathfrak e}_\infty(z)\vert_{q\to q^{-1},t\to t^{-1}}$.

%By taking the $N\to\infty$ limit of the result of Theorem \ref{pletbo}, we obtain the following:
\begin{thm}\label{bobothm}
The limiting currents ${\mathfrak e}_\infty(z),{\mathfrak f}_\infty(z)$ act on functions 
$F[X]$ as:
\begin{eqnarray*}
{\mathfrak e}_\infty(z)\, F[X] 
&=&\frac{q^{1/2}}{(1-q)(1-t^{-1})} \,\frac{C_\infty(q^{1/2}t^{-1}z)}{C_\infty(q^{1/2}z)} \,
F\left[ X+\frac{q^{1/2}-q^{-1/2}}{z}\right],\\
{\mathfrak f}_\infty(z)\, F[X] 
&=&\frac{q^{-1/2}}{(1-q^{-1})(1-t)}\,\frac{C_\infty(q^{-1/2}t z)}{C_\infty(q^{-1/2}z)} \,
F\left[ X-\frac{q^{1/2}-q^{-1/2}}{z}\right].
\end{eqnarray*}
Moreover we have the commutation relations:
\begin{eqnarray*}[{\mathfrak e}_\infty(z),{\mathfrak f}_\infty(w)] &=&\frac{1}{g(1,1)}\left\{ \delta(z/w)
\psi_\infty^+(z)-\delta(q t^{-1} z/w) \psi_\infty^-(z)\right\}\\
{\mathfrak e}_\infty(z)\,{\mathfrak e}_\infty(w)&=&-\frac{g(w,z)}{g(z,w)}\, {\mathfrak e}_\infty(w)\, {\mathfrak e}_\infty(z),\ \ {\mathfrak f}_\infty(z)\,{\mathfrak f}_\infty(w)=-\frac{g(z,w)}{g(w,z)}\, {\mathfrak f}_\infty(w)\, {\mathfrak f}_\infty(z)\\
\psi_\infty^\pm(z)\,{\mathfrak e}_\infty(w)&=& -\frac{g(w,z)}{g(z,w)}\, {\mathfrak e}_\infty(w)\, \psi_\infty^\pm(z), \ \ 
\psi_\infty^+(z)\,{\mathfrak f}_\infty(w)=-\frac{g(z,w)}{g(w,z)} \, {\mathfrak f}_\infty(w)\, \psi_\infty^+(z)\\
\psi_\infty^-(z)\,{\mathfrak f}_\infty(w)&=&-\frac{g(qt^{-1}z,w)}{g(w,qt^{-1}z)} {\mathfrak f}_\infty(w)\, \psi_\infty^-(z), \ \
\psi_\infty^-(z)\,\psi_\infty^+(w)=\frac{g(w,z)g(qt^{-1}z,w)}{g(z,w)g(w,qt^{-1}z)}\, \psi_\infty^+(w)\,\psi_\infty^-(z)
\end{eqnarray*}
where $\psi_\infty^+(z)$ is defined in \eqref{psip}, and  $\psi_\infty^-(z)$ is a power series of $z^{-1}$
acting on symmetric functions $F[X]$ as:
$$\psi_\infty^-(z) F[X]=F\Big[X+\frac{(q^{1/2}-q^{-1/2})(1-t q^{-1})}{z}\Big] $$
In addition, the "Serre relations" of Theorem \ref{thserrre} still hold with 
${\mathfrak e},{\mathfrak f}\to {\mathfrak e}_\infty, {\mathfrak f}_\infty$.
\end{thm}

\begin{remark}
The algebra satisfied by the currents ${\mathfrak e}_\infty(z),{\mathfrak f}_\infty(z)$ and 
$\psi^\pm_\infty(z)\in \C[[z^{\pm 1}]]$ is a particular case of
the general quantum toroidal algebra of Def. \ref{qtorgendef}. 
%The latter has the generating currents $x^{\pm}(z)$ and 
%$\psi^\pm(z)\in \C[[z^{\mp 1}]]$ 
%and two central elements $\hat \gamma$ and $\hat \delta$, subject to 
%the following relations, best expressed in terms of the rational function
%$$G(x)=-\frac{g(1,x)}{g(x,1)}=\frac{(1-q x)(1-t^{-1}x)(1-q^{-1}t x)}{(1-q^{-1} x)(1-t x)(1-qt^{-1} x)}$$
%and read:
%\begin{eqnarray*}
%[\psi^\pm(z),\psi^\pm(w)]&=&0,\qquad \psi^+(z)\,\psi^-(w)=\frac{G({\hat \gamma}w/z)}{G({\hat \gamma}^{-1}w/z)}\,
%\psi^-(w)\,\psi^+(z)\\
%\psi^+(z)\, x^\pm(w)&=&G({\hat \gamma}^{\mp 1}w/z)^{\mp 1}\, x^\pm(w)\,\psi^+(z),\quad 
%\psi^-(z)\, x^\pm(w)=G({\hat \gamma}^{\mp 1}w/z)^{\pm 1}\, x^\pm(w)\,\psi^-(z)\\
%x^\pm(z)\, x^\pm(w)&=&G(z/w)^{\pm 1} \, x^\pm(w)\, x^\pm(z), \qquad \psi_0^\pm={\hat \delta}^{\mp 1},\\
%{[} x^+(z), x^-(w) {] }&=& \frac{(1-q)(1-t^{-1})}{(1-q t^{-1})} \left\{ \delta({\hat \gamma}^{-1}z/w)\psi^+({\hat \gamma}^{-1/2}z)
%-\delta({\hat \gamma} z/w)\psi^-({\hat \gamma}^{1/2}z)\right\} \\
%&&{\rm Sym}_{z_1,z_2,z_3}\left( \frac{z_2}{z_3} {\Big[}{\mathfrak x^\pm}(z_1),{[}{\mathfrak x^\pm}(z_2),{\mathfrak x^\pm}(z_3){]}{\Big]}\right)
%=0
%\end{eqnarray*}
%A particular class of representations \cite{FHHSY} indexed by $(\ell_1,\ell_2)\in \Z^2$ corresponds to diagonal actions of the central elements ${\hat \gamma},{\hat \delta}$
%with respective eigenvalues $\gamma^{\ell_1},\gamma^{\ell_2}$, where $\gamma=(t q^{-1})^{1/2}$.
In this language, the relations of Theorem \ref{bobothm} take place in the so-called horizontal representation
$(\ell_1,\ell_2)=(1,0)$, namely with ${\hat \gamma}=\gamma=(tq^{-1})^{1/2}$ and $\hat \delta=1$,
and with the correspondence (see \eqref{dictio}):
\begin{eqnarray}
&&x^+(z)= \frac{(1-q)(1-t^{-1})}{q^{1/2}}{\mathfrak e}_\infty(q^{-1/2}z),\quad x^-(z)
=\frac{(1-q^{-1})(1-t)}{q^{-1/2}}{\mathfrak f}_\infty(t^{-1/2}z),\nonumber \\ 
&&\varphi^+(z)= \psi_\infty^-(q^{-3/4}t^{1/4}z^),\qquad \varphi^-(z)= \psi_\infty^+(q^{-1/4}t^{-1/4}z)
\label{corresp}
\end{eqnarray}
\end{remark}

Let us now turn to the proof of the theorem.
\begin{proof}
The expressions for ${\mathfrak e}_\infty(z)$ and ${\mathfrak f}_\infty(z)$ follow from Theorem \ref{pletbo}
and the definition (2). To compute the commutator, we must commute the plethysms through the prefactors.
Using:
\begin{eqnarray*}
\frac{C_\infty(q^{-1/2}t w)}{C_\infty(q^{-1/2}w)}\left[ X+\frac{q^{1/2}-q^{-1/2}}{z}\right]&=&
\frac{(1-tw/z)(1-q^{-1}w/z)}{1-w/z)(1-q^{-1}tw/z)}\frac{C_\infty(q^{-1/2}t w)}{C_\infty(q^{-1/2}w)}\\
\frac{C_\infty(q^{1/2}t^{-1}z)}{C_\infty(q^{1/2}z)}\left[ X-\frac{q^{1/2}-q^{-1/2}}{w}\right]&=&
\frac{(1-t^{-1}z/w)(1-qz/w)}{1-z/w)(1-qt^{-1}z/w)}\frac{C_\infty(q^{1/2}t^{-1}z)}{C_\infty(q^{1/2}z)}
\end{eqnarray*}
we arrive at:
\begin{eqnarray*}
[{\mathfrak e}_\infty(z),{\mathfrak f}_\infty(w)]&=& \frac{1}{(1-q)(1-q^{-1})(1-t)(1-t^{-1})}\frac{C_\infty(q^{1/2}t^{-1}z)C_\infty(q^{-1/2}t w)}{C_\infty(q^{1/2}z)C_\infty(q^{-1/2}w)}\\
&&\quad \times \left\{ \frac{(1-tw/z)(1-q^{-1}w/z)}{1-w/z)(1-q^{-1}tw/z)}-\frac{(1-t^{-1}z/w)(1-qz/w)}{1-z/w)(1-qt^{-1}z/w)}\right\}\\
&&\qquad \qquad \times F\left[ X+(q^{1/2}-q^{-1/2})(\frac{1}{z}-\frac{1}{w})\right]\\
&=&\frac{1}{g(1,1)} (\delta(z/w)-\delta(qt^{-1}z/w) )\frac{C_\infty(q^{1/2}t^{-1}z)C_\infty(q^{-1/2}t w)}{C_\infty(q^{1/2}z)C_\infty(q^{-1/2}w)}\\
&&\qquad \qquad \times F\left[ X+(q^{1/2}-q^{-1/2})(\frac{1}{z}-\frac{1}{w})\right]
\end{eqnarray*}
The other relations are obtained in a similar way,
and the theorem follows.
\end{proof}
%\subsubsection{Reformulation in terms of the ${\mathfrak m}$ currents}
\subsubsection{Comparison with the plethystic operators of Bergeron et al \cite{BGLX}}

To simplify the comparison with the operators of \cite{BGLX}, let us introduce the generating functions:
\begin{eqnarray}{\mathfrak m}(z)&:=&\sum_{n\in \Z} z^n\, {\mathcal M}_n= \frac{1-q}{q^{1/2}} \,t^{\frac{N-1}{2}}\, {\mathfrak e}(q^{-1/2}z)\label{mcurdefone}, \\
\widetilde{\mathfrak m}(z)&:=&{\mathfrak m}(q z)\vert_{q\to q^{-1},t\to t^{-1}}=\frac{1-q^{-1}}{q^{-1/2}}\,t^{-\frac{N-1}{2}}\, {\mathfrak f}(q^{-1/2}z).\label{mcurdeftwo}\end{eqnarray}
In the infinite rank limit, we write
$${\mathfrak m}_\infty(z):=\lim_{N\to\infty} t^{1-N}\,  {\mathfrak m}(z)
=\frac{1-q}{q^{1/2}}{\mathfrak e}_\infty(q^{-1/2}z),\quad 
\widetilde{\mathfrak m}_\infty(z)={\mathfrak m}_\infty(q z)\vert_{q\to q^{-1},t\to t^{-1}}
=\frac{1-q^{-1}}{q^{-1/2}}{\mathfrak f}_\infty(q^{-1/2}z) .$$
In this limit, the currents ${\mathfrak m}_\infty(z)$ and $\widetilde{\mathfrak m}_\infty$ act on symmetric functions as follows:
\begin{eqnarray}
{\mathfrak m}_\infty(z)\, F[X] &=&
\frac{1}{1-t^{-1}} \frac{C_\infty(t^{-1}z)}{C_\infty(z)} F\left[ X+\frac{q-1}{z}\right],\label{bosom}\\
\widetilde{\mathfrak m}_\infty(z)\, F[X] &=&
\frac{1}{1-t} \frac{C_\infty(t q^{-1}z)}{C_\infty(q^{-1}z)} F\left[ X-\frac{q-1}{z}\right].\label{bosomt}
\end{eqnarray}

Note that as the prefactor only involves the power sums $p_k$ with $k>0$ (via the function $C_\infty$), we may restrict the action to power series $F\in {\mathcal P}_+$,
where ${\mathcal P}_+$ is the space of formal power series of the $\{p_k\}_{k\in \Z_{>0}}$.

The action (\ref{bosom}-\ref{bosomt}) can now be compared with the definition of the difference operators $D_k,D_k^*$ of \cite {BGLX}. These were 
defined for $k\in \Z_+$ only, via a plethystic formulation. However, it is easy to extend the definition to
all $k\in \Z$, by considering the  
generating currents $D(z):=\sum_{k\in \Z} z^k D_k$ and $D^*(z):=\sum_{k\in \Z} z^k D_k^*$. 
Their action on symmetric functions $F[X]\in {\mathcal P}_+[X]$, $X$
an alphabet
of infinitely many variables $x_1,x_2,...$, is:
\begin{eqnarray}
D(z)\, F[X]&=&C_\infty(z)\, F\left[X+\frac{(1-t)(1-q)}{z}\right] \label{fromBG}, \\
D^*(z)\, F[X]&=&\frac{1}{C_\infty(z)}\, F\left[X-\frac{(1-t^{-1})(1-q^{-1})}{z}\right] .\label{fromBGstar} 
\end{eqnarray}
In \cite{BGLX}, the commutation relations between the $D_k$'s and the $D_k^*$'s are derived. 
They extend to the following commutator of currents:
\begin{equation}\label{BGcurrent}
[D(z),D^*(\frac{w}{q t})]F[X]=
\frac{(1-t)(1-q)}{q t-1} \left\{\delta(z/w)\frac{C_\infty(z)}{C_\infty(\frac{z}{q t})}F[X]-
\delta(q tz/w)F\Big[ X+\frac{(1-t)(1-q)(1-\frac{1}{qt})}{z}\Big] \right\}
\end{equation}
easily derived from the commutation relation of Theorem \ref{bobothm}.
In components, this reads:
$$[D_a,D_b^*]F[X]=\frac{(1-t)(1-q)}{q t-1} \left\{(qt)^b h_{a+b}[X(\frac{1}{qt}-1)]F[X]-
F\Big[ X+\frac{(1-t)(1-q)(1-\frac{1}{qt})}{z}\Big]\vert_{z^{a+b}} \right\}$$
where the notation $F(z)\vert_{z^{a+b}}$ stands for the coefficient of $z^{a+b}$ in the current $F$,
and $h_n[X]$ are the complete symmetric functions of the alphabet $X$.
Note that the case $a,b\geq 0$ agrees with the commutator $[D_a,D_b^*]$ of \cite{BGLX}.

\begin{thm}\label{BGconnect}
Let $\Sigma$ be the operator acting by the plethysm $[X]\mapsto [X/(t-1)]$ 
(accordingly $\Sigma^{-1}$ acts by $[X]\mapsto [X(t-1)]$). 
Then we have the following identities between operators acting on ${\mathcal P}_+[X]$:
\begin{eqnarray*} 
{\mathfrak m}_\infty(z)&=& \frac{1}{1-t^{-1}} \left(\Sigma^{-1} D(z) \Sigma\right)\Big\vert_{t\to t^{-1}}, \\
\widetilde{\mathfrak m}_\infty(z)&=& \frac{1}{1-t} \left(\Sigma^{-1} D^*\left(\frac{z}{q t}\right) \Sigma\right)\Big\vert_{t\to t^{-1}} .
\end{eqnarray*}
\end{thm}
\begin{proof}
The bosonized expressions are matched using \eqref{fromBG} and \eqref{fromBGstar} and 
writing  $F=\Sigma G$. Then, we have:
$$F\left[X\pm \frac{(1-t)(1-q)}{u}\right] =G\left[\frac{X}{t-1}\pm \frac{q-1}{u}\right],$$
while
$$\Sigma^{-1} \, C_\infty(u)\, \Sigma=\frac{C_{\infty}(t u)}{C_{\infty}(u)}, \quad 
\Sigma^{-1} \, \frac{1}{C_\infty(\frac{u}{q t})}\, \Sigma=\frac{C_{\infty}(t^{-1}q^{-1}u)}{C_{\infty}(q^{-1} u)} .$$
Hence 
\begin{eqnarray*}
\Sigma^{-1}D(u)\Sigma \,G[X]&=&\Sigma^{-1}D(u)F[X]= \frac{C_{\infty}(t u)}{C_{\infty}(u)} \, G\left[X + \frac{q-1}{u}\right],\\
\Sigma^{-1}D^*(\frac{u}{q t})\Sigma \,G[X]&=&\Sigma^{-1}D^*(\frac{u}{q t})F[X]=\frac{C_{\infty}(t^{-1}q^{-1}u)}{C_{\infty}(q^{-1}u)}  \, G\left[X - \frac{q-1}{u}\right],
\end{eqnarray*}
where the action of $\Sigma^{-1}$ on $C_\infty(u)^{\pm 1}G\left[\frac{X}{t-1}\pm \frac{q-1}{u}\right]$ is expressed
by considering the latter as a function of $[X]$.
The Theorem follows by taking $t\to t^{-1}$ in the above, and comparing with eqns. (\ref{bosom}-\ref{bosomt}).
\end{proof}

%Note that the commutation relation \eqref{BGcurrent} is equivalent to the $N\to\infty$ limit of
%the commutation relation of Theorem \ref{commuef}, once expressed in terms of the ${\mathfrak m}$ currents.

\begin{remark}\label{betternablarem}
Theorem \ref{BGconnect} allows to refine Remark \ref{nablarem} as follows.
The transformation relating the Macdonald polynomials $P_\lambda$ to the modified Macdonald polynomials
${\widetilde H}_\lambda$ is the following \cite{macdo}:
$${\widetilde H}_\lambda[X]=\phi_\lambda(t) \, P_\lambda\left[ \frac{X}{t-1}\right]\Big\vert_{t\to t^{-1}}
=\phi_\lambda(t) \, (\Sigma \, P_\lambda)\vert_{t\to t^{-1}}, $$
where $\phi_\lambda(t)$ is a normalization factor independent of the $x_i$'s.
The $\nabla$ operator of \cite{BG} has eigenvectors ${\widetilde H}_\lambda$
with eigenvalues $T_\lambda:=t^{n(\lambda)}q^{n(\lambda')}$,
i.e.  $\nabla\, {\widetilde H}_\lambda=T_\lambda \, {\widetilde H}_\lambda$. We deduce that
$$(\Sigma^{-1}\, \nabla\, \Sigma)\vert_{t\to t^{-1}} \, P_\lambda= T_{\lambda}\vert_{t\to t^{-1}}\, P_\lambda
=t^{-n(\lambda)}q^{n(\lambda')}\, P_\lambda. $$
We finally identify
\begin{equation}\label{nabnab}
\eta^{-1}=\nabla^{(N)}=C_N\, (t^{\frac{N-1}{2}}q^{\frac{1}{2}})^{d}\, (\Sigma^{-1}\, \nabla\, \Sigma)\vert_{t\to t^{-1}},
\end{equation}
where $d$ acts on Macdonald polynomials as $d\, P_\lambda =|\lambda| \, P_\lambda$, $C_N$ as in Remark \ref{nablarem},
and $\nabla$ is restricted to act on symmetric functions of $x_1,x_2,...,x_N$. The element $\eta$ in the completion of the DAHA was defined in Equation \eqref{etadefn}.
\end{remark}

\subsubsection{Bosonization formulas for ${\mathfrak e}_\infty(z),{\mathfrak f}_\infty(z)$
in the $N\to\infty$ limit}

The so-called bosonization of the currents uses the operators $p_k$ and $\frac{\partial}{\partial p_k}$, $k\in \Z^*$. 
%For an infinite collection $X$, the functions $\{p_k[X]: k\in \Z^*\}$ are algebraically 
%independent, so we may consider the $p_k$'s as independent variables.
They obey the following commutation relations
$$
\left[ \frac{\partial}{\partial p_j}, p_k\right] = \delta_{jk}
$$
interpreted as independent harmonic oscillator relations. These relations still hold when evaluated on the infinite collection $X=(x_1,x_2,...)$, as the $p_k[X]$ remain independent variables.
We call the expression of the currents in terms of these elements evaluated on the infinite collection $X$, a bosonization.

As $\{p_k[X], k\in\Z^*\}$ are independent, a plethysm which adds one additional variable to the alphabet can be written as
$$F[X+\mu]=\exp\left\{ \sum_{k\neq 0} \mu^k \frac{\partial}{\partial p_k[X]}\right\} F[X], $$
as readily seen from multiple Taylor expansion.
This allows to rewrite Theorem \ref{bobothm} as
\begin{thm}
The limiting Macdonald currents ${\mathfrak e}_\infty(z)$, ${\mathfrak f}_\infty(z)$ and Cartan currents $\psi_\infty^\pm(z)$
action on ${\mathcal P}[X]$ can be expressed in terms of the Heisenberg algebra generators as follows:
\begin{eqnarray*}
{\mathfrak e}_\infty(z)&=&\frac{q^{1/2}}{(1-q)(1-t^{-1})} \, 
e^{\sum_{k> 0} p_k[X] \frac{q^{k/2}(1-t^{-k})z^k}{k}} \, 
e^{\sum_{k\neq 0} \frac{q^{k/2}-q^{-k/2}}{z^k}\frac{\partial}{\partial p_k[X]}}, \nonumber \\
{\mathfrak f}_\infty(z)&=&\frac{q^{-1/2}}{(1-q^{-1})(1-t^{-1})}
e^{\sum_{k>0}p_{k}[X]\frac{q^{-k/2}(1-t^{k})z^{k}}{k} }\,   
e^{-\sum_{k\neq 0} \frac{q^{k/2}-q^{-k/2}}{z^k}\frac{\partial}{\partial p_k[X]}}, \nonumber \\
\psi_\infty^+(z)&=& e^{\sum_{k>0} p_{k}[X]\frac{(t^{k/2}-t^{-k/2})((qt^{-1})^{k/2}-(qt^{-1})^{-k/2})z^k}{k}},\nonumber \\
\psi_\infty^-(z)&=& e^{\sum_{k\neq 0} \frac{(q^{k/2}-q^{-k/2})(1-(t q^{-1})^{k})}{z^{k}}\frac{\partial}{\partial p_k[X]}}.
\nonumber
\end{eqnarray*}
\end{thm}

\begin{remark}
When restricted to ${\mathcal P}_+[X]$ (i.e. dropping all $k<0$ summations), these expressions are identical to the bosonized expressions for the level $(1,0)$ representation of the quantum toroidal algebra \cite{FHHSY} up to simple redefinitions of the generators. This non-trivial central charge is a feature of the $N\to\infty$ limit. 
More precisely, we have the correspondence (see \eqref{corresp} above):
\begin{eqnarray*}
&&\eta(z)=\frac{(1-q)(1-t^{-1})}{q^{1/2}}{\mathfrak e}_\infty(q^{-1/2}z),\ 
\xi(z)=\frac{(1-q^{-1})(1-t^{-1})}{q^{-1/2}}{\mathfrak f}_\infty(t^{-1/2}z),\\ 
&&\varphi^+(z)=\psi_\infty^-(q^{-3/4}t^{1/4}z^{-1}),
\ \varphi^-(z)=\psi_\infty^+(q^{-1/4}t^{-1/4}z^{-1})
\end{eqnarray*}
while the oscillator modes $a_k$, $k\in \Z_{>0}$ correspond to:
$$ a_{-k}=p_k[X],\quad a_k=k\frac{1-q^k}{1-t^k}\,\frac{\partial}{\partial p_k[X]}$$
Note that for finite $N$, the representation $(0,0)$ was completely different, and corresponded to taking 
{\it two} mutually commuting families of harmonic oscillators, namely $a_k=p_k$ for $ k>0$, and $a_k=p_k$ for $k<0$
together with their respective adjoints $\partial/\partial p_k$ for $k>0$ and for $k<0$. 
This is related to the fact that for $\ell_1=0$,
i.e. ${\hat \gamma}=1$, the modes $a_k$ and $a_{-\ell}$ commute for all $k,\ell>0$, leading to commuting $\psi^\pm$
whereas they don't when $\ell_1=1$, and $\psi^\pm$ don't commute either.
\end{remark}

\subsection{Plethystic formulas in the $t\to\infty$ limit}

We now investigate the dual Whittaker limit $t\to\infty$ of the plethystic expressions for the currents. To this end, we use the definition of the limiting currents
${\mathfrak e}^{(\infty)},{\mathfrak f}^{(\infty)},\psi^{\pm(\infty)}$ as in \eqref{limefpsi}.

First, using \eqref{psilim}, we are led to introduce the notation $A$ for the following symmetric function:
$$A:=x_1x_2\cdots x_N.$$
The function ${\rm Log}\, A$ can be understood as a renormalized version of the power sum $p_0[X]$, namely:
$${\rm Log}\, A= \lim_{\epsilon\to 0} \frac{p_\epsilon[X] -N}{\epsilon} .$$
By a slight abuse of notation we shall write $A=e^{p_0}$.
In the dual Whittaker limit,
\begin{eqnarray*}
\psi^{+(\infty)}(z)&=&\frac{(-q^{-1/2}z)^N\,  A}{C(q^{-1/2}z)C(q^{1/2}z)}=
(-q^{-1/2}z)^N\, e^{p_0+\sum_{k>0} p_k(q^{k/2}+q^{-k/2})\frac{z^k}{k}},\\
\psi^{-(\infty)}(z)&=&\frac{(-q^{-1/2}z^{-1})^N\,  A^{-1}}{{\widetilde C}(q^{-1/2}z){\widetilde C}(q^{1/2}z)}=
(-q^{-1/2}z^{-1})^N\, e^{-p_0+\sum_{k>0} p_{-k}(q^{k/2}+q^{-k/2})\frac{z^{-k}}{k}}.
\end{eqnarray*}

The necessity of the introduction of the quantity $A$ (or $p_0$) imposes on us to consider it as an independent
variable, so that plethysms will act on the space ${\mathcal P}_0$ of functions that are power series of all $p_k$, $k\in \Z$ (including $k=0$). The plethysms must therefore acquire a $p_0$ dependence as well.
More precisely, we must take into account an extra commutation relation in addition to those of Lemma \ref{compk}.
\begin{lemma}
We have the relations:
\begin{equation}
{\mathcal D}^{q,t}_{1;n}\, p_0[X]=(p_0[X]+{\rm Log}\, q)  \, {\mathcal D}^{q,t}_{1;n},\quad {\rm and}\quad 
{\mathcal D}^{q^{-1},t^{-1}}_{1;n}\, p_0[X]=(p_0[X]-{\rm Log}\, q)\, {\mathcal D}^{q^{-1},t^{-1}}_{1;n}.
\end{equation}
\end{lemma}
\begin{proof}
We use the $q$-commutation relations:
\begin{equation}
{\mathcal D}^{q,t}_{1;n}\, A=q \, A\, {\mathcal D}^{q,t}_{1;n},\quad {\rm and}\quad 
{\mathcal D}^{q^{-1},t^{-1}}_{1;n}\, A=q^{-1} A \, {\mathcal D}^{q^{-1},t^{-1}}_{1;n},
\end{equation}
due to $\Gamma_i^{\pm 1} \, A= q^{\pm 1} A \, \Gamma_i$.
\end{proof}

We now extend plethysms to functions of the $p_k$, $k\in \Z$ so that, in the case of addition of one variable,
we have:
$$p_k[X+\mu]=\left\{ \begin{matrix} 
p_k[X]+\mu^k & {\rm if}\ k\neq 0\\
p_0[X]+{\rm Log}\, \mu & {\rm if}\  k=0
\end{matrix} \right. \ .$$
As a check, 
$$p_k\left[X+\frac{q^{1/2}-q^{-1/2}}{z}\right]=\left\{ \begin{matrix} 
p_k[X]+\frac{q^{k/2}-q^{-k/2}}{z^k}  & {\rm if}\ k\neq 0\\
p_0[X]+{\rm Log}\, q & {\rm if}\  k=0
\end{matrix} \right. \ .$$
Note that the notation $\left[X+\mu \right]$ now refers to the full plethysm involving all the $p_k$, $k\in \Z$.
With this notation, we easily get the action of the currents on functions 
$F[X]\in {\mathcal P}_0[X]$, i.e. formal power series of all 
$p_k[X]$, $k\in \Z$:
\begin{eqnarray*}
{\mathfrak e}^{(\infty)}(z)\, F[X]&=&\frac{q^{1/2}}{1-q} \left\{
\frac{1}{C(q^{1/2}z)}- 
\frac{(-q^{-1/2}z^{-1})^N A^{-1}}{{\widetilde C}(q^{1/2}z^{-1})}\right\} 
F\left[X+\frac{q^{1/2}-q^{-1/2}}{z} \right], \nonumber \\
{\mathfrak f}^{(\infty)}(z)\, F[X]&=&\frac{q^{-1/2}}{1-q^{-1}} \left\{
\frac{1}{{\widetilde C}(q^{1/2}z^{-1})}- 
\frac{(-q^{-1/2}z )^N A}{C(q^{-1/2}z)}\right\} 
F\left[X-\frac{q^{1/2}-q^{-1/2}}{z} \right]. \nonumber
\end{eqnarray*}

%or equivalently:
%\begin{eqnarray*}
%{\mathfrak e}^{(\infty)}(z)&=&\frac{q^{1/2}}{1-q} \left\{
%e^{\sum_{k>0} p_k \frac{q^{k/2}z^k}{k}}-
%e^{-p_0+\sum_{k>0}p_{-k}\frac{q^{-k/2}z^{-k}}{k} } \right\} \, q^{\frac{\partial}{\partial p_0}}\,
%e^{\sum_{k\neq 0} \frac{q^{k/2}-q^{-k/2}}{z^k}\frac{\partial}{\partial p_k}} \nonumber \\
%{\mathfrak f}^{(\infty)}(z)&=&\frac{q^{-1/2}}{1-q^{-1}}
%\left\{e^{\sum_{k>0}p_{-k}\frac{q^{k/2}z^{-k}}{k} }-
%e^{p_0+\sum_{k>0}p_{k}\frac{q^{-k/2}z^{k}}{k} }\right\}\, q^{-\frac{\partial}{\partial p_0}}\,
%e^{-\sum_{k\neq 0} \frac{q^{k/2}-q^{-k/2}}{z^k}\frac{\partial}{\partial p_k}} \nonumber
%\end{eqnarray*}

\section{Constant terms and shuffle product}\label{shufflesec}
\subsection{A constant term identity for generalized Macdonald operators}
In this section we reformulate the relations between the
generating functions of the previous sections in terms of shuffle products.

\subsubsection{Constant term identities and generating currents}

We use the generating currents ${\mathfrak m}_\al(u)$ for the generalized Macdonald operators \eqref{genmacdop}:
\begin{equation}\label{defm}
{\mathfrak m}_\al(u):=\sum_{n\in \Z} u^n\, {\mathcal M}_{\al;n}
=\sum_{I\subset [1,N]\atop |I|=\al} \delta(u x_I) \prod_{i\in I\atop j\not\in I}
\frac{t x_i-x_j}{x_i - x_j} \, \Gamma_I.
\end{equation}
\begin{remark}
The currents ${\mathfrak m}_\al(u)$ and ${\mathfrak e}_\al(u)$ are related via
\begin{equation}\label{emrela}
{\mathfrak e}_\al(u)=\frac{q^{\frac{\al}{2}}t^{-\frac{\al(N-\al)}{2}}}{(1-q)^\al}\, {\mathfrak m}_\al(q^{\al/2}u)
\end{equation}
\end{remark}
Note that ${\mathfrak m}_1(u)={\mathfrak m}(u)$ of Equation \eqref{mcurdefone}. 
Theorem \ref{efrela} implies the exchange relation
\begin{equation}\label{exchange}
g(u,v) \, {\mathfrak m}(u)\, {\mathfrak m}(v)+g(v,u)\, {\mathfrak m}(v)\, {\mathfrak m}(u)=0,
\end{equation}
where $g$ is as in \eqref{defofg}. 

In terms of components, this exchange relation is equivalent to the relation between components:
\begin{equation}\label{cubicM}
\mu_{a,b}:=q t {\mathcal M}_{a-3}{\mathcal M}_{b}-(t^2+q^2 t+q){\mathcal M}_{a-2}{\mathcal M}_{b-1}+(qt^2+t+q^2){\mathcal M}_{a-1}{\mathcal M}_{b-2}-q t{\mathcal M}_{a}{\mathcal M}_{b-3}=-\mu_{b,a}
\end{equation}
for all $a,b\in\Z$.

\begin{defn}
We define the multiple constant term of any rational symmetric function $f(u_1,...,u_\al)\in {\mathcal F}_\al$ as:
\begin{equation}\label{CTdef}
CT_{u_1,...,u_\al}\left( f(u_1,...,u_\al)\right):=\prod_{i=1}^\al \oint \frac{du_i}{2i\pi\, u_i} f(u_1,...,u_\al)
\end{equation}
where the contour integral picks up the residues at $u_i=0$.
\end{defn}
In particular, we have $CT_v( f(v)\delta(u/v))=f(u)$.

\begin{defn}\label{malphadef}
To each symmetric rational function $P(x_1,x_2,...,x_\al)\in {\mathcal F}_\al$ we associate the difference operator
${\mathcal M}_\al(P)$:
\begin{equation}
{\mathcal M}_\al(P):=\frac{1}{\al!}CT_\bu\left( P(u_1^{-1},u_2^{-1},...,u_\al^{-1}) \prod_{1\leq i<j\leq \al} 
\frac{(u_i-u_j)(u_i-q u_j)}{(u_i-t u_j)(t u_i-q u_j)}\, \prod_{i=1}^\al {\mathfrak m}(u_i)\right) \label{ctdiffop}
\end{equation}
\end{defn}

We have the following remarkable result:

\begin{thm}\label{mainthm}
For any symmetric rational function $P(x_1,...,x_\al)\in {\mathcal F}_\al$, $1\leq \al\leq N$, we have the identity:
\begin{equation}
{\mathcal M}_\al(P)= {\mathcal D}_\al(P) \label{ctPqt}
\end{equation}
with ${\mathcal D}_\al(P)$ as in Definition \ref{gmacdef}.
\end{thm}
\begin{proof}
Let us compute:
\begin{eqnarray*}
&&\prod_{i<j}  \frac{u_i-u_j}{u_i-t u_j}\frac{u_i-q u_j}{t u_i-q u_j} \, \prod_{i=1}^\al {\mathfrak m}(u_i)\\
&=& \prod_{i<j} \frac{u_i-u_j}{u_i-t u_j}\frac{u_i-q u_j}{t u_i-q u_j} \sum_{i_1,i_2,...,i_\al} 
\prod_{k=1}^\al \left(\delta(u_k x_{i_k}) 
\prod_{j_k\neq i_k}\frac{t x_{i_k}-x_{j_k}}{x_{i_k}-x_{j_k}} \Gamma_{i_k}\right) \\
&=& \prod_{i<j} \frac{u_i-u_j}{u_i-t u_j}\frac{u_i-q u_j}{t u_i-q u_j} \times \\
&&\qquad \times\, \sum_{i_1\neq i_2\neq ...\neq i_\al} 
\prod_{j=1}^\al \delta(u_j x_{i_j}) 
\prod_{k<\ell }\frac{t x_{i_k}- x_{i_\ell}}{x_{i_k}- x_{i_\ell}}\frac{t x_{i_\ell}-q x_{i_k}}{x_{i_\ell}- q x_{i_k}} 
\prod_{j=1}^\al \prod_{i\neq i_1,...,i_\al}\frac{t x_{i_j}-x_i}{x_{i_j}-x_i}\, \Gamma_{i_1}\cdots \Gamma_{i_\al}
\\
&=&\sum_{i_1\neq i_2\neq ...\neq i_\al} 
\prod_{j=1}^\al \delta(u_j x_{i_j}) \prod_{i\neq i_1,...,i_\al} \frac{t x_{i_k}-x_i}{x_{i_k}-x_i}
\, \Gamma_{i_1}\cdots \Gamma_{i_\al}\\
&=&\al!\, \sum_{i_1<i_2< ...<i_\al} 
\prod_{j=1}^\al \delta(u_j x_{i_j}) \prod_{i\neq i_1,...,i_\al} \frac{t x_{i_k}-x_i}{x_{i_k}-x_i}
\, \Gamma_{i_1}\cdots \Gamma_{i_\al}\\
&=&\al!\,  \frac{1}{\al!\,(N-\al)!}{\rm Sym}\left( \prod_{k=1}^\al \delta(u_k x_{k})  
\prod_{1\leq i\leq \al<j\leq N} \frac{t x_{i}-x_j}{x_{i}-x_j} \ 
\Gamma_1 \cdots \Gamma_\al \right)
\end{eqnarray*}
We have first noted that terms with any two identical $i_k=i_\ell$, $k<\ell$, in the sum must vanish. 
This is due to the prefactor $(u_{k}-q u_{\ell})$ which when multiplying the delta function 
$\delta(u_k x_{i_k})\delta(u_\ell q x_{i_k})$ yields a zero result.
Comparing with Definition \ref{symdiffop}, 
the constant term \eqref{ctPqt} follows immediately.
\end{proof}

As a by-product of the proof of Theorem \ref{mainthm} above, we note that if $\al>N$, then there are no terms in
which all $i_k$ are distinct (as there are at most $N$ of them), thus causing the result to vanish. 
We deduce the following:

\begin{cor}\label{vanicor}
For any symmetric rational function $P(x_1,...,x_\al)\in {\mathcal F}_\al$, $\al>N$, we have:
\begin{equation}\label{vanimac}   {\mathcal M}_\al(P)= {\mathcal D}_\al(P) =0\qquad \forall\ \al>N .
\end{equation}
\end{cor}

This implies in particular that 
\begin{equation}\label{macvani}
{\mathcal M}_{N+1;n}=0 \quad \forall\ n\in \Z .
\end{equation}
in agreement with \eqref{elipquo}.

Recalling the definition \eqref{schurmacdo}, 
Theorem \ref{mainthm} has also the following immediate application to $P=s_{a_1,...,a_\al}(x_1,...,x_\al)$:

\begin{cor}\label{corctMqt}
We have:
\begin{equation}\label{ctMqt}
{\mathcal M}_{a_1,...,a_\al}=\frac{1}{\al!}CT_\bu\left( s_{a_1,...,a_\al}(\bu^{-1}) \prod_{1\leq i<j\leq \al} \frac{(u_i-u_j)
(u_i-q u_j)}{(u_i-t u_j)(t u_i-q u_j)}
\, \prod_{i=1}^\al {\mathfrak m}(u_i)\right)
\end{equation}
\end{cor}

In particular, this implies:
\begin{equation}\label{notcur}
{\mathcal M}_{\al;n}=\frac{1}{\al!}CT_\bu\left( (u_1 u_2 \cdots u_\al)^{-n} \prod_{1\leq i<j\leq \al} \frac{(u_i-u_j)
(u_i-q u_j)}{(u_i-t u_j)(t u_i-q u_j)}
\, \prod_{i=1}^\al {\mathfrak m}(u_i)\right)
\end{equation}
or equivalently in terms of the currents ${\mathfrak m}_\al(z)$ or ${\mathfrak e}_\al(z)$
of \eqref{highercur}:
\begin{cor}\label{ealine}
We have the following constant term identities:
\begin{eqnarray}
\qquad {\mathfrak m}_{\al}(z)&=&\frac{1}{\al!}CT_\bu\left( \delta(u_1 u_2 \cdots u_\al/z)\prod_{1\leq i<j\leq \al} \frac{(u_i-u_j)
(u_i-q u_j)}{(u_i-t u_j)(t u_i-q u_j)}
\, \prod_{i=1}^\al {\mathfrak m}(u_i)\right)\\
\qquad {\mathfrak e}_{\al}(z)&=&\frac{1}{\al!}CT_\bu\left( \delta(u_1 u_2 \cdots u_\al/z)\prod_{1\leq i<j\leq \al} \frac{(u_i-u_j)
(u_i-q u_j)}{(u_i-t u_j)(u_i-q/t u_j)}
\, \prod_{i=1}^\al {\mathfrak e}(u_i)\right)
\end{eqnarray}
\end{cor}
\begin{proof} The first equation is the current form of \eqref{notcur} obtained by multiplying by $z^n$ 
and summing over $n\in \Z$.
The second equation results from the change of variables $u_i\mapsto q^{1/2} u_i$ for all $i$ and
$z\mapsto q^{\al/2} \, z$ in the previous multiple 
constant term residue integral, while using the relation \eqref{emrela}.
\end{proof}

The result of Corollary \ref{corctMqt} may be rephrased in terms of generating currents as follows.
We consider the generating multi-current with argument $\bv=(v_1,v_2,...,v_\al)$:
\begin{equation}\label{defMal}
{\mathfrak M}_\al(\bv):=\sum_{a_1,...,a_\al\in \Z} {\mathcal M}_{a_1,...,a_\al}\, v_1^{a_1}v_2^{a_2}\cdots v_{\al}^{a_\al}=
 \frac{1}{\prod_{j=1}^\al v_j^{\al-j}}{\mathcal D}_\al \left(\frac{\det\left(\left( \delta(x_i \, v_j) \right)_{1\leq i,j\leq \al}\right)}{\prod_{1\leq i<j\leq \al}(x_i-x_j)} \right)
\end{equation}

\begin{thm}\label{Mofm}
The generating current for the generalized Macdonald operators \eqref{schurmacdo} reads:
\begin{equation}
{\mathfrak M}_\al(\bv)=\prod_{1\leq i<j\leq \al} \frac{(v_i-q v_j)}{(t -v_i v_j^{-1})(t v_i-q  v_j)}
\, \prod_{i=1}^\al {\mathfrak m}(v_i)
\end{equation}
\end{thm}
\begin{proof}
Using the identity \eqref{ctMqt}, and the expression \eqref{defschur} for the generalized Schur function, we compute:
\begin{eqnarray*}
{\mathfrak M}_\al(\bv)&=&\frac{1}{\al!}CT_\bu\left(\sum_{a_1,...,a_\al\in \Z} \det\left( u_i^{-a_j-\al+j}\right)
v_1^{a_1}v_2^{a_2}\cdots v_{\al}^{a_\al} \right. \times \\
&&\qquad \qquad \qquad \times \left. \prod_{1\leq i<j\leq \al} \frac{(u_i-q u_j)}{(t u_i^{-1}-u_j^{-1})(t u_i-q u_j)}
\, \prod_{i=1}^\al {\mathfrak m}(u_i)\right)\\
&=&\frac{1}{\al!}CT_\bu\left(\det\left( \delta(v_j/u_i)\, v_j^{j-\al}\right)
 \prod_{1\leq i<j\leq \al} \frac{(u_i-q u_j)}{(t u_i^{-1}-u_j^{-1})(t u_i-q  u_j)}
\, \prod_{i=1}^\al {\mathfrak m}(u_i)\right)\\
&=&\frac{1}{\al!}CT_\bu\left(\det\left( \delta(v_j/u_i)\right)
 \prod_{1\leq i<j\leq \al} \frac{(u_i-q u_j)}{(t -u_i u_j^{-1})(t u_i-q u_j)}
\, \prod_{i=1}^\al {\mathfrak m}(u_i)\right)\\
&=&\prod_{1\leq i<j\leq \al} \frac{(v_i-q v_j)}{(t -v_i v_j^{-1})(t v_i-q v_j)}
\, \prod_{i=1}^\al {\mathfrak m}(v_i)
\end{eqnarray*}
where in the last step we have used the skew-symmetry of both the determinant and the quantity next to it, 
to see that each of the $\al!$ terms in the expansion of the determinant contributes the same as the diagonal 
term $\prod \delta(v_i/u_i) u_i^{i-\al}$.
\end{proof}

\begin{cor}\label{corCT}
We have the following alternative expression for the generalized Macdonald operators  \eqref{schurmacdo}:
\begin{equation}\label{ctsimp}
{\mathcal M}_{a_1,...,a_\al}=CT_\bu\left(\prod_{i=1}^\al u_i^{-a_i}
 \prod_{1\leq i<j\leq \al} \frac{(u_i-q u_j)}{(t -u_i u_j^{-1})(t u_i-q u_j)}
\, \prod_{i=1}^\al {\mathfrak m}(u_i)\right)
\end{equation}
\end{cor}
\begin{proof}
The constant term \eqref{ctsimp} picks up the coefficient of 
$u_1^{a_1}u_2^{a_2}\cdots u_\al^{a_\al}$ in ${\mathfrak M}_\al(\bu)$.
\end{proof}

We also have the corresponding current version, when $a_1=a_2=\cdots =a_\al=n$:
\begin{eqnarray}
\qquad\quad  {\mathfrak m}_{\al}(z)&=&CT_\bu\left( \delta(u_1 u_2 \cdots u_\al/z) 
\prod_{1\leq i<j\leq \al} \frac{(u_i-q u_j)}{(t -u_i u_j^{-1})(t u_i-q u_j)}\, \prod_{i=1}^\al {\mathfrak m}(u_i)\right)\label{mbo}\\
\qquad\quad {\mathfrak e}_{\al}(z)&=&CT_\bu\left( \delta(u_1 u_2 \cdots u_\al/z) 
\prod_{1\leq i<j\leq \al} \frac{(u_i-q u_j)}{(t -u_i u_j^{-1})(u_i-q/t u_j)}\, \prod_{i=1}^\al {\mathfrak e}(u_i)\right)\label{ebo}
\end{eqnarray}

\subsubsection{Polynomiality and the $(q,t)$-determinant}

In this section, we state the following:

\begin{conj}\label{polyconj}
The generalized Macdonald operators ${\mathcal M}_{a_1,...,a_\al}$ may be expressed as {\it polynomials} of 
finitely many ${\mathcal M}_{p}$'s. These polynomials are $t$-deformation of the quantum determinant expression \eqref{ctsimplim} below.
\end{conj}

Note that if they exist, such polynomials are not necessarily unique, as they can be modified 
by use of the exchange relations \eqref{cubicM}.

We first give the proof of the conjecture in the case $\al=2$ for arbitrary ${\mathcal M}_{a,b}$, $a,b\in \Z$,
by deriving an explicit polynomial expression (see Theorem \ref{polynomialitythm} below), and the sketch 
of the proof in the case $\al=3$ (Theorem \ref{polthree}). 
Further evidence will be derived in Section \ref{EHAsec}, where the connection to Elliptic Hall 
algebra leads naturally to a polynomial expression for 
${\mathcal M}_{\al;n}={\mathcal M}_{n,n,...,n}$ for all $\al,n$, as a function of solely 
${\mathcal M}_n,{\mathcal M}_{n\pm 1}$.

\begin{thm}\label{polynomialitythm}
For all $a,b\in \Z$, the operator ${\mathcal M}_{a,b}$ can be expressed as an explicit quadratic polynomial of the 
${\mathcal M}_{n}$'s, with coefficients in $\C(q,t)$. More precisely, we have:
\begin{equation}\label{mtwopol}
{\mathcal M}_{a,b}=\frac{q(q+t^2)\nu_{a,b}-t(1+q) \nu_{a+1,b-1}+(q+t^2)\nu_{b-1,a+1}
-q t(1+q) \nu_{b-2,a+2}}{(q-1)(q^2-t^2)(1-t^2)}
\end{equation}
where $\nu_{a,b}$ stands for the following ``quantum determinant":
\begin{equation}\label{qdetwo}
\nu_{a,b}=\left\vert
\begin{matrix}
{\mathcal M}_a & {\mathcal M}_{b-1}\\
{\mathcal M}_{a+1} & {\mathcal M}_{b}
\end{matrix} \right\vert_q:={\mathcal M}_a\, {\mathcal M}_b- q\, {\mathcal M}_{a+1}\, {\mathcal M}_{b-1}
\end{equation}
\end{thm}
\begin{proof}
We start from the current relation of Theorem \ref{Mofm}:
\begin{equation}\label{mtwo}
{\mathfrak M}_2(v_1,v_2)=\frac{(v_1-q v_2)}{(t -v_1 v_2^{-1})(t v_1-q  v_2)}
\,{\mathfrak m}(v_1)\,{\mathfrak m}(v_2)
\end{equation}
We note that the exchange relation \eqref{exchange} can be rewritten as
$$\frac{(v_1-q v_2)}{(t v_2 -v_1)(t v_1-q  v_2)}
\,{\mathfrak m}(v_1)\,{\mathfrak m}(v_2)+\frac{(v_2-q v_1)}{(t v_1 -v_2)(t v_2-q  v_1)}
\,{\mathfrak m}(v_2)\,{\mathfrak m}(v_1)=0$$
As a consequence, the renormalized double current ${\mathfrak N}_2(v_1,v_2):={\mathfrak M}_2(v_1,v_2)/v_2$
is skew-symmetric, and we may rewrite \eqref{mtwo} as:
\begin{equation}\label{none}
\delta_{1,2}\,  {\mathfrak N}_2(v_1,v_2)
=\frac{v_1-q v_2}{v_1v_2} \,{\mathfrak m}(v_1)\,{\mathfrak m}(v_2)=:\mu_2(v_1,v_2) 
\end{equation}
where we use the notation
$$\delta_{i,j}:= \Big(t-\frac{v_i}{v_j}\Big)\Big(t-q \frac{v_j}{v_i}\Big)$$
Using the skew-symmetry of ${\mathfrak N}_2$, let us apply the transposition $(12)$ that interchanges 
$v_1\leftrightarrow v_2$, with the result:
\begin{equation}\label{ntwo}
\delta_{2,1}\,  {\mathfrak N}_2(v_1,v_2)=-\mu_2(v_2,v_1) 
\end{equation}
Next, we eliminate the prefactors by using the ``inversion" relation:
\begin{equation}\label{inverela}
\eta_{1,2}\delta_{1,2}-\theta_{1,2}\delta_{2,1}=1,\quad \eta_{i,j}:=\frac{q(q+t^2)-t(1+q)\frac{v_j}{v_i}}{(q-1)(1-t^2)(q^2-t^2)},
\quad \theta_{i,j}:=\frac{(q+t^2)-qt(1+q)\frac{v_j}{v_i}}{(q-1)(1-t^2)(q^2-t^2)}
\end{equation}
which is easily derived by decomposing $1/(\delta_{1,2}\delta_{2,1})$ into simple fractions. 
With the above choice, we conclude that:
$${\mathfrak M}_2(v_1,v_2)= v_2 {\mathfrak N}_2(v_1,v_2)
=v_2(\eta_{1,2}\, \mu_2(v_1,v_2)+\theta_{1,2}\, \mu_2(v_2,v_1))$$
Defining $\nu_2(v_1,v_2):=v_2 \mu_2(v_1,v_2)=\left(1-q\frac{v_2}{v_1}\right) \,{\mathfrak m}(v_1)\,{\mathfrak m}(v_2)$,
we finally get:
\begin{equation}\label{soltwo}
{\mathfrak M}_2(v_1,v_2)=\eta_{1,2}\, \nu_2(v_1,v_2)+\frac{v_2}{v_1}\,\theta_{1,2}\, \nu_2(v_2,v_1)
\end{equation}
Introducing the mode expansion: $\nu_2(v_1,v_2)=\sum_{a,b\in \Z} \nu_{a,b} v_1^a v_2^b$, with
$\nu_{a,b}$ as in \eqref{qdetwo}, the formula \eqref{mtwopol} follows from the mode expansion of \eqref{soltwo}.
\end{proof}

Note that the expression \eqref{mtwopol} expresses ${\mathcal M}_{a,b}$ in terms of more variables that 
just ${\mathcal M}_a,{\mathcal M}_{a+1},{\mathcal M}_{b-1},{\mathcal M}_{b}$. However, if we consider the limit
when $t\to \infty$ of this expression, after defining $M_{a,b}=\lim_{t\to \infty} t^{-2(N-2)}{\mathcal M}_{a,b}$,
$M_a=\lim_{t\to \infty} t^{-(N-1)} {\mathcal M}_a$, and $n_{a,b}=\lim_{t\to \infty} t^{-2(N-1)} \nu_{a,b}$, we obtain:
$$M_{a,b}=\frac{q \ n_{a,b}+n_{b-1,a+1}}{q-1}=n_{a,b}=M_aM_b-q M_{a+1}M_{b-1}$$
where we have used the $t\to \infty$ exchange relation. 
Indeed, for finite $t$ the expression \eqref{mtwopol} is not unique: 
it is unique only up to the exchange relation \eqref{cubicM}
for the ${\mathcal M}_n$'s. Here is a simple example.
Revisiting the proof of the Theorem, we note that there is another 
inverting pair $(\eta_{1,2}',\theta_{1,2}')=(-\theta_{2,1},-\eta_{2,1})$ obtained by acting with the transposition (12) 
on the inversion relation \eqref{inverela}, namely we also have:
$$\eta_{1,2}'\, \delta_{1,2}-\theta_{1,2}' \, \delta_{2,1}=1$$
This choice leads to alternative expressions:
\begin{eqnarray}
{\mathfrak M}_2(v_1,v_2)&=& v_2 {\mathfrak N}_2(v_1,v_2)=-v_2\Big(\theta_{2,1}\, \mu_2(v_1,v_2)+\eta_{2,1}\, \mu_2(v_2,v_1)\Big)\nonumber \\
&=& -\theta_{2,1} \, \nu_2(v_1,v_2)-\frac{v_2}{v_1}\,\eta_{2,1}\, \nu_2(v_2,v_1)\nonumber \\
{\mathcal M}_{a,b}&=&
\frac{q t(1+q)\nu_{a-1,b+1}-(q+t^2)\nu_{a,b}+t(1+q) \nu_{b,a}-q(q+t^2)\nu_{b-1,a+1}}{(q-1)(q^2-t^2)(1-t^2)}
\label{alterM}
\end{eqnarray}
The difference between the two expressions \eqref{mtwopol} and \eqref{alterM} is proportional to:
\begin{equation}\label{difference}
(q+t^2)(\nu_{a,b}+\nu_{b-1,a+1})-t(\nu_{a+1,b-1}+\nu_{b,a})-qt(\nu_{a-1,b+1}+\nu_{b-2,a+2})
\end{equation}
Rewriting the exchange relation \eqref{cubicM} as:
\begin{equation}\label{rewexch}\varphi_{a,b}:=qt \nu_{a-3,b}-(q+t^2)\nu_{a-2,b-1}+t \nu_{a-1,b-2}=-\varphi_{b,a}
\end{equation}
we see that \eqref{difference} is nothing but 
$-\varphi_{a+2,b+1}-\varphi_{b+1,a+2}=0$, as a direct consequence of the exchange relation.
This gives a simple example of equivalence of two polynomial expressions for ${\mathcal M}_{a,b}$
modulo the exchange relations of the algebra.

\begin{thm}\label{polthree}
The conjecture \ref{polyconj} holds for $\al=3$.
\end{thm}
\begin{proof}
Sketch of the proof. Use simple fraction decomposition of the quantity
$1/(\delta_{1,2}\delta_{2,1}\delta_{1,3}\delta_{3,1}\delta_{2,3}\delta_{3,2})$ to obtain a relation of the form:
$$\sum_{\sigma\in S_3} {\rm sgn}(\sigma)\,A_\sigma(v_1,v_2,v_3) \, 
\delta_{\sigma(1),\sigma(2)}\delta_{\sigma(1),\sigma(3)}\delta_{\sigma(2),\sigma(3)}=1$$
with explicit Laurent polynomials $A_\sigma(v_1,...,v_\al)$. This allows to express the skew-symmetric 
current ${\mathfrak N}_\al={\mathfrak M}_\al/(v_2 v_3^2)$ as a sum over the symmetric group $S_3$ 
of Laurent polynomials of the $v$'s times permuted products of the fundamental currents 
${\mathfrak m}(v_i)$'s. The polynomiality property of the coefficients follows.
\end{proof}

More generally, one could try to generalize the above argument. Defining the skew-symmetric function
${\mathfrak N}_\al={\mathfrak M}_\al/(v_2 v_3^2\cdots v_\al^{\al-1})$, we wish to invert the relation
$$\left(\prod_{1\leq i<j\leq \al} \delta_{i,j} \right)\, {\mathfrak N}_\al(v_1,...,v_\al)
=\prod_{1\leq i<j\leq \al}\frac{v_i-q v_j}{v_i v _j}\,
{\mathfrak m}(v_1)\cdots {\mathfrak m}(v_\al)
=:\mu_\al(v_1,...,v_\al)$$
Acting with the permutation group of the $v$'s, we have accordingly for all $\sigma \in S_\al$:
$$\prod_{1\leq i<j\leq \al} \delta_{\sigma(i),\sigma(j)} \, {\mathfrak N}_\al(v_1,...,v_\al)
={\rm sgn}(\sigma)\, \mu_\al(v_{\sigma(1)},...,v_{\sigma(\al)})$$
Inverting the system could be done by looking for Laurent polynomials $A_\sigma(v_1,...,v_\al)$ such that
$$\sum_{\sigma\in S_\al} {\rm sgn}(\sigma)\,A_\sigma(v_1,...,v_\al) \,\prod_{1\leq i<j\leq \al} \delta_{\sigma(i),\sigma(j)}=1$$
If such $A_\sigma$'s existed, then we could write
$${\mathfrak M}_\al=v_2 v_3^2\cdots v_\al^{\al-1}\sum_{\sigma\in S_\al} 
A_\sigma(v_1,...,v_\al)  \, \mu_\al(v_{\sigma(1)},...,v_{\sigma(\al)})$$
and polynomiality would follow.

%\begin{example} \label{exampletwo}
%Rewriting the result of Theorem \ref{Mofm} for $\al=2$ as:
%$$(t -v_1 v_2^{-1})(t -q  v_2v_1^{-1})\, {\mathfrak M}_2(v_1,v_2)=(1-q v_2v_1^{-1})\, {\mathfrak m}(v_1){\mathfrak m}(v_2)\ ,
%$$
%and picking the coefficient of $v_1^nv_2^p$, we get the equivalent recursion relation:
%\begin{equation}
%\label{recualphatwo}
%-qt {\mathcal M}_{n+1,p-1}+(q+t^2) {\mathcal M}_{n,p}-t  {\mathcal M}_{n-1,p+1}={\mathcal M}_{n}\,{\mathcal M}_{p}
%-q \, {\mathcal M}_{n+1}\,{\mathcal M}_{p-1}
%\end{equation}
%Note also the skew-symmetry formula ${\mathcal M}_{a,b}=-{\mathcal M}_{b-1,a+1}$ as a direct consequence of the
%skew-symmetry of the determinantal definition \eqref{defschur} of the generalized Schur function $s_{a,b}=-s_{b-1,a+1}$.
%Together with the recursion relation \eqref{recualphatwo}, it allows to express inductively all the operators ${\mathcal M}_{i,j}$ as quadratic polynomials of the ${\mathcal M}_{n}$'s 
%(see Theorems \ref{oddMthm} and \ref{evenMthm} and Lemma \ref{evenMlemma}
%below for a simpler proof using shuffle relations). 
%\end{example}

\subsection{Plethystic formulation}

Using the plethystic formulas of Section \ref{secboso}, we derive a plethystic formula for the higher currents ${\mathfrak e}_{\al}$.

\begin{thm}\label{bosoethm}
The current ${\mathfrak e}_{\al}(z)$ acts on functions $F[X]$ as follows:
\begin{eqnarray*}
{\mathfrak e}_{\al}(z)\cdot F[X]&=& \left(\frac{q^{1/2}t^{-1/2}}{(1-q)(1-t^{-1})}\right)^\al
{\rm CT}_\bu\left(\prod_{1\leq i<j\leq \al} \frac{(u_i-u_j)}{(u_i-t u_j)(u_i-t^{-1}u_j)} \right. \\
&&\quad  \times \prod_{i=1}^\al \left\{ t^{\frac{N}{2}} \frac{C(q^{1/2}t^{-1}u_i)}{C(q^{1/2} u_i)} 
-t^{-\frac{N}{2}} \frac{{\tilde C}(q^{-1/2}t u_i^{-1})}{{\tilde C}(q^{-1/2} u_i^{-1})} \right\} \\
&&\qquad \qquad \qquad \delta(u_1u_2 \cdots u_\al/z)\cdot \left.
F\left[X+(q^{1/2}-q^{-1/2})\sum_{i=1}^\al \frac{1}{u_i} \right]\right)
\end{eqnarray*}
\end{thm}
\begin{proof}
We start from the bosonized formula for ${\mathfrak e}(z)$ of Theorem \ref{pletbo}, and substitute it into
\eqref{ebo}. We need to compute the action of the plethysm $[X+\frac{q^{1/2}-q^{-1/2}}{z}]$ on ${\mathfrak e}(w)$:
it sends respectively
\begin{eqnarray*}
\frac{C(q^{1/2}t^{-1}w)}{C(q^{1/2}w)}&\mapsto& \frac{(z-q t^{-1} w)(z-w)}{(z-t^{-1}w)(z-q w)} 
\,\frac{C(q^{1/2}t^{-1}w)}{C(q^{1/2}w)}\\
\frac{{\tilde C}(q^{-1/2}t w^{-1})}{{\tilde C}(q^{-1/2}w^{-1})}&\mapsto& \frac{(w-q^{-1} t z)(w-z)}{(w-t z)(w-q^{-1} z)} 
\, \frac{{\tilde C}(q^{-1/2}t w^{-1})}{{\tilde C}(q^{-1/2}w^{-1})}
\end{eqnarray*}
hence both terms are mapped identically, and the theorem follows by repeated applications until all $u_i$'s are
exhausted.
\end{proof}

Similarly, we have the following bosonization of the multi-current ${\mathfrak M}_\al(\bv)$ of \eqref{defMal}.
\begin{thm}
The multi-current ${\mathfrak M}_\al(\bv)$  for the generalized Macdonald operators acts on functions $F[X]$ as follows.
\begin{eqnarray*}
{\mathfrak M}_\al(\bv)\cdot F[X]&=&\left(\frac{t^{\frac{N}{2}}}{t-1}\right)^\al
\prod_{1\leq i<j\leq \al} \frac{(v_i-v_j)}{(t -v_i v_j^{-1})(t v_i-v_j)}\times \\
&&\quad \times\prod_{i=1}^\al \left\{t^{\frac{N}{2}} \frac{C(t^{-1}v_i)}{C(v_i)} 
-t^{-\frac{N}{2}} \frac{{\tilde C}(t v_i^{-1})}{{\tilde C}(v_i^{-1})} \right\} \cdot  F\left[X+(q-1)\sum_{i=1}^\al \frac{1}{v_i} \right]
\end{eqnarray*}
\end{thm}
\begin{proof}
Starting from \eqref{defMal}, we substitute ${\mathfrak m}(v_i)= q^{-1/2}(1-q)t^{\frac{N-1}{2}}{\mathfrak e}(q^{-1/2}v_i)$,
and then use the proof of Theorem \ref{bosoethm} to rearrange the prefactors of the plethysms.
\end{proof}

This immediately implies the following corollary for the limit of infinite alphabet $N\to \infty$:
\begin{cor}
Defining the limiting multi-current ${\mathfrak M}_\al^{\infty}(\bv):=\lim_{N\to \infty} t^{-N \al}\,{\mathfrak M}_\al(\bv)$, we 
have the following plethystic formula:
$$
{\mathfrak M}_\al^{\infty}(\bv)F[X]=\frac{1}{(t-1)^\al}\prod_{1\leq i<j\leq \al} \frac{(v_i-v_j)}{(t -v_i v_j^{-1})(t v_i-v_j)}\prod_{i=1}^\al \frac{C_\infty(t^{-1}v_i)}{C_\infty(v_i)}\cdot  F\left[X+(q-1)\sum_{i=1}^\al \frac{1}{v_i} \right]
$$
\end{cor}

\subsection{Shuffle product}\label{shuprosec}

The shuffle product is a non-commutative product 
$*:{\mathcal F}_\al \times {\mathcal F}_\beta \to {\mathcal F}_{\al+\beta}$, sending $(P,P')\mapsto Q=P*P'$. 
It allows to express the product 
of any two difference operators ${\mathcal D}_\al(P)$ and ${\mathcal D}_\beta(P')$
for arbitrary rational functions $P\in {\mathcal F}_\al $ and $P'\in {\mathcal F}_\beta$, 
as a difference operator of the form
${\mathcal D}_{\al+\beta}(Q)$, for some $Q\in {\mathcal F}_{\al+\beta}$.
In the following sections, we first define the shuffle product, and then present various applications, among which 
alternative proofs of the exchange relation \eqref{exchange}, and of Theorems \ref{Mofm},  \ref{thserrre}
and \ref{polynomialitythm}.

\subsubsection{Definition and main result}

Let us introduce the quantity:
$$\zeta(x):=\frac{1-t x}{1-x}\frac{t-q x}{1-q x}$$

\begin{defn}\label{shufdef}
The shuffle product $P*P'\in {\mathcal F}_{\al+\beta}$ of $P\in {\mathcal F}_\al $ and $P'\in {\mathcal F}_\beta$ 
is defined as the symmetrized expression:
\begin{equation}\label{defshuf}
P*P'(x_1,...,x_{\al+\beta}):=\frac{1}{\al!\, \beta!} {\rm Sym}\left(  P(x_1,...,x_\al)P'(x_{\al+1},...,x_{\al+\beta}) 
\prod_{1\leq i\leq \al<j\leq \al+\beta} 
\zeta(x_i/x_j) \right)
\end{equation}
where the symmetrization $\rm Sym$ is over $x_1,x_2,...,x_{\al+\beta}$.
\end{defn}

This product is in general non-commutative by construction, but associative.

\begin{thm}\label{shufmac}
For any rational functions $P\in {\mathcal F}_\al $ and $P'\in {\mathcal F}_\beta$, we have the relations
\begin{equation}\label{prodshuf}
{\mathcal D}_\al(P)\, {\mathcal D}_\beta(P')={\mathcal D}_{\al+\beta}(P * P'), \qquad 
{\mathcal M}_\al(P)\, {\mathcal M}_\beta(P')={\mathcal M}_{\al+\beta}(P * P')
\end{equation}
with the shuffle product $P*P'\in {\mathcal F}_{\al+\beta}$ defined in \eqref{defshuf}.
\end{thm}
\begin{proof}
We use Theorem \ref{mainthm} and compute, for $\bu=(u_1,u_2,...,u_{\al+\beta})$:
\begin{eqnarray*}
{\mathcal M}_\al(P)\, {\mathcal M}_\beta(P')&=&CT_{\bu}\left( \frac{P(u_1^{-1},...,u_\al^{-1})P'(u_{\al+1}^{-1},...,u_{\al+\beta}^{-1})}{\al!\, \beta!}\times \right. \\
&&\qquad \qquad \qquad \qquad \left. \times \prod_{i<j\,  \in [1,\al]\,
{\rm or} \, [\al+1,\al+\beta] }\zeta(u_j/u_i)^{-1} \prod_{k=1}^{\al+\beta} {\mathfrak m}(u_k)\right) \\
&=&CT_{\bu}\left(  \frac{P(u_1^{-1},...,u_\al^{-1})P'(u_{\al+1}^{-1},...,u_{\al+\beta}^{-1})}{\al!\, \beta!}\times \right. \\
&&\qquad \qquad \qquad \qquad \left. \times\prod_{1\leq i\leq \al<j\leq \al+\beta} \zeta(u_j/u_i)
\prod_{i<j \,\in [1,\al+\beta]}\zeta(u_j/u_i)^{-1}
\prod_{k=1}^{\al+\beta} {\mathfrak m}(u_k)\right)\\
&=& \frac{1}{(\al+\beta)!}CT_{\bu}\left( P*P'(u_1^{-1},...,u_{\al+\beta}^{-1})
\prod_{i<j \,\in [1,\al+\beta]}\zeta(u_j/u_i)^{-1}
\prod_{k=1}^{\al+\beta} {\mathfrak m}(u_k)\right)\\
&=&{\mathcal M}_{\al+\beta}(P*P')
\end{eqnarray*}
where we have used the symmetry of the last factor in the second line, and the fact that the 
constant term is preserved under symmetrization, namely 
$CT_{u_1,...,u_m}\big({\rm Sym}(f(u_1,...,u_{m}))\big)/m! =CT_{u_1,...,u_m}\big(f(u_1,...,u_{m})\big)$.
The Theorem follows.
\end{proof}

In the next subsections, we explore applications of Theorem \ref{shufmac}.

\subsubsection{Application I: Macdonald current relations}\label{appsecone}
Note that
$${\mathfrak m}(v)={\mathcal D}_1(\delta(v x_1))$$
By iterated use of Theorem \ref{shufmac},
we may express any product ${\mathfrak m}(v_1){\mathfrak m}(v_2)\cdots {\mathfrak m}(v_\al)$ as:
$${\mathfrak m}(v_1){\mathfrak m}(v_2)\cdots {\mathfrak m}(v_\al)={\mathcal D}_\al\Big(
\delta(v_1 x_1)*\delta(v_2 x_1)* \cdots *\delta(v_\al x_1)\Big)$$
where we used associativity of the shuffle product to drop parentheses.
In particular, the exchange relation \eqref{exchange} boils down to the following shuffle identity on ${\mathcal F}_2$.

\begin{thm}
We have the relation:
$$g(u,v)\, \delta(u x_1)*\delta(v x_1)+g(v,u)\, \delta(v x_1)*\delta(u x_1)=0$$
\end{thm}
\begin{proof}
Use Definition~\ref{shufdef} to compute:
\begin{eqnarray*}
g(u,v)\, \delta(u x_1)*\delta(v x_1)&=&
g(u,v){\rm Sym}\left( \delta(u x_1)\,\delta(v x_2)\frac{(x_2-tx_1)(t x_2-q x_1)}{(x_2-x_1)(x_2-q x_1)} \right)\\
&=& g(u,v)\frac{(u-tv)(t u-q v)}{(u-v)(u-q v)}{\rm Sym}\left( \delta(u x_1)\,\delta(v x_2) \right)\\
&=& \frac{(v-t u)(u-tv)(t v-q u)(t u-q v)}{q t (u-v)}\left( \delta(u x_1)\,\delta(v x_2)+\delta(v x_1)\,\delta(u x_2)\right)
\end{eqnarray*}
which is manifestly skew-symmetric in $(u,v)$.
\end{proof}

More generally, Theorem \ref{Mofm} translates into the following shuffle identity.
\begin{thm}\label{detshuf}
We have the relation:
$$\frac{\det\left(\left(\delta(x_i v_j)\right)_{1\leq i,j\leq \al}\right)}{\prod_{1\leq i<j\leq \al} (x_i-x_j)}
=\prod_{i=1}^\al v_j^{\al-1}\,
\prod_{1\leq i<j\leq \al} \frac{v_i-q v_j}{(t v_j -v_i)(t v_i-q v_j)}\, \delta(v_1x_1)*\delta(v_2 x_1)*\cdots *\delta(v_\al x_1)$$
\end{thm}
\begin{proof}
We compute:
\begin{eqnarray*}
&&\delta(v_1x_1)*\delta(v_2 x_1)*\cdots *\delta(v_\al x_1)={\rm Sym}\left(
\delta(v_1x_1)\delta(v_2 x_2)\cdots\delta(v_\al x_\al)\prod_{1\leq i<j\leq \al}\frac{(x_j-tx_i)(t x_j-q x_i)}{(x_j-x_i)(x_j-q x_i)}\right)\\
&&\qquad\qquad =\frac{1}{(\prod_{i=1}^\al v_j)^{\al-1}} \prod_{1\leq i<j\leq \al}\frac{(tv_j-v_i)(t v_i-q v_j)}{(v_i-q v_j)} {\rm Sym}\left(
\frac{\delta(v_1x_1)\delta(v_2 x_2)\cdots\delta(v_\al x_\al)}{\prod_{1\leq i<j\leq \al} (x_i-x_j)}\right)
\end{eqnarray*}
and the result follows, as the symmetrization produces the desired determinant.
\end{proof}

Finally let us revisit Theorem \ref{thserrre} for ${\mathfrak m}$, namely the identity 
$${\rm Sym}_{v_1,v_2,v_3}\left( 
\frac{v_2}{v_3} {\Big[}{\mathfrak m}(v_1),{[}{\mathfrak m}(v_2),{\mathfrak m}(v_3){]}{\Big]}
\right)=0$$
The corresponding shuffle identity is the following.
\begin{thm}\label{shufserre}
We have the identity:
\begin{eqnarray*}
&&{\rm Sym}_{v_1,v_2,v_3}\left\{ 
\frac{v_2}{v_3}\Big(\delta(v_1x_1)*\big(\delta(v_2 x_1)*\delta(v_3x_1)-\delta(v_3 x_1)*\delta(v_2x_1)\big)\right. \\
&&\qquad\qquad\qquad \left. -
\big(\delta(v_2 x_1)*\delta(v_3x_1)-\delta(v_3 x_1)*\delta(v_2x_1)\big)*\delta(v_1x_1)\Big)\right\}=0
\end{eqnarray*}
\end{thm}
\begin{proof}
Using the proof of Theorem \ref{detshuf} for $\al=3$, we compute:
$$\delta(v_1x_1)*\delta(v_2 x_1)*\delta(v_3 x_1)=
\prod_{1\leq i<j\leq 3}\frac{(tv_j-v_i)(t v_i-q v_j)}{(v_i-q v_j)(v_j-v_i)} {\rm Sym}_{x_1,x_2,x_3}\left(
\delta(v_1x_1)\delta(v_2 x_2)\delta(v_3 x_3)\right)$$
Noting that the symmetrized term is also symmetric in $(v_1,v_2,v_3)$, the statement of the Theorem
boils down to:
$${\rm Sym}_{v_1,v_2,v_3}\left(\frac{v_2}{v_3}\big(1-(23)-(123)+(13)\big)
\prod_{1\leq i<j\leq 3}\frac{(tv_j-v_i)(t v_i-q v_j)}{(v_i-q v_j)(v_j-v_i)}\right)=0$$
where the permutations on the left act by permuting the $v$'s. This latter identity is easily checked.
\end{proof}

\subsubsection{Application II: commuting difference operators}\label{appsectwo}

The Macdonald operators form a commuting family whose common eigenfunctions 
are the celebrated Macdonald polynomials \cite{macdo}.
Their commutativity boils down via the relations \eqref{prodshuf} to the identity
$$1_\al*1_\beta=1_\beta*1_\al$$
where we denote by $1_\al$ the constant function $1\in {\mathcal F}_\al$.
This identity amounts to the following:
\begin{lemma}\label{comac}
The following identity in ${\mathcal F}_{\al+\beta}$ holds.
\begin{equation}
{\rm Sym}\left(\prod_{1\leq i\leq \al<j\leq \al+\beta} \zeta(x_i/x_j)\right)=
{\rm Sym}\left(\prod_{1\leq i'\leq \beta<j'\leq \al+\beta} \zeta(x_{i'}/x_{j'})\right)
\end{equation}
\end{lemma}
\begin{proof}
The symmetrized function is by definition invariant under any permutation of the $x$'s. 
Consider the permutation $\sigma\in S_{\al+\beta}$ such that $\sigma(i)=\al+\beta+1-i$
for all $i\in \{1,2,...,\al+\beta\}$.
The permutation $\sigma$ clearly maps the range $1\leq i\leq \al<j\leq \al+\beta$ to
$1\leq \sigma(j)\leq \beta< \sigma(i)\leq \al+\beta$, and the identity follows by noting that
the result is invariant under $x_i\mapsto x_i^{-1}$ for all $i$.
\end{proof}

Likewise, using the definition ${\mathcal M}_{\al;n}={\mathcal D}_\al((x_1\cdots x_\al)^n)$,
the commutation of the operators ${\mathcal M}_{\al;n}$ for any fixed $n$ and $\al=1,2,...,N$
is a consequence of the commutation:
$$(x_1\cdots x_\al)^n * (x_1\cdots x_\beta)^n =(x_1\cdots x_\beta)^n * (x_1\cdots x_\al)^n $$
The latter is a trivial consequence of Lemma  \ref{comac}, as the product
$(x_1\cdots x_\al)^n (x_{\al+1}\cdots x_{\al+\beta})^n$ is symmetric in the $x$'s, and therefore 
drops out of the symmetrization.

More generally, we could try to find commuting families of similar difference operators, by looking
for families of shuffle-commuting rational functions. In fact, any factorized choice for
$P(x_1,...,x_\al)=f(x_1)f(x_2)\cdots f(x_\al)$ and similarly for $P'(x_1,...,x_\beta)=f(x_1)f(x_2)\cdots f(x_\beta)$
will trivially lead to commuting operators from $P*P'=P'*P$,
as the product $f(x_1)\cdots f(x_{\al+\beta})$ drops out of the 
symmetrization and leaves us with the identity of Lemma \ref{comac}. 
The previous case corresponds to $f(x)=x^n$. 

Consider now the function $f(x)=1+a x$ for some fixed arbitrary coefficient $a$. Note that
$$\prod_{i=1}^\al f(x_i)=\sum_{k=0}^\al a^k s_{1^k,0^{\al-k}}(x_1,...,x_\al)$$
where the notation $1^k0^{\al-k}$ stands for $k$ 1's followed by $\al-k$ 0's. We deduce that
$${\mathcal D}_\al(f(x_1)\cdots f(x_\al))=\sum_{k=0}^\al a^k {\mathcal M}_{1^k,0^{\al-k}}$$
Expressing the commutation relation 
$[{\mathcal D}_\al(f(x_1)\cdots f(x_\al)),{\mathcal D}_\beta(f(x_1)\cdots f(x_\beta))]=0$
and identifying the coefficient of $a^k$ yields the following relations between the ${\mathcal M}$'s:
$$\sum_{\ell={\rm Max}(0,k-\beta)}^{{\rm Min}(k,\al)} 
[{\mathcal M}_{1^\ell,0^{\al-\ell}},{\mathcal M}_{1^{k-\ell},0^{\beta-k+\ell}}]=0\qquad (0\leq k\leq \al+\beta)$$
For $\al=1$, $\beta=2$, this reduces for instance to:
$$[{\mathcal M}_0,{\mathcal M}_{0,0}]=0, 
\ \ [{\mathcal M}_0,{\mathcal M}_{1,0}]+[{\mathcal M}_1,{\mathcal M}_{0,0}]=0, \ \ 
[{\mathcal M}_0,{\mathcal M}_{1,1}]+[{\mathcal M}_1,{\mathcal M}_{1,0}]=0, 
\ \  [{\mathcal M}_1,{\mathcal M}_{1,1}]=0\ .
$$

\subsubsection{Application III: proving relations between the difference operators}\label{appsecthree}

Another application of Theorem \ref{shufmac} consists in using shuffle identities to prove identities between 
difference operators, the former being much simpler than the latter. Let us illustrate this on the following
examples, which allow to express any operator ${\mathcal M}_{n,p}$ as a quadratic polynomial 
of the ${\mathcal M}_{n}$'s (an alternative expression to that of Theorem \ref{polynomialitythm}). 
We denote by $[x,y]_q=xy-q yx$ the $q$-commutator of $x,y$.

We shall proceed in several steps. First we note that, due to the skew-symmetry of the determinantal 
definition of the generalized Schur function \eqref{defschur}, we have $s_{n,p}=-s_{p-1,n+1}$, and therefore
${\mathcal M}_{n,p}=-{\mathcal M}_{p-1,n+1}$. This allows to restrict ourselves to $n\geq p$, as
${\mathcal M}_{n,n+1}=0$. In the two following Theorems \ref{oddMthm} and \ref{evenMthm}, we find explicit 
expressions for ${\mathcal M}_{n+k,n}$,
respectively for $k>0$ odd and even. Finally, we express ${\mathcal M}_{n,n}$ in Lemma \ref{Mnn}, which allows to
complete the Theorem, with explicit quadratic expressions for all ${\mathcal M}_{n,p}$.

\begin{thm}\label{oddMthm}
We have the relation:
\begin{equation}\label{theone}
(1-q)t \,{\mathcal M}_{n+2k+1,n}= \sum_{\ell=0}^k \nu_{n+2\ell,n+2k-2\ell+1}
%{\mathcal M}_{n+2k-2\ell}\,{\mathcal M}_{n+2\ell+1}-q\,
%{\mathcal M}_{n+2k-2\ell+1}\,{\mathcal M}_{n+2\ell} \label{theone} 
=\sum_{\ell=0}^k \left[ {\mathcal M}_{n+2\ell},{\mathcal M}_{n+2k-2\ell+1}\right]_q 
\end{equation}
\end{thm}
\begin{proof}
The shuffle relation corresponding to \eqref{theone} reads, in ${\mathcal F}_2$:
\begin{eqnarray*}
\sum_{\ell=0}^k (x_1)^{n+2k-2\ell} * (x_1)^{n+2\ell+1}-q (x_1)^{n+2k-2\ell+1} * (x_1)^{n+2\ell}
&=&(1-q)t \,s_{n+2k+1,n}(x_1,x_2)\\
&=&(1-q)t\,  (x_1x_2)^n\, s_{2k+1,0}(x_1,x_2)
\end{eqnarray*}
Explicitly, we write
\begin{eqnarray*}
&&{\rm Sym}\left( \Big(\sum_{\ell=0}^k (x_1)^{n+2k-2\ell} (x_2)^{n+2\ell+1}-q (x_1)^{n+2k-2\ell+1} (x_2)^{n+2\ell}\Big)\frac{(x_2-t x_1)(t x_2-q x_1)}{(x_2-x_1)(x_2-q x_1)}\right) \\
&&\qquad =\, 
(x_1x_2)^n {\rm Sym}\left(\frac{(x_2^2)^{k+1}-(x_1^2)^{k+1}}{x_2^2-x_1^2} 
\frac{(x_2-t x_1)(t x_2-q  x_1)}{x_2-x_1}\right)\\
&&\qquad =\,\frac{(x_1x_2)^n}{(x_2-x_1)(x_2^2-x_1^2)} 
{\rm Sym}\left(((x_2)^{2k+2}-(x_1)^{2k+2})(x_2-t x_1)(t x_2-q  x_1)\right)\\
&&\qquad =(1-q)t \, (x_1x_2)^n\, s_{2k+1,0}(x_1,x_2)
\end{eqnarray*}
where we identified $s_{a,0}(x_1,x_2)=(x_1^{a+1}-x_2^{a+1})/(x_1-x_2)$.
The Theorem follows.
\end{proof}

This gives an explicit expression for ${\mathcal M}_{n+2k+1,n}$ for $k\geq 0$.

\begin{thm}\label{evenMthm}
We have the relations, for all $k\geq 0$:
\begin{eqnarray}
(1-q)t\, ({\mathcal M}_{n+4k,n}-{\mathcal M}_{n+2k,n+2k})
&=&
\sum_{\ell=0}^{k-1} \nu_{n+2\ell,n+4k-2\ell}+\nu_{n+4k-1-2\ell,n+2\ell+1}
%\sum_{\ell=0}^{k-1} \left\{ {\mathcal M}_{n+4k-1-2\ell} {\mathcal M}_{n+2\ell+1}
%-q {\mathcal M}_{n+4k-2\ell}{\mathcal M}_{n+2\ell} \right.\nonumber  \\
%&&\qquad \left. +{\mathcal M}_{n+2\ell}{\mathcal M}_{n+4k-2\ell} 
%-q {\mathcal M}_{n+2\ell+1}{\mathcal M}_{n+4k-1-2\ell}\right\}
\nonumber \\
&=&
\sum_{\ell=0}^{k-1} \left[ {\mathcal M}_{n+2\ell},{\mathcal M}_{n+4k-2\ell} \right]_q+
\left[{\mathcal M}_{n+4k-1-2\ell}, {\mathcal M}_{n+2\ell+1}\right]_q \label{zeromodfour} \\
(1-q)t\, ({\mathcal M}_{n+4k+2,n}+{\mathcal M}_{n+2k+1,n+2k+1})
&=&
\sum_{\ell=0}^{k} \nu_{n+2\ell,n+4k+2-2\ell}+\nu_{n+4k+1-2\ell,n+2\ell+1}
%\sum_{\ell=0}^{k}\left\{  {\mathcal M}_{n+4k+1-2\ell}\,{\mathcal M}_{n+2\ell+1}
%-q \, {\mathcal M}_{n+4k+2-2\ell}\,{\mathcal M}_{n+2\ell}\right. \nonumber \\
%&&\qquad \left. +{\mathcal M}_{n+2\ell}{\mathcal M}_{n+4k+2-2\ell} 
%-q \, {\mathcal M}_{n+2\ell+1}{\mathcal M}_{n+4k+1-2\ell} \right\}
\nonumber \\
&=& \sum_{\ell=0}^k \left[{\mathcal M}_{n+2\ell},{\mathcal M}_{n+4k+2-2\ell}\right]_q
+\left[{\mathcal M}_{n+4k+1-2\ell},{\mathcal M}_{n+2\ell+1}\right]_q \label{twomodfour}
\end{eqnarray}
\end{thm}
\begin{proof}
The shuffle relations corresponding to \eqref{zeromodfour} and \eqref{twomodfour} read respectively in $\cF_2$:
\begin{eqnarray*}
&&\sum_{\ell=0}^{k-1} \left\{ (x_1)^{n+4k-1-2\ell} * (x_1)^{n+2\ell+1}-q (x_1)^{n+4k-2\ell} * (x_1)^{n+2\ell}
\right.\\
&&\qquad \qquad \left. 
+(x_1)^{n+2\ell} * (x_1)^{n+4k-2\ell}-q (x_1)^{n+2\ell+1} * (x_1)^{n+4k-1-2\ell} \right\}\\
&&\qquad \qquad\qquad\qquad\qquad\qquad\qquad = (1-q)t \big(s_{n+4k,n}(x_1,x_2)-s_{n+2k,n+2k}(x_1,x_2)\big)\\
&&\sum_{\ell=0}^{k} \left\{ (x_1)^{n+4k+1-2\ell} * (x_1)^{n+2\ell+1}-q (x_1)^{n+4k+2-2\ell} * (x_1)^{n+2\ell}\right.\\
&&\qquad \qquad \left. +(x_1)^{n+2\ell} * (x_1)^{n+4k+2-2\ell}-q (x_1)^{n+2\ell+1} * (x_1)^{n+4k+1-2\ell} \right\}\\
&&\qquad \qquad\qquad\qquad\qquad\qquad\qquad =(1-q)t \big(s_{n+4k+2,n}(x_1,x_2)+s_{n+2k+1,n+2k+1}(x_1,x_2)\big)
\end{eqnarray*}
Explicitly, we write:
\begin{eqnarray*}
&&{\rm Sym}\left( \Big(\sum_{\ell=0}^{k-1} (x_1)^{n+4k-1-2\ell} (x_2)^{n+2\ell}+(x_1)^{n+2\ell}(x_2)^{n+4k-1-2\ell}\Big)
\frac{(x_2-t x_1)(t x_2-q x_1)}{x_2-x_1}\right) \\
&&=\, 
t(1-q)(x_1+x_2) (x_1x_2)^{n} \Big(\sum_{\ell=0}^{k-1} (x_1)^{4k-1-2\ell} (x_2)^{2\ell}
+(x_1)^{2\ell}(x_2)^{4k-1-2\ell}\Big)\\
&&=\, 
t(1-q)(x_1x_2)^{n} \Big(\sum_{\ell=0}^{k-1} (x_1)^{4k-2\ell} (x_2)^{2\ell}+(x_1)^{2\ell}(x_2)^{4k-2\ell} \\
&&\qquad \qquad\qquad\qquad\qquad+(x_1)^{2\ell+1}(x_2)^{4k-1-2\ell}+(x_1)^{4k-1-2\ell}(x_2)^{2\ell+1}\Big)\\
&&\qquad \qquad =t(1-q)\, (x_1x_2)^{n}  \Big(\sum_{\ell=0}^{4k}(x_1)^{4k-\ell} (x_2)^{\ell} -(x_1x_2)^{2k}\Big)\\
&&\qquad \qquad =t(1-q)\, \Big(s_{n+4k,n}(x_1,x_2)-s_{n+2k,n+2k}(x_1,x_2)\Big)
\end{eqnarray*}
Analogously, we have:
\begin{eqnarray*}
&&{\rm Sym}\left( \Big(\sum_{\ell=0}^{k} (x_1)^{n+4k+1-2\ell} (x_2)^{n+2\ell}+(x_1)^{n+2\ell}(x_2)^{n+4k+1-2\ell}\Big)
\frac{(x_2-t x_1)(t x_2-q x_1)}{x_2-x_1}\right) \\
&&=\, 
t(1-q)(x_1+x_2) (x_1x_2)^{n} \Big(\sum_{\ell=0}^{k} (x_1)^{4k+1-2\ell} (x_2)^{2\ell}
+(x_1)^{2\ell}(x_2)^{4k+1-2\ell}\Big)\\
&&=\, 
t(1-q)(x_1x_2)^{n} \Big(\sum_{\ell=0}^{k} (x_1)^{4k+2-2\ell} (x_2)^{2\ell}+(x_1)^{4k+1-2\ell} (x_2)^{2\ell+1}\\
&&\qquad \qquad\qquad\qquad\qquad+(x_1)^{2\ell+1}(x_2)^{4k+1-2\ell}+(x_1)^{2\ell}(x_2)^{4k+2-2\ell}\Big)\\
&&\qquad \qquad =t(1-q)\, (x_1x_2)^{n}  \Big(\sum_{\ell=0}^{4k+2}(x_1)^{4k+2-\ell} (x_2)^{\ell} +(x_1x_2)^{2k+1}\Big)\\
&&\qquad \qquad =t(1-q)\, \Big(s_{n+4k+2,n}(x_1,x_2)+s_{n+2k+1,n+2k+1}(x_1,x_2)\Big)
\end{eqnarray*}
This completes the proof of the shuffle relations and the Theorem follows.
\end{proof}

To get an expression for ${\mathcal M}_{n+2k,n}$ for all $k>0$, we also need to compute ${\mathcal M}_{n,n}$.

\begin{lemma}\label{evenMlemma}
We have the relations:
\begin{eqnarray}
({\mathcal M}_n)^2 -q\, {\mathcal M}_{n+1}{\mathcal M}_{n-1}
&=&(q+t+t^2)\,{\mathcal M}_{n,n}-q t\, {\mathcal M}_{n+1,n-1} \label{Mevenone}\\
({\mathcal M}_n)^2 -q^{-1} {\mathcal M}_{n-1}{\mathcal M}_{n+1}
&=&(1+t+q^{-1}t^2)\,{\mathcal M}_{n,n}-q^{-1} t \,{\mathcal M}_{n+1,n-1} \label{Meventwo}
\end{eqnarray}
Equivalently, we have:
\begin{equation}\label{Mnn}  {\mathcal M}_{n,n}
=\frac{(1-q^2)({\mathcal M}_{n})^2+q \left[{\mathcal M}_{n-1},{\mathcal M}_{n+1}\right]}{(1-q)(1+t)(q+t)}
\end{equation}
and
\begin{equation}\label{Mnnplus}  {\mathcal M}_{n+1,n-1}
=\frac{(1-q)(q+t^2)({\mathcal M}_{n})^2+(1+t)(q+t)\left[{\mathcal M}_{n-1},{\mathcal M}_{n+1}\right]_q
-t q \left[{\mathcal M}_{n-1},{\mathcal M}_{n+1}\right]}{(1-q)(1+t)(q+t)}
\end{equation}
\end{lemma}
\begin{proof}
The shuffle relations corresponding to \eqref{Mevenone} and \eqref{Meventwo} read
in ${\mathcal F}_2$:
\begin{eqnarray*} 
(x_1)^n * (x_1)^{n}-q (x_1)^{n+1} * (x_1)^{n-1}&=& (q+t+t^2)s_{n,n}(x_1,x_2)-q t \,s_{n+1,n-1}(x_1,x_2)\\
(x_1)^n * (x_1)^{n}-q (x_1)^{n-1} * (x_1)^{n+1}&=& (1+t+q^{-1}t^2)s_{n,n}(x_1,x_2)-q^{-1} t \,s_{n+1,n-1}(x_1,x_2)
\end{eqnarray*}
%with $s_{n+1,n-1}(x_1,x_2)=(x_1x_2)^{n-1}(x_1^2+x_1x_2+x_2^2)$ and $s_{n,n}=(x_1x_2)^n$.
These read respectively:
\begin{eqnarray*}&&{\rm Sym}\left( ( x_1^n x_2^{n}-q x_1^{n+1}x_{2}^{n-1})
\frac{(x_2-t x_1)(t x_2-q x_1)}{(x_2-x_1)(x_2-q x_1)} \right)\\
&&\qquad=(x_1x_2)^{n-1}{\rm Sym}\left(\frac{x_1(x_2-t x_1)(t x_2-q x_1)}{x_2-x_1}\right)\\
&&\qquad=(x_1x_2)^{n-1}((q+t+t^2)x_1x_2-q t (x_1^2+x_1 x_2+x_2^2))
\end{eqnarray*}
and
\begin{eqnarray*}&&{\rm Sym}\left( ( x_1^n x_2^{n}-q^{-1} x_1^{n-1}x_{2}^{n+1})
\frac{(x_2-t x_1)(t x_2-q x_1)}{(x_2-x_1)(x_2-q x_1)} \right)\\
&&\qquad=-q^{-1}(x_1x_2)^{n-1}{\rm Sym}\left(\frac{x_2(x_2-t x_1)(t x_2-q x_1)}{x_2-x_1}\right)\\
&&\qquad=-q^{-1}(x_1x_2)^{n-1}(t (x_1^2+x_1 x_2+x_2^2)-(q+qt+t^2)x_1x_2)
\end{eqnarray*}
and eqs.~\ref{Mevenone} and \ref{Meventwo} follow from the values $s_{n,n}=(x_1x_2)^n$ and $s_{n+1,n-1}=(x_1x_2)^{n-1}(x_1^2+x_1x_2+x_2^2)$. Finally eq.~\eqref{Mnn} follows from the combination \eqref{Mevenone}$- q\, $\eqref{Meventwo}.
\end{proof}

Substitution of the expression \eqref{Mnn} into the relations (\ref{zeromodfour}-\ref{twomodfour}) yields alternative polynomial expressions to those of Theorem \ref{polynomialitythm}.

%\begin{example}
%Theorem \ref{oddMthm} for $k=0,1$ gives:
%\begin{eqnarray}
%(1-q)t\, {\mathcal M}_{n+1,n}&=&\nu_{n,n+1}\label{nnplusone}\\
%(1-q)t\, {\mathcal M}_{n+3,n}&=&\nu{n,n+3}+\nu_{n+2,n+1}\label{nnplusthree}
%\end{eqnarray}
%with $\nu_{a,b}$ as in \eqref{qdetwo}.
%\end{example}
%Theorem \ref{evenMthm} for $k=0,1$ gives:
%\begin{eqnarray}
%(1-q)t\, {\mathcal M}_{n+2,n}&=&\nu_{n,n+2}+\nu_{n+1,n+1}-(1-q)t\, {\mathcal M}_{n+1,n+1}\label{nnplustwo}\\
%(1-q)t\, {\mathcal M}_{n+4,n}&=&\nu{n,n+4}+\nu_{n+3,n+1}-(1-q)t\, {\mathcal M}_{n+2,n+2}\label{nnplusfour}
%\end{eqnarray}
%where we have used \eqref{Mnn}.

%\section{Equivalence between the DAHA and quantum Toroidal deformations}
%\input{daha-qtor}
\section{Elliptic Hall algebra and $q,t$-deformed $Q$-system relations}\label{EHAsec}
In this section, we use the known relation between DAHA and Elliptic Hall Algebra
to obtain new relations between our generalized Macdonald operators. 

\subsection{Elliptic Hall Algebra: definition and isomorphisms}

In this section, we recall known results from \cite{BS,SCHIFFEHA}.
Let
\begin{equation}\label{defalpha}
\al_k:=\frac{1}{k}(1-q^k)(1-t^{-k})(1-q^{-k}t^k) ,\qquad (k\in \Z^*) .
\end{equation}

\begin{defn}\label{EHAdefn}
The Elliptic Hall Algebra (EHA) has generators $u_{a,b},\theta_{a,b}$, $a,b\in \Z^*$,
subject to the commutations:
$$\!\!\!\!\!\!\!\!\!\!(E1)\quad [u_{c,d},u_{a,b}]=0 \qquad \hbox{\rm if (0,0), (a,b), (c,d)  are aligned}
$$
and
$$\qquad\qquad (E2)\quad  [u_{c,d},u_{a,b}]=\frac{\epsilon_{a,b;c,d}}{\al_1}\theta_{a+c,b+d}, \quad {\rm if}\, gcd(a,b)=1\quad 
{\rm and}\quad \Delta_{a,b;c,d}=\emptyset$$
where $\Delta_{a,b;c,d}$ is the intersection with $\Z^2$ of the strict interior of the triangle $(0,0),(a,b),(a+c,b+d)$,
and $\epsilon_{a,b;c,d}={\rm sgn}(ad-bc)$.
The generators $\theta$ and $u$ for non-coprime indices are further related via:
$$(E3)\quad 1+\sum_{n=1}^\infty \theta_{n(a,b)} z^n =e^{\sum_{k=1}^\infty \al_k u_{k(a,b)} z^k}\qquad (gcd(a,b)=1)$$
and in particular $\theta_{n,p}=\al_1\,u_{n,p}$ whenever $gcd(n,p)=1$.
\end{defn}

The EHA is isomorphic to the quantum toroidal algebra of Section \ref{qtorsec} via the following
assignments:
\begin{equation}\label{ehatoqtor}
e(z)=\sum_{n\in\Z} u_{1,n} \,z^n,\quad f(z)=\sum_{n\in\Z} u_{-1,n} \,z^n,\quad 
\psi_{\pm}(z)=1+\sum_{n\geq 1} \theta_{0,\pm n}\, z^{\pm n}
\end{equation}

There is a natural action of $SL(2,\Z)$ on the generators $u_{a,b}$ given by:
$$ \begin{pmatrix} a_0 & a_1 \\ a_2 & a_3 \end{pmatrix}\cdot u_{a,b}=u_{a_0 a+a_2 b,a_1 a+a_3 b}$$
For further use, we single out the two generators:
\begin{equation}\label{defTU}
T=\begin{pmatrix} 1 & 1 \\ 0 & 1 \end{pmatrix}\qquad {\rm and}\qquad U=\begin{pmatrix} 1 & 0 \\ 1 & 1 \end{pmatrix}
\end{equation}
acting on the EHA generators respectively as $T \, u_{a,b}=u_{a,a+b}$ and $U\, u_{a,b}=u_{a+b,b}$.

In \cite{SHIVAS}, the authors constructed an isomorphism between the EHA and the Spherical DAHA 
(in infinitely many variables $X_i$, $i\in \N$). In particular, the natural $SL(2,\Z)$ action on the EHA
maps onto the natural $SL(2,\Z)$ action on the DAHA \cite{Cheredbook}, namely: 
$$T\mapsto \tau_+={\rm ad}_{\gamma^{-1}}, \quad {\rm and}\quad U\mapsto \tau_-={\rm ad}_{\eta^{-1}}$$
with $\gamma$ and $\eta$ as in \eqref{gammadefn} and \eqref{etadefn}, and Lemmas \ref{taupluslemma} and
\ref{taumoinslemma} respectively.

The isomorphism maps the generators $u_{k,0}$
to the power sums $P_k:=\sum_i (Y_i)^k$ operators in the spherical version of DAHA,
which is respected by the functional representation of Sect. \ref{secpol} above.

\subsection{EHA representation via generalized Macdonald operators}

We are now ready to complete the identification of generators of the EHA in terms of generalized 
Macdonald operators.

%\color{red} check involutive isomorphisms of DAHA in Cherednik's book like:
%$$\begin{pmatrix}\end{pmatrix}:\ (X_i\to X_i, Y_i\to Y_i^{-1},q\to q^{-1},t\to t^{-1})$$
%$$(X_i\to X_i^{-1}, Y_i\to Y_i^{-1},T_i\to T_i^{-1},q\to q,t\to t)$$
%$$\begin{pmatrix}0 & -1 \\ -1 & 0\end{pmatrix}:\ (X_i\to Y_i, Y_i\to X_i,T_i\to T_i^{-1},q\to q^{-1},t\to t^{-1})$$
%\color{black}

Comparing \eqref{ehatoqtor}
with the representation of Theorem \ref{gentoro} involving generalized Macdonald operators, 
we easily deduce the following representation of the EHA, or rather a quotient thereof corresponding 
to the finite number $N$ of variables\footnote{These assignments first appeared in \cite{Miki07} 
Proposition 3.3, with different notations $(q,\gamma,y_i)\mapsto (\theta=t^{1/2},q/\theta,x_i)$.}:
\begin{equation} u_{1,n}=\frac{q^{1/2}}{1-q} q^{\frac{n}{2}}\,{\mathcal D}_{1;n}^{q,t}
=\frac{q^{1/2}}{1-q} q^{\frac{n}{2}}t^{\frac{1-N}{2}}\,{\mathcal M}_{n},\qquad   
u_{-1,n}=\frac{q^{-1/2}}{1-q^{-1}} q^{-\frac{n}{2}}\,{\mathcal D}_{1;n}^{q^{-1},t^{-1}},
\end{equation}
whereas the relation:
$$
1+\sum_{n\geq 1} \theta_{0,\pm n}\, z^{\pm n}
=e^{\sum_{k\geq 1} \frac{z^{\pm k}}{k}p_{\pm k} q^{\frac{k}{2}}(1-t^{-k})(1-t^kq^{-k})}
$$
fixes the values:
\begin{equation} u_{0,\pm k}= \frac{q^{\frac{k}{2}}}{(1-q^k)} p_{\pm k}\qquad (k\in \Z_{>0})
\end{equation}

Moreover, from the spherical DAHA homomorphism, it is easy to identify the $u_{k,0}$ from
the DAHA power sum operators, which in the functional representation are:
$${\mathcal P}_k := \sum_{i=1}^N (Y_i)^k \vert_{{\mathcal S}_N} $$
and are related for $k>0$ to the Macdonald operators ${\mathcal D}_\al\equiv {\mathcal D}_{\al;0}^{q,t}$ via:
\begin{equation}\label{posiP}\prod_{i=1}^N(1-z Y_i)\vert_{{\mathcal S}_N}=\sum_{\al=0}^N (-z)^\al {\mathcal D}_\al =
e^{-\sum_{k\geq 1} {\mathcal P}_k \frac{z^k}{k}} 
\end{equation}
while for $k<0$ they are expressed in terms of the dual operators 
$\bar{\mathcal D}_\al\equiv {\mathcal D}_{\al;0}^{q^{-1},t^{-1}}$ via:
\begin{equation}\label{negaP}\sum_{\al=0}^N (-z)^\al \bar{\mathcal D}_\al =
e^{-\sum_{k\geq 1} {\mathcal P}_{-k} \frac{z^k}{k}} 
\end{equation}

Finally from the action of $-\epsilon$ which maps $u_{0,k}$ to $u_{k,0}$, we have the identification:
\begin{equation} u_{\pm k,0}= \frac{q^{\frac{k}{2}}}{(1-q^k)} {\mathcal P}_{\pm k}\qquad (k\in \Z_{>0})
\end{equation}

\begin{remark}
We continue the comparison with the work of \cite{BGLX} started in Remarks \ref{nablarem} and
\ref{betternablarem}. In this paper a connection to the positive quadrant of the Elliptic Hall Algebra
(with generators $u_{a,b}\to Q_{a,b}$ with $a,b\geq 0$) was found by setting
$$Q_{1,k}= D_k \quad (k\geq 0), \quad {\rm and} \quad Q_{0,1}=-e_1=-p_1 .$$
for $k\geq 0$,
where the $D_k$'s are defined via their generating current in \eqref{fromBG}, and extending the identification
via the $SL(2,Z)$ action on the EHA. 
By Theorem \ref{BGconnect}, we have the relations
\begin{eqnarray*}
\lim_{N\to \infty} t^{1-N}\, {\mathcal M}_{n}&=& \frac{1}{1-t^{-1}}(\Sigma^{-1}\, D_n\, \Sigma)\vert_{t\to t^{-1}} \\
\lim_{N\to \infty} -p_1(x_1,x_2,...,x_N) &=&-p_1(x_1,x_2,...)
\end{eqnarray*}
Comparing with our functional representation of EHA in the limit of infinite number of variables of $x_1,x_2,...$, we find:
$$ \lim_{N\to \infty} t^{\frac{1-N}{2}}\, u_{1,k}=
\frac{q^{\frac{k+1}{2}}}{(1-q)(1-t^{-1})}\,(\Sigma^{-1}\, Q_{1,k}\, \Sigma)\vert_{t\to t^{-1}}$$
and
$$\lim_{N\to \infty} u_{0,1}=\frac{q^{\frac{1}{2}}}{(1-q)(1-t^{-1})}\, (\Sigma^{-1}\, (-p_1)\, \Sigma)\vert_{t\to t^{-1}}$$
Moreover, by using \eqref{nabnab}, and $\nabla^{(N)}u_{a,b}{\nabla^{(N)}}^{-1}=u_{a+b,b}$ we have:
\begin{eqnarray*}
&&\!\!\!\!\!\!\!\!\!\!\!\! \lim_{N\to \infty} t^{\frac{(1-N)(k+1)}{2}} \nabla^{(N)}u_{1,k}{\nabla^{(N)}}^{-1}=
\lim_{N\to \infty} t^{\frac{(1-N)(k+1)}{2}} u_{k+1,k}\\
&=&\lim_{N\to \infty}  t^{\frac{(1-N)(k+1)}{2}}(t^{\frac{N-1}{2}}q^{\frac{1}{2}})^{d}\, (\Sigma^{-1}\, \nabla\, \Sigma)\vert_{t\to t^{-1}} u_{1,k} (\Sigma^{-1}\, \nabla^{-1}\, \Sigma)\vert_{t\to t^{-1}}\,(t^{\frac{N-1}{2}}q^{\frac{1}{2}})^{-d}\\
&=&\frac{q^{\frac{k+1}{2}}}{(1-q)(1-t^{-1})}\lim_{N\to \infty} t^{\frac{(1-N)(k)}{2}} (t^{\frac{N-1}{2}}q^{\frac{1}{2}})^{d}\, (\Sigma^{-1}\, \nabla\, Q_{1,k}\, \nabla^{-1} \Sigma)\vert_{t\to t^{-1}} 
\,(t^{\frac{N-1}{2}}q^{\frac{1}{2}})^{-d}
\end{eqnarray*}
Note that the action of $Q_{1,k}$ on ${\tilde H}_\lambda$ is a linear combination of terms in which $k$ boxes are added
to $\lambda$, and similarly for $u_{1,k}$  acting on $P_\lambda$. We deduce that the conjugation
by $(t^{\frac{N-1}{2}}q^{\frac{1}{2}})^{d}$ amounts to a factor $(t^{\frac{N-1}{2}}q^{\frac{1}{2}})^{k}$, and using
$\nabla Q_{a,b}\nabla^{-1}=Q_{a+b,b}$ (from \cite{BGLX}) we finally get:
$$\lim_{N\to \infty} t^{\frac{(1-N)(k+1)}{2}}u_{k+1,k}=\frac{q^{\frac{2k+1}{2}}}{(1-q)(1-t^{-1})}
(\Sigma^{-1}\, Q_{k+1,k} \Sigma)\vert_{t\to t^{-1}} 
$$
Repeating the inductive proof of \cite{BGLX}, we arrive in general at:
$$ \lim_{N\to \infty} t^{\frac{(1-N)a}{2}}u_{a,b}=\frac{q^{\frac{a+b}{2}}}{(1-q)(1-t^{-1})}
(\Sigma^{-1}\, Q_{a,b} \Sigma)\vert_{t\to t^{-1}} $$
for any coprime $(a,b)$.
%In particular, taking $a=0$, $b=k$, we identify
%$$\lim_{N\to \infty} p_k= \frac{1-q^k}{(1-q)(1-t^{-1})}\,  (\Sigma^{-1}\, Q_{0,k} \Sigma)\vert_{t\to t^{-1}}$$
\end{remark}

\subsection{EHA and relations between generalized Macdonald operators}

The EHA representation provides us with an alternative way of finding relations 
between the generalized Macdonald operators. We first concentrate on an alternative version
of Theorem \ref{Mofm}, aimed at expressing the generalized Macdonald operator ${\mathcal M}_{\al;n}$
as an explicit polynomial of the ${\mathcal M}_i$'s.

%The main tool is the identification of the $SL(2,\Z)$ generator $T\equiv {\rm ad}_{\gamma^{-1}}$.

\begin{thm}\label{polpol}
The operator ${\mathcal D}_{\al;n}$ is expressible as a homogeneous polynomial of degree $\al$ in the variables
${\mathcal D}_{1;n},{\mathcal D}_{1;n\pm 1}$, with coefficients in $\C(q,t)$.
\end{thm}
\begin{proof}
Let us first use the definition of EHA to compute $\theta_{n,0}$ as an iterated commutator. We first note that
\begin{eqnarray*}u_{n-1,1}&=&[u_{n-2,1},u_{1,0}]=\cdots=\big[ [\cdots [ u_{1,1},u_{1,0}],u_{1,0}]\cdots , u_{1,0}\big]\\
&=&\frac{q^{\frac{n}{2}}}{(1-q)^{n-1}}
\big[ [\cdots [ {\mathcal D}_{1;1},{\mathcal D}_{1;0}],{\mathcal D}_{1;0}]\cdots , {\mathcal D}_{1;0}\big]
\end{eqnarray*}
Moreover, we have:
\begin{equation}\label{thetanzero}
\theta_{n,0}=\al_1\, [u_{n-1,1},u_{1,-1}]=\frac{\al_1\, q^{\frac{n}{2}}}{(1-q)^n}\,
\big[\cdots [{\mathcal D}_{1;1},{\mathcal D}_{1;0}],{\mathcal D}_{1;0}],\ldots ,{\mathcal D}_{1;0}],{\mathcal D}_{1;-1}\big]
\end{equation}
for $n\geq 2$, while $\theta_{1,0}=\al_1 u_{1,0}=\frac{\al_1\, q^{\frac{1}{2}}}{(1-q)}\,{\mathcal D}_{1;0}$.
On the other hand, $\theta_{n,0}$ is related to the $u_{k,0}$ via the relation $(E3)$ of Definition \ref{EHAdefn} above:
\begin{equation}\label{Ptotheta}
1+\sum_{n>0} \theta_{n,0} z^n=e^{\sum_{k>0} \al_k u_{k,0} z^k} =
e^{\sum_{k>0} q^{k/2} (1-t^{-k})(1-q^{-k}t^k){\mathcal P}_k\frac{z^k}{k}}
\end{equation}
Eliminating ${\mathcal P}_k$ between this and \eqref{posiP}, we are left with an algebraic relation of the form:
$$ {\mathcal D}_{\al;0}=\varphi(\theta_{1,0},\theta_{2,0},...,\theta_{\al,0})$$
where $\varphi$ is a quasi-homogeneous polynomial of total degree $\al$ 
(assuming $\theta_{i,0}$ has degree $i$), with coefficients in $\C(q,t)$. 
Combining this with \eqref{thetanzero}, we deduce that 
${\mathcal D}_{\al;0}$ is a homogeneous polynomial of  degree $\al$ in the three variables
$({\mathcal D}_{1;1},{\mathcal D}_{1;0},{\mathcal D}_{1;-1})$, with coefficients in $\C(q,t)$. 
This proves the Theorem for $n=0$. 
To get to arbitrary $n$, we simply have to iteratively conjugate by $\gamma^{-1}$ and use  ${\mathcal D}_{\al;n}=
q^{\al n/2} \gamma^{-n} {\mathcal D}_{\al;0}\gamma^n$.
\end{proof}

\begin{cor}\label{polpolcor}
The operator ${\mathcal M}_{\al;n}$ is expressible as a homogeneous polynomial of degree $\al$ in the variables
${\mathcal M}_{n},{\mathcal M}_{n\pm 1}$, with coefficients in $\C(q,t)$.
\end{cor}
\begin{proof}
The result for ${\mathcal D}_{\al;n}$ of Theorem \ref{polpol}
immediately translates into that for ${\mathcal M}_{\al;n}$, due to the relations
${\mathcal D}_{\al;n}=t^{\al(\al-N)/2}{\mathcal M}_{\al;n}$.
\end{proof}

\begin{example}
Let us write the explicit polynomial relation of Theorem \ref{polpol} above for $\al=2,3$.
We start by expressing ${\mathcal D}_{i;0}$'s in terms of ${\mathcal P}_i$'s via \eqref{posiP}, which gives:
\begin{equation}\label{initexe}
 {\mathcal D}_{1;0}={\mathcal P}_1,\quad {\mathcal D}_{2;0}=\frac{{\mathcal P}_1^2-{\mathcal P}_2}{2},
\quad {\mathcal D}_{3;0}=\frac{{\mathcal P}_1^3-3 {\mathcal P}_1{\mathcal P}_2+2{\mathcal P}_3}{6} \ .
\end{equation}
Next, we express the ${\mathcal P}_i$'s in terms of the $\theta_{i,0}$'s via \eqref{Ptotheta}, which gives:
$${\mathcal P}_1=\frac{\theta_{1,0}}{q^{1/2}(1-t^{-1})(1-q^{-1}t)}, \ 
{\mathcal P}_2=\frac{2\theta_{2,0}-\theta_{1,0}^2}{q(1-t^{-2})(1-q^{-2}t^2)}, \ 
{\mathcal P}_3=\frac{3\theta_{3,0}-3\theta_{1,0}\theta_{2,0}+\theta_{1,0}^3}{q^{3/2}(1-t^{-3})(1-q^{-3}t^3)}\ .$$
These are reexpressed in terms of ${\mathcal D}_{1;1},{\mathcal D}_{1;0},{\mathcal D}_{1;-1}$ by using:
$$\theta_{1,0}=\al_1 \frac{q^{1/2}}{1-q}{\mathcal D}_{1;0}, \ 
\theta_{2,0}=\al_1  \frac{q}{(1-q)^2}[{\mathcal D}_{1;1},{\mathcal D}_{1;-1}],\ 
\theta_{3,0}=\al_1  \frac{q^{3/2}}{(1-q)^3}\big[ [{\mathcal D}_{1;1},{\mathcal D}_{1;0}],{\mathcal D}_{1;-1}\big],$$
as:
\begin{eqnarray*}
{\mathcal P}_1&=&{\mathcal D}_{1;0}\\
{\mathcal P}_2&=&\frac{\al_1}{(1-q)^2(1-t^{-2})(1-q^{-2}t^2)} 
\left(2 [{\mathcal D}_{1;1},{\mathcal D}_{1;-1}]
-\al_1 {\mathcal D}_{1;0}^2 \right)\\
{\mathcal P}_3&=&\frac{\al_1}{(1-q)^3(1-t^{-3})(1-q^{-3}t^3)}
\left(3\big[ [{\mathcal D}_{1;1},{\mathcal D}_{1;0}],{\mathcal D}_{1;-1}\big] -3\al_1{\mathcal D}_{1;0}[{\mathcal D}_{1;1},{\mathcal D}_{1;-1}] +\al_1^2 {\mathcal D}_{1;0}^3 \right)
\end{eqnarray*}
Substituting this into \eqref{initexe} gives the polynomial relations:
\begin{eqnarray*}
{\mathcal D}_{2;0}&=&\frac{t}{(1+t)(q+t)} \left\{(1+q){\mathcal D}_{1;0}^2-\frac{q}{1-q} \,[{\mathcal D}_{1;1},{\mathcal D}_{1;-1}] \right\}\\
{\mathcal D}_{3;0}&=&\frac{t^2}{(1+t)(q+t)(1+t+t^2) (q^2+q t+t^2)}\left\{ (q(q+t^2)+t(1+q+q^2+q^3)){\mathcal D}_{1;0}^3 \right.\\
&& \qquad \left. \!\!\!\!\!\!\!\!\!\!\!\!\!\!\!\!\!\!\!+\frac{q((1+q)(1+t)(q+t)+t(1-q+q^2))}{1-q} 
{\mathcal D}_{1;0}[{\mathcal D}_{1;1},{\mathcal D}_{1;-1}]+\frac{q^2}{(1-q)^2} 
\big[ [{\mathcal D}_{1;1},{\mathcal D}_{1;0}],{\mathcal D}_{1;-1}\big]  \right\} 
\end{eqnarray*}
Finally conjugating $n$ times w.r.t. $\gamma^{-1}$ gives:
\begin{eqnarray*}
{\mathcal D}_{2;n}&=&\frac{t}{(1+t)(q+t)} \left\{(1+q){\mathcal D}_{1;n}^2-\frac{q}{1-q} \,[{\mathcal D}_{1;n+1},{\mathcal D}_{1;n-1}] \right\}\\
{\mathcal D}_{3;n}&=&\frac{t^2}{(1+t)(q+t)(1+t+t^2) (q^2+q t+t^2)}
\left\{ (q(q+t^2)+t(1+q+q^2+q^3)){\mathcal D}_{1;n}^3 \right.\\
&& \left.  \!\!\!\!\!\!\!\!\!\!\!\!\!\!\!\!\!\!\!\!\!\!\!+\frac{q((1+q)(1+t)(q+t)+t(1-q+q^2))}{1-q} {\mathcal D}_{1;n}[{\mathcal D}_{1;n+1},{\mathcal D}_{1;n-1}]+\frac{q^2}{(1-q)^2}
\big[ [{\mathcal D}_{1;n+1},{\mathcal D}_{1;n}],{\mathcal D}_{1;n-1}\big]  \right\} 
\end{eqnarray*}
Note that the expression for ${\mathcal D}_{2;n}$ above is in agreement with \eqref{Mnn}, upon identifying
${\mathcal D}_{2;n}=t^{2-N}{\mathcal M}_{n,n}$ and ${\mathcal D}_{1;n}=t^{\frac{1-N}{2}}{\mathcal M}_{n}$.
Moreover, using also ${\mathcal D}_{3;n}=t^{3(3-N)/2}{\mathcal M}_{n,n,n}$, we get:
\begin{eqnarray*}
{\mathcal M}_{3;n}&=&{\mathcal M}_{n,n,n}=\frac{t^{-1}}{(1+t)(q+t)(1+t+t^2) (q^2+q t+t^2)}
\left\{ (q(q+t^2)+t(1+q+q^2+q^3)){\mathcal M}_{n}^3 \right.\\
&& \left. \!\!\!\!\!\!\!\!\!\!\!\!\!\!\!\!\!\!\!+\frac{q((1+q)(1+t)(q+t)+t(1-q+q^2))}{1-q}
{\mathcal M}_{n}[{\mathcal M}_{n+1},{\mathcal M}_{n-1}]+\frac{q^2}{(1-q)^2}
\big[ [{\mathcal M}_{n+1},{\mathcal M}_{n}],{\mathcal M}_{n-1}\big]  \right\} 
\end{eqnarray*}
\end{example}

\section{The dual Whittaker limit $t\to\infty$}\label{whitaklimsec}

In this section, we describe the $t\to \infty$ limit of the constructions of this paper, in particular the 
degenerations of our generalized Macdonald operators, and of the shuffle product.

\subsection{Quantum $M$-system and quantum determinant}

\subsubsection{Quantum $M$-system}
The generalized Macdonald operators ${\mathcal M}_{\al;n}$ \eqref{genmacdop} were introduced in \cite{DFK15} 
as the natural $t$-deformation of the difference operators:
\begin{equation}\label{oldmac}
M_{\al;n}:=\lim_{t\to \infty} t^{-\al(N-\al)}\, {\mathcal M}_{\al;n}=\sum_{I\subset [1,N]\atop |I|=\al} x_I^n \prod_{i\in I\atop j\not\in I}
\frac{x_i}{x_i-x_j} \, \Gamma_I
\end{equation}

The difference operators $M_{\al;n}$, together with the quantities $\Delta=\Gamma_1\Gamma_2\cdots \Gamma_N$
and $A=x_1x_2\cdots x_N$ satisfy the quantum $M$-system relations, inherited from the (graded) quantum cluster algebra
associated to the $A_{N-1}$ quantum $Q$-system with a coefficient \cite{DFK15}.
These relations read for $\al,\beta =1,2,...,N$:
\begin{eqnarray*}
M_{\al;n}\, M_{\beta;p}&=& q^{\min(\al,\beta)(p-n)}\, M_{\beta;p}\, M_{\al;n}\qquad (|n-p|\leq |\al-\beta|+1)\\
q^\al \, M_{\al;n+1} \, M_{\al;n-1}&=& (M_{\al;n})^2-M_{\al+1;n}\, M_{\al-1;n} \qquad (1\leq \al \leq N)\\
M_{0;n}&=&1 , \qquad M_{N+1;n}=0, \qquad M_{N,n}=A^n\,\Delta \\
\Delta\, M_{\al;n}&=& q^{\al\, n}\, M_{\al;n}\, \Delta, \quad M_{\al;n}\, A= q^\al \, A \, M_{\al;n},\quad 
\Delta\, A= q^N\, A\, \Delta
\end{eqnarray*}

\subsection{Currents}

There is a simple linear relation between the $t$-deformed ${\mathcal M}_{1;n}$'s and the $M_{1;n}$'s.
\begin{thm}\label{DofM}
We have the relation:
\begin{equation}\label{calMofM}
{\mathcal M}_{1;n}=\frac{t^{N}}{t-1} \sum_{j=0}^{N}(-t^{-1})^j 
\,  e_j(x_1,...,x_{N})\, M_{1;n-j}
\end{equation}
where $e_i$ are the elementary symmetric functions.
\end{thm}
\begin{proof}
Evaluating
\begin{equation}\label{evaone}
\prod_{i=1}^{N} (t x-x_i)=\sum_{j=0}^{N} (-1)^j t^{N-j} x^{N-j}  e_j(x_1,...,x_{N})
\end{equation}
at $x=x_1$, we get:
\begin{equation}\label{evatwo}
\prod_{i=2}^{N} (t x_1-x_i)=\frac{t^{N}}{t-1} \sum_{j=0}^{N}(-t^{-1})^j 
x_1^{N-1-j}  e_j(x_1,...,x_{N})
\end{equation}
%Taking $t=1$ and $x=x_1$ in \eqref{evaone}, we also get:
%\begin{equation}\label{evathree}
%0 = \sum_{j=0}^{N}(-1)^j  
%x_1^{N-1-j}  e_j(x_1,...,x_{N})
%\end{equation}
%Combining \eqref{evatwo} and \eqref{evathree} we arrive at:
%$$ \prod_{i=2}^{N} (t x_1-x_i)= 
%
%When $\theta\to 1$, this tends to a finite limit, thanks to the relation $\sum_{j=0}^{r+1}(-z)^j \,
%e_j(X_1,...,X_{r+1})=\prod_{i=1}^{r+1}(1-z X_i)$, which vanishes at $z=X_1^{-1}$.
%Subtracting $\theta^{-r-1}/(\theta-\theta^{-1})$ times this from the above leads to:
%$$\prod_{i=2}^{r+1} (\theta X_1-\theta^{-1} X_i)=\theta^{-r}\sum_{j=0}^{r+1}(-1)^j \frac{\theta^{2(r+1-j)}-1}{\theta^2-1} 
%X_1^{r-j}  e_j(X_1,...,X_{r+1})$$
We deduce that
\begin{equation}
x_1^n\,\prod_{i=2}^{N}  \frac{t x_1-x_i}{x_1-x_i}=\frac{t^{N}}{t-1} \sum_{j=0}^{N}(-t^{-1})^j \,e_j(x_1,...,x_{N})\,
x_1^{n-j} \,\prod_{i=2}^{N} \frac{x_1}{x_1-x_i} ,
\end{equation}
and the lemma follows by substituting this and its images under the interchanges $x_1 \leftrightarrow x_i$
for $i=2,3,...,N$ into the expression \eqref{genmacdop} for $\al=1$.
\end{proof}

In terms of currents, defining $m(z):=\sum_{n\in \Z} z^n\, M_{1;n}$, 
Theorem \ref{DofM} translates into the following relation.

\begin{cor}
The $t$-deformed Macdonald current ${\mathfrak m}(u)$ is expressed in terms of the $t\to\infty$ limit
$m(u)$ as:
\begin{equation}
{\mathfrak m}(u)=\frac{t^{N}}{t-1}\, C(t^{-1}u)\, m(u)
\end{equation}
with $C$ as in \eqref{Cdef}.
\end{cor}
\begin{proof}
Note that by definition $C(t^{-1} u)=\sum_{j=0}^N (-t^{-1} u)^j e_j(x_1,...,x_N)$ and compute
the generating currents for both sides of \eqref{calMofM}.
\end{proof}

With the obvious definition of limiting currents 
$$m_\al(z):=\sum_{n\in \Z} z^n M_{\al;n}=\lim_{t\to\infty}t^{-\al(N-\al)} {\mathfrak m}_\al(z)\ ,$$
it is easy to recover from eq.\eqref{mbo} the following quantum determinant expression for $m_\al(z)$
(see \cite{DFK16}, Theorem 2.10):
\begin{equation}\label{qdetmal} m_\al(z)={\rm CT}_\bu\left(\delta(u_1\cdots u_\al/z) 
\left(\prod_{1\leq i<j\leq \al} 1-q \frac{u_j}{u_i}\right) \prod_{i=1}^\al m(u_i)\right) 
\end{equation}
Note also that the limit of the result of Corollary \ref{ealine} yields the following alternative formula:
$$ m_\al(z)=\frac{1}{\al!}{\rm CT}_\bu\left(\delta(u_1\cdots u_\al/z) 
\prod_{1\leq i<j\leq \al} (u_i^{-1}-u_j^{-1})(u_i-q u_j)\,  \prod_{i=1}^\al m(u_i)\right) $$

Similarly, we may consider the limiting difference operators:
\begin{equation}\label{limitDM}
D_\al(P):=\lim_{t\to\infty} t^{-\al(N-\al)}\, {\mathcal D}_\al(P)= 
M_\al(P):=\lim_{t\to\infty} t^{-\al(N-\al)}\, {\mathcal M}_\al(P)\ ,
\end{equation}
where
\begin{eqnarray}
D_\al(P)&=&\frac{1}{\al!(N-\al)!}{\rm Sym}\left(P(x_1,...,x_\al) 
\prod_{1\leq i\leq \al<j\leq N}\frac{x_i}{x_i-x_j} \Gamma_1\cdots \Gamma_\al \right)\label{limDal}\\
M_\al(P)&=&\frac{1}{\al!}{\rm CT}_\bu \left(P(u_1^{-1},...,u_\al^{-1}) \prod_{1\leq i<j \leq \al}(u_i^{-1}-u_j^{-1})(u_i-q u_j)
\prod_{i=1}^\al m(u_i) \right) \label{limMal}
\end{eqnarray}
as well as 
$$M_{a_1,...,a_\al}:=\lim_{t\to \infty} t^{-\al(N-\al)} {\mathcal M}_{a_1,...,a_\al}=M_\al(s_{a_1,...,a_\al})\ ,$$
which is a polynomial of degree $\al$ in the $M_\ell$'s, as a direct consequence of formula \eqref{limMal} and the fact that
$s_{a_1,...,a_\al}$ is a Laurent polynomial of $x_1,...,x_\al$. 

\subsubsection{Quantum determinants and Alternating Sign Matrices}
The formula of Corollary \ref{ctsimp} also gives the following
alternative ``quantum determinant" expression:
\begin{equation}\label{ctsimplim}
{M}_{a_1,...,a_\al}=CT_\bu\left(\prod_{i=1}^\al u_i^{-a_i}
 \left( \prod_{1\leq i<j\leq \al}1-q \frac{u_j}{u_i}\right)
\, \prod_{i=1}^\al {m}(u_i)\right)=: \left\vert \left( M_{a_j+i-j}\right)_{1\leq i,j\leq \al} \right\vert_q
\end{equation}
or equivalently the generating multi-current expression:
\begin{equation}\label{geneM}
M_\al(v_1,...,v_\al):=\sum_{a_1,...,a_\al\in\Z}{M}_{a_1,...,a_\al} v_1^{a_1}v_2^{a_2}\cdots v_\al^{a_\al}
=  \left(\prod_{1\leq i<j\leq \al}1-q \frac{v_j}{v_i}\right)\, \prod_{i=1}^\al {m}(v_i)
\end{equation}

There is a very nice expression of the quantum determinant \eqref{ctsimplim}, involving a sum 
over Alternating Sign Matrices (ASM). 
This is because the quantity $\prod_{i<j}v_i+\lambda v_j$ is the $\lambda$-determinant $\lambda\!\det( V_n)$
(as defined by Robbins and Rumsey \cite{RR})
of the Vandermonde matrix $V_n:=(v_i^{n-j})_{1\leq i,j\leq n}$.
Recall that an $n\times n$ ASM $A$ has elements $a_{i,j}\in \{0,1,-1\}$
such that each row and column sum is $1$, and the non-zero entries alternate in sign along each row and column.
We denote by $ASM_n$ the set of such matrices.
We need a few more definitions.
The inversion number of an ASM is the quantity $I(A)=\sum_{i>k,j<\ell} A_{i,j}A_{k,\ell}$. Moreover, for any ASM $A$,
the total number of entries equal to $-1$, which we call the $-1$ number, is denoted by $N(A)$. 
Let us also define the column vector $v=(n-1,n-2,..,1,0)^t$,
and for each ASM $A$ we denote by $m_i(A):= (Av)_i$.
Then we have the explicit formula, obtained by taking $\lambda =-q$ for the $\lambda$-determinant of the 
$\al\times \al$ Vandermonde matrix $V_\al$:
$$\prod_{1\leq i<j\leq \al}v_i-q v_j=\sum_{A\in ASM_n} (-q)^{I(A)-N(A)} (1-q)^{N(A)} \prod_{i=1}^n v_i^{m_i(A)} $$
Combining this with \eqref{ctsimplim}, we deduce the following compact expression for the quantum determinant:

\begin{thm}\label{qdethm}
The quantum determinant of the matrix $\left( M_{a_j+i-j}\right)_{1\leq i,j\leq \al}$ reads:
\begin{equation}\label{ASMqdet}
\left\vert \left( M_{a_j+i-j}\right)_{1\leq i,j\leq \al} \right\vert_q=\sum_{A\in ASM_\al}
(-q)^{I(A)-N(A)}(1-q)^{N(A)} \, \prod_{i=1}^\al M_{ a_i+\al-i-m_i(A)}
\end{equation}
\end{thm}

\begin{example}
For $\al=2$, we have two ASMs:
$$\begin{pmatrix} 1 & 0 \\ 0 & 1\end{pmatrix} \qquad {\rm and}\qquad \begin{pmatrix} 0 & 1 \\ 1 & 0\end{pmatrix}$$
with respective inversion and $-1$ numbers $I(A)=0,1$ and $N(A)=0,0$, and with $(m_1(A),m_2(A))=(1,0),(0,1)$. 
The formula \eqref{ASMqdet}  gives:
$$M_{a_1,a_2}=\left\vert \begin{matrix}M_{a_1} & M_{a_2-1}\\ M_{a_1+1} & M_{a_2}\end{matrix}\right\vert_q
:=M_{a_1}M_{a_2}-qM_{a_1+1}M_{a_2-1}$$

For $\al=3$, we have seven ASMs:
$$ \begin{pmatrix} 1 & 0 & 0\\ 0 & 1 & 0\\ 0& 0& 1\end{pmatrix},
\begin{pmatrix} 0 & 1 & 0\\ 1 & 0 & 0\\ 0& 0& 1\end{pmatrix},
\begin{pmatrix} 1 & 0 & 0\\ 0 & 0 & 1\\ 0& 1& 0\end{pmatrix},
\begin{pmatrix} 0 & 0 & 1\\ 0 & 1 & 0\\ 1& 0& 0\end{pmatrix},
\begin{pmatrix} 0 & 0 & 1\\ 1 & 0 & 0\\ 0& 1& 0\end{pmatrix},
\begin{pmatrix} 0 & 1 & 0\\ 0 & 0 & 1\\ 1& 0& 0\end{pmatrix},
\begin{pmatrix} 0 & 1 & 0\\ 1 & -1 & 1\\ 0& 1& 0\end{pmatrix}
$$
with respective inversion and $-1$ numbers $I(A)=0,1,1,3,2,2,2$ and $N(A)=0,0,0,0,0,0,1$,
and $(m_1(A),m_2(A),m_3(A))=(2,1,0),(1,2,0),(2,0,1),(0,1,2),(0,2,1),(1,0,2),(1,1,1)$.
The formula \eqref{ASMqdet} gives:
\begin{eqnarray*}&&\!\!\!\!\!\!\!\!\!\!\! M_{a_1,a_2,a_3}=\left\vert \begin{matrix}
M_{a_1} & M_{a_2-1}& M_{a_3-2}\\ 
M_{a_1+1} & M_{a_2}& M_{a_3-1}\\
M_{a_1+2} & M_{a_2+1}& M_{a_3}
\end{matrix}\right\vert_q
:=M_{a_1}M_{a_2}M_{a_3}-qM_{a_1+1}M_{a_2-1}M_{a_3}-qM_{a_1}M_{a_2+1}M_{a_3-1}\\
&&\!\!\!\!\!\!\!\!\!\!\!  -q^3 M_{a_1+2}M_{a_2}M_{a_3-2}+ q^2 M_{a_1+2}M_{a_2-1}M_{a_3-1}
+q^2 M_{a_1+1}M_{a_2+1}M_{a_3-2}-q(1-q)M_{a_1+1}M_{a_2}M_{a_3-1}
\end{eqnarray*}
\end{example}

\subsection{The $t\to \infty$ limit of the shuffle product and the quantum M-system}

It is instructive to consider the limiting definition $\star$ of the shuffle product $*$, 
compatible with the limiting difference operators \eqref{limitDM}.

Define for $(P,P')\in {\mathcal F}_\al\times {\mathcal F}_\beta$ the product:
$$P\star P'(x_1,...,x_{\al+\beta}):=\frac{1}{\al!\, \beta!}{\rm Sym}\left( 
\frac{P(x_1,...,x_\al)P'(x_{\al+1},...,x_{\al+\beta})}{\prod_{1\leq i\leq \al<j\leq \al+\beta} (x_j^{-1}-x_i^{-1})(x_j-q x_i)}\right)$$
It is given by the limit 
$$P\star P'(x_1,...,x_{\al+\beta})=\lim_{t\to\infty} t^{-2\al\beta}\, P*P'(x_1,...,x_{\al+\beta})$$
The compatibility:
$$D_\al(P)D_\beta(P')=D_{\al+\beta}(P\star P')$$
simply follows by computing the limit
$$\lim_{t\to\infty} t^{-\al(N-\al)}\, {\mathcal D}_\al(P)\, t^{-\beta(N-\beta)}\,{\mathcal D}_\beta(P')
=\lim_{t\to\infty} t^{-(\al+\beta)(N-\al-\beta)}\,{\mathcal D}_{\al+\beta}(t^{-2\al\beta}\,P*P')$$.

Recall that $M_{\al;n}=D_\al((x_1x_2\cdots x_\al)^n)$. The quantum $M$-system relations
boil down to relations in the $\star$ shuffle algebra, namely:

\begin{thm}\label{msyshuf}
We have the following relations:
\begin{eqnarray*}
&&(x_1x_2\cdots x_\al)^n\star (x_1x_2\cdots x_\beta)^p\\
&&\qquad\qquad=q^{Min(\al,\beta)(p-n)}
(x_1x_2\cdots x_\beta)^p\star (x_1x_2\cdots x_\al)^n\quad (n,p\in \Z,\ |p-n|\leq |\al-\beta|+1)\\
&& q^\al (x_1x_2\cdots x_\al)^{n+1}\star (x_1x_2\cdots x_\al)^{n-1}\\
&&\qquad\qquad=
(x_1x_2\cdots x_\al)^n\star (x_1x_2\cdots x_\al)^n -
(x_1x_2\cdots x_{\al+1})^n\star (x_1x_2\cdots x_{\al-1})^n\quad (n\in \Z)
\end{eqnarray*}
which hold respectively in ${\mathcal F}_{\al+\beta}$ and ${\mathcal F}_{2\al}$, with $\al,\beta\in [1,N-1]$.
\end{thm}

\begin{example}
For $\al=\beta=1$, we have in ${\mathcal F}_{2}$:
$$x_1^n\star x_1^{n+1}=q\,x_1^{n+1}\star x_1^n,\quad (x_1x_2)^n=x_1^n\star x_1^n-q \,x_1^{n+1}\star x_1^{n-1}$$
These follow from respectively
\begin{eqnarray*}
{\rm Sym}\left( \frac{(x_1x_2)^n}{(x_2^{-1}-x_1^{-1})} \right)=
(x_1x_2)^{n+1}\,{\rm Sym}\left( \frac{1}{(x_1-x_2)} \right)&=& 0\\
{\rm Sym}\left( \frac{x_1(x_1x_2)^{n-1}}{(x_2^{-1}-x_1^{-1})} \right)=
(x_1x_2)^{n}\,{\rm Sym}\left( \frac{x_1}{(x_1-x_2)} \right)&=& (x_1x_2)^{n}
\end{eqnarray*}
where ${\rm Sym}$ denotes the symmetrization in $x_1,x_2$. The ``higher" identities:
$${\rm Sym}\left( \frac{x_1^{k+1}(x_1x_2)^{n-1}}{(x_2^{-1}-x_1^{-1})} \right)=
(x_1x_2)^{n}\,{\rm Sym}\left( \frac{x_1^{k+1}}{(x_1-x_2)} \right)=(x_1x_2)^{n}\,s_{k,0}(x_1,x_2)=s_{n+k,n}(x_1,x_2)$$
amount to:
$$ s_{n+k,n}(x_1,x_2)=x_1^{k+n}\star x_1^{n}-q \,x_1^{k+n+1}\star x_1^{n-1}$$
which amounts to the relation
$$M_{n+k,n}=M_{n+k}\, M_n-q \, M_{n+k+1}\, M_{n-1}=
\left\vert \begin{matrix}M_{n+k} & M_{n-1}\\ M_{n+k+1} & M_{n}\end{matrix}\right\vert_q$$
This is nothing but the $v_1^{n+k}v_2^n$ coefficient of $M_2(v_1,v_2)$ \eqref{geneM}.
\end{example}

\section{Discussion/Conclusion}\label{concsec}
\subsection{$(q,t)$-quantum determinant}

One of the goals of this paper can be understood as the construction of a $t$-deformation of the quantum determinant
of \cite{DFK15} defined in the context of the quantum $Q$-system for $A_{N-1}$. 
Concretely, we conjectured that the quantity
${\mathcal M}_{a_1,a_2,...,a_\al}$ 
is a polynomial of the ${\mathcal M}_n$'s, and proved the conjecture in the case $\al=2$ (Theorem \ref{polynomialitythm}),
and in the case when $a_1=a_2=\cdots =a_\al=n$ (Theorem \ref{polpol}, and Corollary \ref{polpolcor}). 
The latter case uses explicitly the relations in the Elliptic Hall Algebra.
We would like to interpret this polynomial as the $(q,t)$-determinant of the matrix
$\left( {\mathcal M}_{a_j+i-j}\right)_{1\leq i,j\leq \al}$.

The main difference with the usual and the quantum determinant is that, in both of the cases above, the polynomial expression for ${\mathcal M}_{a_1,a_2,...,a_\al}$ depends, in general, on more than just the
matrix elements $\{{\mathcal M}_{a_j+i-j}$, $1\leq i, j \leq \al\}$ (on which it depends in the quantum determinant case). This polynomial is unique modulo the relations of
the quantum toroidal algebra, namely the exchange relation \eqref{exchange} and the Serre relations \eqref{serre1}, 
expressed respectively as quadratic and cubic relations between the ${\mathcal M}_n$'s. 
In that respect, the expressions \eqref{theone} of Theorem \ref{oddMthm} for $k=0$,
 \eqref{Mnn} derived from Lemma \ref{evenMlemma}, and the alternative expression \eqref{alterM}
 for $a=n+2$, $b=n$, all have the property to be polynomials of just the matrix elements 
 $\{{\mathcal M}_{a_j+i-j}$, $1\leq i, j \leq \al\}$, as we may write for them:
\begin{eqnarray*}
{\mathcal M}_{n+1,n}&=&\left\vert \begin{matrix} {\mathcal M}_n & {\mathcal M}_{n} \\
{\mathcal M}_{n+1} & {\mathcal M}_{n+1}
\end{matrix}\right\vert_{q,t}=\frac{1}{(1-q)t} \left[ {\mathcal M}_n,  {\mathcal M}_{n+1}\right]_q\\
{\mathcal M}_{n,n}&=&\left\vert \begin{matrix} {\mathcal M}_n & {\mathcal M}_{n-1} \\
{\mathcal M}_{n+1} & {\mathcal M}_{n}
\end{matrix}\right\vert_{q,t}= \frac{1}{(1-q)(1+t)(q+t)} 
\left( (q^{-1}-q){\mathcal M}_n^2 +\left[{\mathcal M}_{n-1},{\mathcal M}_{n+1}\right]\right) \\
{\mathcal M}_{n+2,n}&=&\left\vert \begin{matrix} {\mathcal M}_{n+2} & {\mathcal M}_{n-1} \\
{\mathcal M}_{n+3} & {\mathcal M}_{n}
\end{matrix}\right\vert_{q,t}=\frac{1}{(q-1)(1-t^2)(q^2-t^2)} \times \\
&&\times
\left( t(1+q)\left[{\mathcal M}_{n},{\mathcal M}_{n+2}\right]_{q^2}-(q+t^2)\left[{\mathcal M}_{n+2},{\mathcal M}_{n}\right]_{q^2}-q(q+t^2)\left[{\mathcal M}_{n-1},{\mathcal M}_{n+3}\right]\right)
\end{eqnarray*}
The property still holds for ${\mathcal M}_{n+3,n}$, but breaks down for ${\mathcal M}_{n+4,n}$ which can at best be expressed as a polynomial
of ${\mathcal M}_{n+4},{\mathcal M}_{n},{\mathcal M}_{n+5},{\mathcal M}_{n-1},{\mathcal M}_{n+3},{\mathcal M}_{n+1}$.

The $(q,t)$-determinant is therefore a subtle deformation of the  quantum determinant. However, it is possible that
one of the expressions for this quantity has a nice combinatorial expression, generalizing eq. \eqref{ASMqdet}
of Theorem \ref{qdethm}.

\subsection{EHA as $t$-deformed quantum cluster algebra}

In this paper, we have constructed a representation of the EHA for finitely many variables $x_1,x_2,...,x_N$.
The algebra itself admits a quotient by the ideal generated by the relations ${\mathcal M}_{N+1;n}=0$, $n\in \Z$,
and the relations expressing that $\psi^{\pm}\propto \sum_{k>0} u_{0,\pm k} z^{\pm k}$ 
are series expansions of finite products. In particular, the condition ${\mathcal M}_{N+1;n}=0$
expresses that a degree $N+1$ polynomial of ${\mathcal M}_n,{\mathcal M}_{n\pm 1}$ 
must vanish (by use of Theorem \ref{polpol} and Corollary \ref{polpolcor}).

This allows to view the EHA as a natural $t$-deformation of the quantum $A_{N-1}$ M-system algebra, 
which corresponds to the limit $t\to\infty$.
We expect the defining relations of the M-system to be explicitly $t$-deformed.
For instance, the first q-commutation relation:
\begin{equation}\label{Msysone}M_nM_{n+1}-qM_{n+1}M_n=0
\end{equation}
is $t$-deformed into:
$$ {\mathcal M}_n \,{\mathcal M}_{n+1}-q\, {\mathcal M}_{n+1}\, {\mathcal M}_n=t(1-q){\mathcal M}_{n+1,n} \ ,$$
obtained from eq.~\eqref{theone} of Theorem \ref{oddMthm} for $k=0$.
Indeed, as ${\mathcal M}_n\sim t^{N-1}$ and ${\mathcal M}_{n,p}\sim t^{2(N-2)}$ for large $t$, we write:
$$ t^{1-N}{\mathcal M}_n \,t^{1-N} {\mathcal M}_{n+1}-q\,t^{1-N} {\mathcal M}_{n+1}\, t^{1-N}{\mathcal M}_n=
t^{-1}(1-q)t^{2(N-2)}{\mathcal M}_{n+1,n} $$
which shows that the r.h.s. is subleading at $t\to \infty$, and we recover the M-system relation \eqref{Msysone} in this limit.

Similarly, the M-system relation 
\begin{equation}\label{Msystwo}
M_{2;n}=M_n^2-q\, M_{n+1}M_{n-1}
\end{equation}
is deformed into \eqref{Mevenone},
namely:
$${\mathcal M}_n^2-q\, {\mathcal M}_{n+1}{\mathcal M}_{n-1}=(q+t+t^2){\mathcal M}_{2;n}-q t {\mathcal M}_{n+1,n-1}.$$
Repeating the scaling analysis, we get
$$ t^{2(1-N)}{\mathcal M}_n^2-q\,  t^{1-N}{\mathcal M}_{n+1} t^{1-N}{\mathcal M}_{n-1}=(1+t^{-1}+q t^{-2}) t^{2(2-N)}{\mathcal M}_{2;n}-q t^{-1}  t^{2(2-N)}{\mathcal M}_{n+1,n-1}$$
and we recover \eqref{Msystwo} in the $t\to\infty$ limit, by neglecting the subleading terms $O(t^{-1}),O(t^{-2})$.
%Note that the dual $Q$-system relation obtained by conjugating \eqref{Msystwo} by $M_{n-1}$:
%\begin{equation}\label{Msysdutwo}
%M_{2;n}=q M_n^2-M_{n-1}M_{n+1}
%\end{equation}
%is $t$-deformed into \eqref{Meventwo}, namely:
%$$q({\mathcal M}_n)^2 -{\mathcal M}_{n-1}{\mathcal M}_{n+1}
%=(q+qt+t^2)\,{\mathcal M}_{2;n}- t \,{\mathcal M}_{n+1,n-1} $$
%and the same scaling analysis as before shows that this tends to \eqref{Msysdutwo} as $t\to \infty$.

More generally, the quantum M-system relations were shown in \cite{DFK16} to be solved by the quantum determinant
\eqref{qdetmal} expressions for $M_{\al;n}$ as polynomials of $\{M_{1;n+i}\}_{|i|<\al}$, with the 
condition that $M_{N+1;n}=0$ for all $n\in \Z$. More generally, we may consider the quantum determinant expression \eqref{ctsimplim}
for $M_{a_1,a_2,...,a_\al}$. Let us rewrite the $(q,t)$ relation of Theorem \ref{Mofm} in the following manner:
\begin{equation}\label{tdefor}
\prod_{1\leq i<j\leq \al}  \left(1-t^{-1}\frac{v_i}{v_j}\right)\left(1-t^{-1}q \frac{v_j}{v_i}\right) \, t^{-\al(N-\al)}{\mathfrak M}_{\al}(\bv)
=\prod_{1\leq i<j\leq \al}\left(1-q \frac{v_j}{v_i}\right) \, \prod_{i=1}^\al t^{-(N-1)} {\mathfrak m}(v_i)
\end{equation}
where $\lim_{t\to\infty}t^{-\al(N-\al)}{\mathfrak M}_{\al}(\bv)=M_\al(\bv)$, and accordingly
$\lim_{t\to\infty}t^{-(N-1)}{\mathfrak m}(v)=m(v)$.
The r.h.s. of \eqref{tdefor} is simply the generating function for the quantum determinants 
$\vert (t^{1-N}{\mathcal M}_{a_j+j-i})_{1\leq i,j\leq \al}\vert_q$ which tend to 
$\vert ({M}_{a_j+j-i})_{1\leq i,j\leq \al}\vert_q$ when $t\to \infty$.
Expanding the l.h.s. in powers of $t^{-1}$ at large $t$, we see that the dominant term
is $t^{-\al(N-\al)}{\mathfrak M}_{\al}(\bv)$ and all other terms are of stricly smaller order. 
This displays explicitly in which sense this
relation is a $t$-deformation of the quantum determinant relation \eqref{ctsimplim}.

\begin{example}
For $\al=2$, \eqref{tdefor} gives in components:
$$t^{2(2-N)}\left((1+qt^{-2}){\mathcal M}_{a_1,a_2} -t^{-1}{\mathcal M}_{a_1-1,a_2+1}-t^{-1}q{\mathcal M}_{a_1+1,a_2-1}\right)
=t^{2(1-N)}\left( {\mathcal M}_{a_1}{\mathcal M}_{a_2}-q {\mathcal M}_{a_1+1}{\mathcal M}_{a_2-1}\right).$$
\end{example}

%More generally, the relations of the quantum $Q$-system are quantum generalizations of the Desnanot-Jacobi identity
%that relates the minors of any square matrix, when applied to the matrix $({M}_{n+j-i})_{1\leq i,j\leq \al}$.
%We expect the $(q,t)$-determinants of the matrices $({\mathcal M}_{a_j+j-i})_{1\leq i,j\leq \al}$ introduced above to satisfy more involved relations, that would be the natural $t$-deformations of the quantum $Q$-system relations.

%We also expect other mutations of the $Q$-system quantum cluster algebra to give rise to new operators, 
%and it would be interesting to investigate their $t$-deformations as well.

\subsection{Relation to graded characters}

The difference operators $M_{\al;n}=\lim_{t\to\infty} t^{\al(\al-N)} {\mathcal M}_{\al;n}$ 
were introduced in \cite{DFK15} to generate graded characters of tensor products of Kirillov-Reshetikhin modules
by iterated action on the constant function $1$.
In particular, any expression of the form $\prod_{i=k}^1 \prod_{\al=1}^{N-1} (M_{\al;i})^{n_{\al,i}} \cdot 1$
 for $n_{\al,i}\in \Z_+$
is Schur positive, namely decomposes onto Schur functions with graded multiplicities in $\Z_+[q]$.
This is not the case for the $t$-deformed version. As an example, it is easy to see that
${\mathcal M}_2\cdot 1=t^{N-1} s_2 -t^{N-2} s_{1,1}$, so Schur positivity is lost. It would be interesting to understand the geometric or representation-theoretical meaning of this $t$-deformation of the $q$-graded characters.
%However, the $t$-deformed expressions $\prod_{i=k}^1 \prod_{\al=1}^{N-1} ({\mathcal M}_{\al;i})^{n_{\al,i}} \cdot 1$
%have a Schur decomposition with coefficients in $\Z[t,q]$, which might still have combinatorial interpretations.

\subsection{Possible generalizations}

The objects and structures of this paper are all linked to the type $A$ groups/algebras. However, many
of them can be extended to other types: on the one hand, DAHA was defined for other types \cite{Cheredbook};
on the other hand (quantum) cluster algebras and $Q$-sytems have been defined for other types as well. 
Preliminary results indicate that similar constructions to those of this paper should exist for the other types. 

Another interesting direction, even in the $A_{N-1}$ case, is to try to understand the meaning of the other
cluster variables (not of the form $M_{\al;n}$) in the quantum cluster algebra of the $Q$-system. 
Preliminary explorations show that those
other variables are also difference operators. Understanding these
could be a step in the direction of fully comprehending the $t$-deformation of the quantum cluster algebra.

\bibliographystyle{alpha}

\bibliography{refs}

\end{document}